\documentclass[acmsmall,screen]{acmart}

\usepackage[T1]{fontenc} 
\usepackage[utf8]{inputenc} 
\usepackage{multirow}

\usepackage{graphicx}

\usepackage{url}
\usepackage{verbatim}
\usepackage{color}
\usepackage{amsmath}
\usepackage{latexsym}
\usepackage{stmaryrd}
\usepackage{listings}
\usepackage{enumerate}
\usepackage{xspace}
\usepackage{multicol}
\usepackage{proof}
\usepackage{mathpartir}
\usepackage{multirow}
\usepackage{hhline}
\usepackage{array}
\usepackage{balance}
\definecolor{light-gray}{gray}{0.80}
\usepackage{algorithm}
\usepackage{algpseudocode}
\usepackage{enumitem}

\usepackage{qtree}
\usepackage{comment}
\usepackage{semantic}
\usepackage{galois}
\usepackage{pgffor}
\usepackage{float}

\usepackage{bm}
\definecolor{listinggray}{gray}{0.9}
\usepackage{layout}
\usepackage{caption}
\usepackage{mathtools}

\usepackage{svgcolor}
\definecolor{darkblue}{rgb}{0,0,0.35}

\usepackage{hyperref}

\newcommand{\vtag}{\mathit{tag}}

\newcommand{\restr}[2]{#1|_{#2}}

\newcommand{\vars}{\mi{Var}}

\newcommand{\locs}{\mi{Loc}}

\newcommand{\ZZ}{\mathbb{Z}}
\newcommand{\NN}{\mathbb{N}}

\newcommand{\pto}{\rightharpoonup}

\newtheorem{zproof*}{Proof}

\newif\iftr
\trtrue
\newcommand{\less}[1]{\iftr \else #1 \fi}
\newcommand{\more}[1]{\iftr {#1} \else \fi}
\newcommand{\lessmore}[2]{\iftr #2 \else #1 \fi}

\newcommand{\m}[1]{\mathsf{#1}}
\newcommand{\mi}[1]{\mathit{#1}}

\newcommand{\kw}[1]{\textsf{\textbf{#1}}} 

\newcommand{\massert}[1]{\kw{assert}(#1)}

\newcommand{\mbind}[2]{\kw{bind}(#1, #2)}
\newcommand{\mread}[1]{\,!#1}
\newcommand{\mto}{\leftarrow}
\newcommand{\menv}{\kw{env}}
\newcommand{\mupd}[2]{#1 \mathrel{:=} #2}
\newcommand{\Mlet}{\kw{let}\,}
\newcommand{\Mdo}{\kw{do}\,}
\newcommand{\Mreturn}{\kw{return}\,}
\newcommand{\Min}{\kw{in}\,}
\newcommand{\Mif}{\kw{if}\,}
\newcommand{\Melse}{\kw{else}\,}

\newcommand{\Mthen}{\kw{then}\,}

\newcommand{\Mend}{\kw{;}\;}
\newcommand{\Mfor}{\kw{for}\;}

\newcommand{\bdiamond}{\mathbin{\diamond}}

\newcommand{\stacklabel}[1]{\stackrel{\smash{\scriptscriptstyle \mathrm{#1}}}}
\newcommand{\Def}{\stacklabel{\mathrm{def}}}

\newcommand{\tool}[1]{\textsc{#1}\xspace}
\newcommand{\ourtool}{\tool{Drift}}

\newcommand{\dom}{\operatorname{\mathsf{dom}}}

\newcommand{\set}[1]{\{#1\}}
\newcommand{\pset}[2]{\set{\,#1\mid#2\,}}

\newcommand{\sqleq}{\sqsubseteq}

\newcommand{\Zrepeat}[2]{\foreach \n in {1, ..., #1}{#2}}
\newcommand{\setc}[2]{\{#1 \mid #2\}}
\newcommand{\pwset}[1]{\mathcal{}}
\newcommand{\powerset}{\mathcal{\wp}}

\newcommand{\pair}[2]{\langle #1, #2 \rangle}
\newcommand{\vecop}[1]{{\dot #1}}
\newcommand{\dvecop}[1]{{\ddot #1}}

\newcommand{\ind}[1][1]{\quad\Zrepeat{#1}{\quad}}
\newcommand{\sind}[1][1]{\;\Zrepeat{#1}{\;}}

\newcommand{\elabels}{Loc}
\newcommand{\elabel}{\ell}
\newcommand{\elabela}{i}
\newcommand{\elabelb}{j}

\newcommand{\Exp}{\mi{Exp}}
\newcommand{\const}{c}
\newcommand{\consts}{\mi{Cons}}

\newcommand{\btrue}{\mi{true}}

\newcommand{\ifte}[3]{#1\,?\,#2:#3}

\newcommand{\cn}[2]{#1{\bdiamond}#2}
\newcommand{\cns}[3]{#1{\bdiamond}#2{\bdiamond}#3}
\newcommand{\cvn}[3]{#1{\bdiamond}#2{\bdiamond}#3}

\newcommand{\nenv}{env}
\newcommand{\nloc}{loc}
\newcommand{\nstack}{stack}

\newcommand{\cerr}{\top}
\newcommand{\cval}{v}
\newcommand{\cvalues}{\mathcal{V}}
\newcommand{\ctable}{\bm{\cval}}
\newcommand{\ctables}{\mathcal{T}}
\newcommand{\tin}[1]{#1_{\mathsf{in}}}
\newcommand{\tout}[1]{#1_{\mathsf{out}}}
\newcommand{\safe}{\mathit{safe}}

\newcommand{\cmap}{M}
\newcommand{\cmaps}{\mathcal{M}}

\newcommand{\cenvironments}{\mathcal{E}}
\newcommand{\cenv}{E}
\newcommand{\cnodes}{\mathcal{N}}
\newcommand{\cnode}{n}
\newcommand{\cenodes}{\mathcal{N}_e}

\newcommand{\cvnodes}{\mathcal{N}_{x}}

\newcommand{\ccallsites}{\cstacks}
\newcommand{\ccsite}{\cstack_{\m{cs}}}
\newcommand{\cstack}{S}
\newcommand{\cstacks}{\mathcal{S}}

\newcommand{\cvord}{\sqsubseteq}

\newcommand{\cvjoin}{\sqcup}

\newcommand{\cmord}{\mathrel{\vecop{\cvord}}}

\newcommand{\cmjoinb}{\vecop{\bigsqcup}}

\newcommand{\cprop}{\mathsf{prop}}
\newcommand{\ciprop}{\ltimes}
\newcommand{\cstepbody}{\mathsf{body}}
\newcommand{\fuel}{k}
\newcommand{\fuels}{\mathbb{N}}

\newcommand{\ctransname}{\m{step}}
\newcommand{\Ctransname}{\m{Step}}
\newcommand{\ctrans}[2][k+1]{\ctransname\llbracket#2\rrbracket}
\newcommand{\ccoll}{\mathbf{C}}
\newcommand{\csem}{\mathbf{S}}
\newcommand{\cdomain}{\mathcal{P}}

\newcommand{\psafe}{P_{\m{safe}}}

\newcommand{\tentry}[3]{#1 \triangleleft #2\to#3}
\newcommand{\rdesignation}{\mathsf{r}}

\newcommand{\rscope}{X}

\newcommand{\rbot}{\bot^{\rdesignation}}
\newcommand{\rerr}{\top^{\rdesignation}}
\newcommand{\rval}{u}
\newcommand{\rvalues}[1][\rscope]{\cvalues^{\rdesignation}_{#1}}
\newcommand{\rtable}{\bm{\rval}}
\newcommand{\rtables}[1][\rscope]{\ctables^{\rdesignation}_{#1}}

\newcommand{\scmap}{\Gamma}

\newcommand{\relname}{r}
\newcommand{\relsname}[1][\rscope]{\mathcal{R}_{#1}}
\newcommand{\rel}{\relname}
\newcommand{\rels}[1][\rscope]{\relsname^{\rdesignation}}

\newcommand{\dmapname}{d}
\newcommand{\dmapsname}{\mathcal{D}}
\newcommand{\dmap}{\dmapname}
\newcommand{\dmaps}{\dmapsname^{\rdesignation}}
\newcommand{\rdf}{\mathsf{F}}

\newcommand{\rmap}{\cmap^{\rdesignation}}
\newcommand{\rmaps}{\cmaps^{\rdesignation}}

\newcommand{\rvord}{\sqleq^{\rdesignation}}
\newcommand{\rvmeet}{\sqcap^{\rdesignation}}
\newcommand{\rvjoin}{\sqcup^{\rdesignation}}
\newcommand{\rvmeetb}{\mathsf{\bigsqcap}^{\rdesignation}}

\newcommand{\rvecord}{\vecop{\sqleq}^{\rdesignation}}

\newcommand{\rmord}{\mathrel{\vecop{\sqleq^{\rdesignation}}}}
\newcommand{\rmmeet}{\mathop{\vecop{\sqcap^{\rdesignation}}}}
\newcommand{\rmjoin}{\mathop{\vecop{\sqcup^{\rdesignation}}}}
\newcommand{\rmjoinb}{\vecop{\bigsqcup^{\rdesignation}}}
\newcommand{\rmmeetb}{\vecop{\bigsqcap^{\rdesignation}}}

\newcommand{\dgamma}{\gamma^{\m{d}}}
\newcommand{\rgamma}[1][\rscope]{\gamma^{\rdesignation}_{#1}}
\newcommand{\ralpha}[1][\rscope]{\alpha^{\rdesignation}_{#1}}
\newcommand{\rmgamma}{\vecop{\gamma^{\rdesignation}}}
\newcommand{\rmalpha}{\vecop{\alpha^{\rdesignation}}}

\newcommand{\riprop}{\ltimes^{\rdesignation}}
\newcommand{\rprop}{\cprop^{\rdesignation}}
\newcommand{\rupdate}[3]{#1\rupd{#2}{#3}}
\newcommand{\rupd}[2]{[#1 \leftarrow #2]}
\newcommand{\rrescope}[2]{\rstrengthen{#1}{#2}}
\newcommand{\rstrengthen}[2]{#1[#2]}
\newcommand{\requality}[3]{\rstrengthen{#1}{#2 {=} #3}}

\newcommand{\rconst}{[\nu\!=\!c]^{\rdesignation}}

\newcommand{\rstepbody}{\mathsf{body}^{\rdesignation}}
\newcommand{\rtransname}{\ctransname^{\rdesignation}}
\newcommand{\Rtransname}{\Ctransname^{\rdesignation}}
\newcommand{\rtrans}[2][k+1]{\rtransname\llbracket #2 \rrbracket}
\newcommand{\Rtrans}[2][k+1]{\Rtransname\llbracket #2 \rrbracket}
\newcommand{\rcoll}{\mathbf{C}^{\rdesignation}}

\newcommand{\rdomain}{\cdomain^{\rdesignation}}
\newcommand{\rdomelem}{P^{\rdesignation}}
\newcommand{\rdgamma}{{\dvecop{\gamma}}^{\rdesignation}}
\newcommand{\rdalpha}{{\dvecop{\alpha}}^{\rdesignation}}
\newcommand{\rdord}{\mathrel{\dvecop{\sqleq^{\rdesignation}}}}

\newcommand{\penvironments}{\hat{\cenvironments}}
\newcommand{\penv}{\hat{\cenv}}
\newcommand{\pnodes}{\hat{\cnodes}}
\newcommand{\pnode}{\hat{\cnode}}
\newcommand{\pvnodes}{\hat{\cnodes}_{x}}
\newcommand{\penodes}{\hat{\cnodes}_{e}}

\newcommand{\pstack}{\hat{\cstack}}
\newcommand{\pstacks}{\hat{\cstacks}}
\newcommand{\pconcat}{\mathop{\hat{\cdot}}}

\newcommand{\pcsite}{\pstack_{\m{cs}}}

\newcommand{\absn}{x}

\newcommand{\pdesignation}{\hat{\mathsf{r}}}

\newcommand{\pscope}{{\rscope}}

\newcommand{\perr}{\top^{\pdesignation}}
\newcommand{\pbot}{\bot^{\pdesignation}}
\newcommand{\pval}{\hat{\rval}}
\newcommand{\pvalues}[1][\pscope]{\cvalues^{\pdesignation}_{#1}}
\newcommand{\ptable}{\bm{\pval}}
\newcommand{\ptables}[1][\pscope]{\ctables^{\pdesignation}_{#1}}
\newcommand{\ptab}[3][\absn]{#1\!:#2\to#3}

\newcommand{\prel}{\relname^{\pdesignation}}
\newcommand{\prels}[1][\pscope]{\relsname[#1]^{\pdesignation}}

\newcommand{\pdmap}{\dmapname^{\pdesignation}}

\newcommand{\pmap}{\cmap^{\pdesignation}}
\newcommand{\pmaps}{\cmaps^{\pdesignation}}

\newcommand{\pvord}[1][\pscope]{\sqleq^{\pdesignation}_{#1}}
\newcommand{\pvecord}{\mathrel{\vecop{\sqleq}^{\pdesignation}}}
\newcommand{\pvmeet}[1][\pscope]{\sqcap^{\pdesignation}_{#1}}
\newcommand{\pvjoin}[1][\pscope]{\sqcup^{\pdesignation}_{#1}}

\newcommand{\pmord}{\vecop{\sqleq^{\pdesignation}}}
\newcommand{\pmmeet}{\vecop{\sqcap^{\pdesignation}}}
\newcommand{\pmjoin}{\vecop{\sqcup^{\pdesignation}}}

\newcommand{\pgammac}{\gamma^{\pdesignation}}
\newcommand{\palphac}{\alpha^{\pdesignation}}
\newcommand{\pgamma}[1][\pscope]{\gamma^{\pdesignation}_{#1}}
\newcommand{\palpha}[1][\pscope]{\alpha^{\pdesignation}_{#1}}
\newcommand{\pmgamma}{\vecop{\gamma^{\pdesignation}}}
\newcommand{\pmalpha}{\vecop{\alpha^{\pdesignation}}}

\newcommand{\piprop}{\ltimes^{\pdesignation}}
\newcommand{\pprop}{\cprop^{\pdesignation}}

\newcommand{\pstrengthen}[2]{#1[#2]}

\newcommand{\ptransname}{\ctransname^{\pdesignation}}

\newcommand{\ptrans}[2][k+1]{\ctransname^{\pdesignation}\llbracket#2\rrbracket}

\newcommand{\pcoll}[1][\cdot]{\mathbf{C}^{\pdesignation}\llbracket#1\rrbracket}

\newcommand{\pdomain}{\mathcal{P}^{\pdesignation}}

\newcommand{\tdesignation}{\mathsf{t}}

\newcommand{\tbot}{\bot^{\tdesignation}}
\newcommand{\terr}{\top^{\tdesignation}}
\newcommand{\ttop}{\terr}
\newcommand{\tval}{t}
\newcommand{\tvalues}[1][\pscope]{\cvalues^{\tdesignation}_{#1}}
\newcommand{\ttable}{\bm{\tval}}
\newcommand{\ttables}[1][\pscope]{\ctables^{\tdesignation}_{#1}}
\newcommand{\ttab}[3][\absn]{#1\!:#2 \to #3}

\newcommand{\bdesignation}{b}
\newcommand{\bdomain}[1][\tdesignation]{\mathcal{R}^{#1}}
\newcommand{\bval}{\mathit{b}}
\newcommand{\lia}{\mathsf{lia}}
\newcommand{\bbot}{\bot^{\bdesignation}}
\newcommand{\btop}{\top^{\bdesignation}}
\newcommand{\bord}{\mathrel{\sqleq^{\bdesignation}}}
\newcommand{\bordinv}{\sqsupseteq^{\bdesignation}}
\newcommand{\bmeet}{\mathop{\mathsf{\sqcap}^{\bdesignation}}}

\newcommand{\bjoin}{\mathop{\mathsf{\sqcup}^{\bdesignation}}}
\newcommand{\balpha}{\alpha^{\bdesignation}}
\newcommand{\bgamma}{\gamma^{\bdesignation}}
\newcommand{\bwid}[1][\pscope]{\mathop{\triangledown^{\tdesignation}_{#1}}}

\newcommand{\basetype}{\mathsf{b}}

\newcommand{\paramelem}{\mi{a}}

\newcommand{\paramdesignation}{\mathsf{a}}

\newcommand{\paramjoin}[1][\pscope]{\mathsf{\sqcup}^{\paramdesignation}_{#1}}
\newcommand{\paramalpha}[1][\pscope]{\alpha^{\paramdesignation}}
\newcommand{\paramgamma}[1][\pscope]{\gamma^{\paramdesignation}}

\newcommand{\tyrel}[2][\basetype]{\{\nu\!:#1\;|\;#2\}}
\newcommand{\tyrelv}[1][\pscope]{\relname^{\tdesignation}}
\newcommand{\tyrels}[1][\pscope]{\relsname[#1]^{\tdesignation}}

\newcommand{\tmap}{\cmap^{\tdesignation}}
\newcommand{\tmaps}{\cmaps^{\tdesignation}}

\newcommand{\tvord}[1][\pscope]{\sqleq^{\tdesignation}_{#1}}
\newcommand{\tvecord}{\vecop{\sqleq}^{\tdesignation}}
\newcommand{\tvmeet}[1][\pscope]{\sqcap^{\tdesignation}_{#1}}
\newcommand{\tvjoin}[1][\pscope]{\sqcup^{\tdesignation}_{#1}}
\newcommand{\tvmeetb}{\bigsqcap^{\tdesignation}}
\newcommand{\tvjoinb}{\bigsqcup^{\tdesignation}}

\newcommand{\tmord}{\mathrel{\dot{\sqleq}^{\tdesignation}}}
\newcommand{\tmmeet}{\dot{\sqcap}^{\tdesignation}}
\newcommand{\tmjoin}{\mathop{\dot{\sqcup}^{\tdesignation}}}

\newcommand{\tvgamma}[1][\pscope]{\gamma^{\tdesignation}_{#1}}

\newcommand{\tiprop}{\ltimes^{\tdesignation}}
\newcommand{\tprop}{\cprop^{\tdesignation}}
\newcommand{\tupdate}[3]{#1\tupd{#2}{#3}}
\newcommand{\tupd}[2]{[#1 \leftarrow #2]}

\newcommand{\tstrengthen}[2]{#1[#2]}
\newcommand{\tequality}[3]{\tstrengthen{#1}{#2=#3}}
\newcommand{\tsubst}[3]{#1[#2/#3]}

\newcommand{\tconst}{[\nu\!=\!c]^{\tdesignation}}
\newcommand{\tvareq}[1]{[\nu\!=\!#1]^{\tdesignation}}

\newcommand{\tstepbody}{\mathsf{body}^{\tdesignation}}
\newcommand{\ttransname}{\ctransname^{\tdesignation}}

\newcommand{\ttrans}[2][k+1]{\mathsf{step}^{\tdesignation}\llbracket#2\rrbracket}

\newcommand{\tcoll}[1][\cdot]{\mathbf{C}^{\tdesignation}\llbracket#1\rrbracket}
\newcommand{\tcollwid}[1][\cdot]{\mathbf{C}^{\tdesignation}_{\wid}\llbracket#1\rrbracket}

\newcommand{\tdomain}{\mathcal{P}^{\tdesignation}}

\newcommand{\wid}{\mathop{\triangledown}}
\newcommand{\shapewid}[1][\pscope]{\mathop{\triangledown^{\m{sh}}_{#1}}}
\newcommand{\shape}{\m{sh}}
\newcommand{\twid}[1][\pscope]{\mathop{\triangledown^{\tdesignation}_{#1}}}

\newcommand{\trelwid}[1][\pscope]{\mathop{\triangledown^{\m{ra}}_{#1}}}

\newcommand{\btype}[1]{\mathsf{#1}}
\newcommand{\tint}{\btype{int}}

\newcommand{\rtype}[2]{\{\nu\!: #1 \mid #2\}}

\newcommand{\subtype}{\mathrel{<:}}
\newcommand{\tenv}{\Gamma^{\tdesignation}}
\newcommand{\typrel}{\vdash}

\more{\settopmatter{printccs=false,printacmref=false}}

\newcommand{\smartparagraph}[1]{\smallskip\noindent
{\bf #1}\ }
\newcommand{\smartparagraphnb}[1]{\smallskip\noindent
{\it #1}\ }

\newcommand\nocolour{}

\usepackage[normalem]{ulem}

\newcommand{\zp}[1]  {\ifdefined\nocolour{#1}\else{\color{blue}{#1}}\fi}

\newcommand{\tw}[1]  {\ifdefined\nocolour{#1}\else{\color{blue}{#1}}\fi}

\newcommand{\ys}[1]  {\ifdefined\nocolour{#1}\else{\color{blue}{#1}}\fi}

\lessmore{
\setcopyright{acmcopyright}
\acmPrice{}
\acmDOI{10.1145/3434300}
\acmYear{2021}
\copyrightyear{2021}
\acmSubmissionID{popl21main-p144-p}
\acmJournal{PACMPL}
\acmVolume{5}
\acmNumber{POPL}
\acmArticle{19}
\acmMonth{1}
}{
  \setcopyright{none}
  \settopmatter{printacmref=false}
  \acmConference{Technical Report}
  \acmNumber{} 
  \acmYear{2020}
  \acmMonth{11}
  \startPage{1}
}

\begin{document}

\title{Data Flow Refinement Type Inference}
\more{\titlenote{This technical report is an extended version
    of~\cite{PavlinovicETAL21Drift}.}}

\author{Zvonimir Pavlinovic}
\affiliation{
  \institution{New York University}            
  \country{USA}
}
\affiliation{
  \institution{Google}
  \country{USA}
}
\email{zvonimir.pavlinovic@gmail.com}

\author{Yusen Su}
\affiliation{
  \institution{New York University}            
  \country{USA}
}
\affiliation{
  \institution{University of Waterloo}
  \country{Canada}
}
\email{ys3547@nyu.edu}

\author{Thomas Wies}
\affiliation{
  \institution{New York University}            
  \country{USA}
}
\email{wies@cs.nyu.edu}

%
%

\begin{abstract}

  Refinement types enable lightweight verification of functional programs.
  Algorithms for statically inferring refinement
  types typically work by reduction to solving systems of constrained
  Horn clauses extracted from typing derivations. An example is Liquid
  type inference, which solves the extracted constraints using
  predicate abstraction. However, the reduction to constraint solving
  in itself already signifies an abstraction of the program semantics
  that affects the precision of the overall static analysis. To better
  understand this issue, we study the type inference problem in its
  entirety through the lens of abstract interpretation. We propose a
  new refinement type system that is parametric with the choice of the
  abstract domain of type refinements as well as the degree to which
  it tracks context-sensitive control flow information. We then derive
  an accompanying parametric inference algorithm as an abstract
  interpretation of a novel data flow semantics of functional
  programs. We further show that the type system is sound and complete
  with respect to the constructed abstract semantics.
  Our theoretical development reveals the key abstraction steps
  inherent in refinement type inference algorithms. The trade-off
  between precision and efficiency of these abstraction steps is controlled
  by the parameters of the type system. Existing refinement type
  systems and their respective inference algorithms, such as Liquid
  types, are captured by concrete parameter instantiations.
  We have implemented our framework in a prototype tool and evaluated
  it for a range of new parameter instantiations (e.g., using octagons and
  polyhedra for expressing type refinements).
  \tw{The tool compares favorably against other existing tools. Our
    evaluation indicates that our approach can be used to systematically construct new
    refinement type inference algorithms that are both robust and precise.}

\end{abstract}



\maketitle



\section{Introduction}
\label{sec:introduction}

Refinement types are at the heart of static type systems that can
check a range of safety properties of functional
programs~\cite{DBLP:conf/pldi/FreemanP91,DBLP:conf/popl/DunfieldP04,DBLP:conf/fossacs/DunfieldP03,DBLP:conf/popl/XiP99,DBLP:conf/pldi/RondonKJ08,DBLP:conf/icfp/VazouSJVJ14,DBLP:conf/vmcai/ZhuJ13,DBLP:conf/pldi/VekrisCJ16,
  DBLP:conf/oopsla/ChughHJ12, DBLP:conf/tacas/ChampionC0S18}. Here, basic types are
augmented with \emph{refinement predicates} that
express relational dependencies between
inputs and outputs of functions. For example, the contract of an array
read operator can be expressed using the refinement type
\setlength{\abovedisplayskip}{6pt}
\setlength{\belowdisplayskip}{6pt}
\[\mathsf{get} :: (a:\alpha\; \mathsf{array}) \to (i:\{\nu:
  \mathsf{int} \mid 0 \leq \nu < \mathsf{length}\; a\}) \to \alpha \enspace.\]
This type indicates that $\mathsf{get}$ is a function that takes an
array $a$ over some element type $\alpha$ and an index $i$ as input
and returns a value of type $\alpha$. The type
$\{\nu: \mathsf{int} \mid 0 \leq \nu < \mathsf{length}\; a\}$ of the
parameter $i$ refines the base type $\mathsf{int}$ to indicate that
$i$ must be an index within the bounds of the array $a$. This type can
then be used to statically check the absence of erroneous array reads
in a program. However, such a check will only succeed if the index
expressions used in calls to $\mathsf{get}$ are also constrained by
appropriate refinement types. Therefore, a number of type inference
algorithms have been proposed that relieve programmers of the burden to provide such
type annotations manually. These
algorithms deploy a variety of analysis techniques ranging from
predicate abstraction~\cite{DBLP:conf/pldi/RondonKJ08, DBLP:conf/esop/VazouRJ13, DBLP:conf/icfp/VazouSJVJ14} to
interpolation~\cite{DBLP:conf/ppdp/UnnoK09, DBLP:conf/vmcai/ZhuJ13}
and machine
learning~\cite{DBLP:conf/icfp/ZhuNJ15,DBLP:conf/pldi/ZhuPJ16,
  DBLP:conf/tacas/ChampionC0S18}.  However, a common intermediate step
is that they reduce the inference problem to solving a system of
constrained Horn clauses that is induced by a typing derivation for the
program to be analyzed. However, this reduction already represents an
abstraction of the program's higher-order control flow and affects the
precision of the overall static analysis.



To better understand the interplay between the various abstraction steps
underlying refinement type inference algorithms, this paper forgoes
the usual reduction to constraints and instead studies the inference
problem in its entirety through the lens of abstract
interpretation~\cite{cousot1977abstract,cousot1979systematic}.
We start by introducing a parametric \emph{data flow refinement type
  system}. The type system generalizes from the specific choice of
logical predicates by allowing for the use of
arbitrary relational abstract domains as type refinements (including
e.g. octagons~\cite{DBLP:journals/corr/abs-cs-0703084},
polyhedra~\cite{DBLP:conf/popl/SinghPV17,DBLP:journals/scp/BagnaraHZ08,
  DBLP:conf/popl/CousotH78} and automata-based
domains~\cite{DBLP:conf/ictac/ArceriOCM19,DBLP:conf/vmcai/KimC11}).
Moreover, it is parametric in the degree to which it tracks
context-sensitive control flow information. This is achieved through
intersection function types, where the granularity at which such
intersections are considered is determined by how stacks are
abstracted at function call sites. Next, we propose a novel concrete
data flow semantics of functional programs that captures the program
properties abstracted by refinement type inference algorithms. From
this concrete semantics we then
construct an abstract semantics through a series of Galois
abstractions and show that the type system is sound and complete with
respect to this abstract semantics. Finally, we combine the abstract
semantics with an appropriate widening strategy to obtain an
accompanying parametric refinement type inference algorithm that is
sound by construction. The resulting analysis framework enables the
exploration of the broader design space of refinement type inference
algorithms. Existing algorithms such as Liquid type
inference~\cite{DBLP:conf/pldi/RondonKJ08} represent specific points
in this design space. \less{Additional details, including proofs,
  are available in a companion technical
  report~\cite{techreport}.}

To demonstrate the versatility of our framework, we have implemented
it in a verification tool targeting a subset of OCaml. We
have evaluated this tool for a range of new parameter instantiations and
compared it against existing verification tools for
functional programs. \tw{Our evaluation shows that the tool compares
favorably against the state of the art. In particular, for the 
higher-order programs over integers and lists in our benchmark suite, the tool improves over
existing tools both in terms of precision (more benchmarks solved) as
well as robustness (no analysis timeouts)}.


\section{Motivation}
\label{sec:df-overview}

To motivate our work, we provide an overview of common approaches to
inferring refinement types and discuss their limitations.

\subsection{Refinement Type Inference}
\label{sec:df-overview-liquid-algo}

\renewcommand{\lstlistingname}{Program}



\noindent
Consider the following definition of the Fibonacci
function in OCaml:
\begin{lstlisting}[language=Caml,aboveskip=0.5em,belowskip=0.2em,numbers=none]
let rec fib x = if x >= 2 then fib (x - 1) + fib (x - 2) else 1
\end{lstlisting}
%
The typical refinement type inference algorithm works as follows. First, the
analysis performs a standard Hindley-Milner type inference to infer
the basic shape of the refinement type for every subexpression of the
program. For instance, the inferred type for the
function \lstinline+fib+ is $\tint \to \tint$. 
For every function type $\tau_1 {\to} \tau_2$, where
$\tau_1$ is a base type such as $\tint$, the analysis  next introduces a fresh
dependency variable $x$ which stands for the function parameter,
$x: \tau_1 {\to} \tau_2$. The scope of $x$ is the result type
$\tau_2$, i.e., refinement predicates inferred for $\tau_2$ can 
express dependencies on the input value of type $\tau_1$ by referring to $x$. 
Further, every base type $\tau$ is replaced by a
refinement type, $\rtype{\tau}{\phi(\nu, \vec{x})}$, with a placeholder
refinement predicate $\phi(\nu, \vec{x})$ that expresses a relation between the 
members $\nu$ of $\tau$ and the other variables $\vec{x}$ in scope of the
type. For example, the augmented type for function \lstinline+fib+ is
\begin{align*}
x:\rtype{\tint}{\phi_1(\nu)} \to \rtype{\tint}{\phi_2(\nu, x)} \enspace.
\end{align*}
The algorithm then derives, either explicitly or implicitly, a system
of Horn clauses modeling the subtyping constraints imposed on the
refinement predicates by the program \emph{data flow}. For example,
the body of \lstinline+fib+ induces the following Horn clauses over
the refinement predicates in \lstinline+fib+'s type:
\begin{align*}
\phi_1(x) \land x \ge 2 \land \nu = x - 1 & \Rightarrow \phi_1(\nu)
&(1)&&
\phi_1(x) \land x \ge 2 \land \nu = x - 2 & \Rightarrow \phi_1(\nu)
&(2)\\
\multicolumn{5}{r}{
  $\phi_1(x) \land x \ge 2 \land \phi_2(\nu_1,x-1) \land \phi_2(\nu_2,x-2)
  \land \nu = \nu_1 + \nu_2$} & \Rightarrow \phi_2(\nu,x) & (3) \\
      \phi_1(x) \land x < 2 \land \nu = 1 & \Rightarrow \phi_2(\nu, x)
  & (4) &&
  0 \leq \nu & \Rightarrow \phi_1(\nu) & (5)
\end{align*}
Clauses (1) and (2) model the data flow from 
\lstinline+x+ to the two recursive calls in the \emph{then} branch of
the conditional. Clauses (3) and (4) capture the constraints
on the result value returned by \lstinline+fib+ in the \emph{then} and
\emph{else} branch. Clause (5) captures an assumption that we
make about the \emph{external calls} to \lstinline+fib+, namely, that
these calls always pass non negative values. We note that the
inference algorithm performs a whole program analysis. Hence, when one
analyzes a program fragment or individual function as in this case,
one has to specify explicitly any assumptions made about the
context.

%

The analysis then solves the obtained Horn clauses to derive the
refinement predicates $\phi_i$. For instance, Liquid type inference
uses monomial predicate abstraction for this purpose. That is, the analysis assumes a given set of atomic
predicates $Q = \{p_1(\nu,\vec{x}_1), \dots, p_n(\nu,\vec{x}_n)\}$, which
are either provided by the programmer or derived from the program
using heuristics, and then infers an assignment for each $\phi_i$ to a
conjunction over $Q$ such that all Horn clauses are valid. This can be
done effectively and efficiently using the Houdini
algorithm~\cite{DBLP:conf/fm/FlanaganL01,DBLP:conf/cade/LahiriQ09}.
For example, if we choose
$Q = \{0 \leq \nu, 0 > \nu, \nu < 2, \nu \ge 2\}$, then the final type
inferred for function \lstinline+fib+ will be:
\begin{align*}
x:\rtype{\tint}{0 \leq \nu} \to \rtype{\tint}{0 \leq \nu}\enspace.
\end{align*}
The meaning of this type is tied to the current program, or in this
case the assumptions made about the context of the program fragment
being analyzed. In particular, note that the
analysis does not infer a more general type 
that would leave the input parameter \lstinline+x+ of \lstinline+fib+
unconstrained.

\subsection{Challenges in Inferring Precise Refinement Types}

Now, suppose that the goal of the analysis is to verify that
function \lstinline+fib+ is increasing, which can be done by inferring
a refinement predicate $\phi_2(\nu, x)$ for the return type of
\lstinline+fib+ that implies $x \leq \nu$. Note that $x \leq \nu$ itself
is not inductive for the system of Horn clauses derived
above \tw{because clause (3) does not hold for $x=2$, $v_1=1$, and $v_2=0$.}
However, if we strengthen $\phi_2(\nu, x)$ to $x \leq \nu \land 1 \leq
\nu$, then it is inductive.

One issue with using predicate abstraction for inferring type
refinements is that the analysis needs to guess in advance which
auxiliary predicates will be needed, here $1 \leq \nu$. Existing tools
based on this approach, such as \tool{DSolve}~\cite{DBLP:conf/pldi/RondonKJ08} and \tool{Liquid
Haskell}~\cite{DBLP:conf/haskell/VazouBKHH18}, use heuristics for this
purpose. However, these heuristics tend to be brittle. In fact, both tools
fail to verify that \lstinline+fib+ is increasing if the
predicate $1 \leq \nu$ is not explicitly provided by the user. Other
tools such as \tool{R\_Type}~\cite{DBLP:conf/tacas/ChampionC0S18} are based on more complex analyses that
use counterexamples to inductiveness to automatically infer the
necessary auxiliary predicates. However, these tools no longer
guarantee that the analysis terminates. Instead, our approach enables
the use of \tw{expressive numerical abstract domains such
as polyhedra to infer sufficiently precise refinement types in practice, without giving up on
termination of the analysis or requiring user annotations (see \S~\ref{sec:implementation})}.

If the goal is to improve precision, one may of course ask why it is
necessary to develop a new refinement type inference analysis from
scratch. Is it not sufficient to improve the deployed Horn clause
solvers, e.g. by using better abstract domains? Unfortunately, the
answer is ``no''\zp{~\cite{DBLP:conf/popl/UnnoTK13}}. The derived Horn clause system already signifies an
abstraction of the program's semantics and, in general, entails an inherent
loss of precision for the overall analysis. 

\begin{figure}[t]
\begin{lstlisting}[language=Caml,aboveskip=0.5em,belowskip=0.2em,escapechar=|]
let apply f x = f x and g y = 2 * y and h y = -2 * y
let main z =
  let v = if 0 <= z then (apply|${}_i$| g)|${}_j$| z else (apply|${}_k$| h)|${}_\ell$| z in |\label{line:v-decl}|
  assert (0 <= v) |\label{line:assert}|
\end{lstlisting}
\vspace*{-1em}
\captionof{lstlisting}{\label{prg:apply}}
\end{figure}

To motivate this issue, consider Program~\ref{prg:apply}. You may
ignore the program location labels $i,j,k,\ell$ for now.
Suppose that the goal is to verify that the \lstinline+assert+ statement in the last line is
safe. The templates for the refinement types of the top-level
functions are as follows:
\begin{align*}
\mathtt{apply}::&\; (y:\rtype{\tint}{\phi_1(\nu)} \to
\rtype{\tint}{\phi_2(\nu, y)}) \to x:{\rtype{\tint}{\phi_3(\nu)}} \to
\rtype{\tint}{\phi_4(\nu, x)}\\
\mathtt{g}::&\; y:\rtype{\tint}{\phi_5(\nu)} \to \rtype{\tint}{\phi_6(\nu,y)}
\quad \mathtt{h}:: y:\rtype{\tint}{\phi_7(\nu)} \to \rtype{\tint}{\phi_8(\nu,y)}
\end{align*}
Moreover, the key Horn clauses are:
\begin{align*}
0 \leq z \land \nu = z & \Rightarrow \phi_3(\nu)
& \phi_5(y) \land \nu = 2 y & \Rightarrow \phi_6(\nu, y)
& \phi_3(x) & \Rightarrow \phi_1(\nu)\\
0 \leq z \land \phi_1(\nu) & \Rightarrow
\phi_5(\nu)
& 0 \leq z \land \phi_6(\nu, y) & \Rightarrow
\phi_2(\nu, y)
&
\phi_2(\nu, x) & \Rightarrow \phi_4(\nu, x)
\\
0 > z \land \nu = z & \Rightarrow \phi_3(\nu)
& \phi_7(y) \land \nu = - (2 y) & \Rightarrow \phi_8(\nu, y)
\\
0 > z \land \phi_1(\nu) & \Rightarrow
\phi_7(\nu)
& 0 > z \land \phi_8(\nu, y) & \Rightarrow
\phi_2(\nu, y)
\end{align*}
Note that the least solution of these Horn clauses satisfies
$\phi_1(\nu)=\phi_3(\nu)=\phi_5(\nu)=\phi_7(\nu)=\mathsf{true}$ and
$\phi_2(\nu, x)=\phi_4(\nu,x)=(\nu = 2x \lor \nu =
-2x)$. Specifically, $\phi_4(\nu,x)$ is too weak to conclude that
the two calls to \lstinline+apply+ on line~\ref{line:v-decl} always return positive integers. Hence, any analysis based on deriving a
solution to this Horn clause system will fail to infer refinement
predicates that are sufficiently strong to entail the safety of the
assertion in \lstinline+main+.  The problem is that the generated Horn
clauses do not distinguish between the two functions \lstinline+g+
and \lstinline+h+ that \lstinline+apply+ is called with. All existing
refinement type inference tools that follow this approach of
generating a context-insensitive Horn clause abstraction of the
program therefore fail to verify Program~\ref{prg:apply}.

To obtain a better understanding where existing refinement type inference
algorithms lose precision, we take a fresh look at this problem through
the lens of abstract interpretation.
\section{Preliminaries}
\label{sec:preliminaries}

We introduce a few notations and basic definitions that we use throughout the paper.

\smartparagraph{Notation.}
We often use meta-level $\Mlet x = t_1\, \Min t_0$ and conditional
$\Mif t_0\, \Mthen t_1\, \Melse t_2$ constructs in mathematical
definitions.
We compress consecutive let bindings
$\Mlet x_1 = t_1 \;\Min \dots \Mlet x_n = t_n \;\Min t_0$ as $\Mlet
x_1 = t_1 \Mend \dots \Mend x_n = t_n \;\Min t_0$.
We use capital lambda notation
($\Lambda x.\;\dots$) for defining mathematical functions. For a function
$f:X\to Y$, $x \in X$, and $y \in Y$, we write $f[x \mapsto y]$
to denote a function that maps $x$ to $y$ and otherwise agrees with $f$ on
every element of $X\setminus\{x\}$. \tw{We use the notation $f.x\!:y$ instead of $f[x
\mapsto y]$ if $f$ is an environment mapping variables $x$ to their
bindings $y$.}
For a set $X$ we denote its powerset by $\powerset(X)$. 
For a relation
$R \subseteq X \times Y$ over sets $X$, $Y$ and a natural number $n > 0$, we use
$\vecop{R}_n$ to refer to the point-wise lifting of $R$ to a relation
on $n$-tuples $X^n \times Y^n$. That is $\pair{(x_1, \dots, x_n)}{(y_1, \dots, y_n)} \in \vecop{R}_n$
iff $(x_i, y_i) \in R$ for all $1 \leq i \leq n$. Similarly, for any nonempty set $Z$ we
denote by $\vecop{R}_Z$ the point-wise lifting of $R$ to a relation
over $(Z \to X) \times (Z \to Y)$. More precisely, if $f_1: Z \to X$ and $f_2: Z \to Y$, then
$(f_1, f_2) \in \vecop{R}_Z$ iff $\forall z \in Z.\,(f_1(z), f_2(z)) \in R$.
Typically, we drop the subscripts from these lifted relations when they are clear from 
the context.

For sets $X$, $Y$ and a function $d: X \to \powerset(Y)$, we use the notation
$\Pi x \in X. d(x)$ to refer to the set
$\pset{f: X \to Y}{\forall x \in X.\, f(x) \in d(x)}$ of all dependent
functions with respect to $d$. Similarly, for given sets $X$ and $Y$ 
we use the notation $\Sigma x \in X. d(x)$ to refer to the set
$\pset{\pair{x}{y}: X \times Y}{y \in d(x)}$ of all dependent pairs
with respect to $d$. We use the operators $\pi_1$ and $\pi_2$ to
project to the first, resp., second pair component.

\smartparagraph{Abstract interpretation.}
A partially ordered set (poset) is a pair $(L, \sqleq)$ consisting of
a set $L$ and a binary relation $\sqleq$ on $L$
that is reflexive, transitive, and antisymmetric.
Let $(L_1, \sqleq_1)$ and $(L_2, \sqleq_2)$ be two posets.  
We say that two functions
$\alpha \in L_1 \rightarrow L_2$ and $\gamma \in L_2 \rightarrow L_1$
form a \textit{Galois connection} iff
\[
  \forall x \in L_1, \forall y \in L_2.\; \alpha(x) \sqleq_2 y \iff x
  \sqleq_1 \gamma(y) \enspace.
\]
We call $L_1$ the \textit{concrete} domain and $L_2$ the
\textit{abstract} domain of the Galois connection. 
Similarly, $\alpha$ is called \textit{abstraction} function (or left
adjoint) and
$\gamma$ \textit{concretization} function (or right adjoint). 
Intuitively, $\alpha(x)$ is the most precise approximation of $x \in L_1$ in $L_2$
while $\gamma(y)$ is the least precise element of $L_1$ that can be approximated by $y \in L_2$.

A \emph{complete lattice} is a tuple
$\langle L, \sqleq, \bot, \top, \sqcup, \sqcap \rangle$ where $(L, \sqleq)$ is a
poset such that for
any $X \subseteq L$, the least upper bound $\sqcup X$ (join)
and greatest lower bound $\sqcap X$ (meet) with respect to $\sqleq$
exist. In particular, we have $\bot = \sqcap L$ and
$\top = \sqcup L$. We often identify a complete lattice with
its carrier set $L$.

Let $(L_1, \sqleq_1, \bot_1, \top_1, \sqcup_1, \sqcap_1)$ and $(L_2, \sqleq_2, \bot_2, \top_2, \sqcup_2, \sqcap_2)$
be two complete lattices and let $(\alpha, \gamma)$ be a Galois
connection between $L_1$ and $L_2$.
Each of these functions uniquely determines the other:
\[\alpha(x) = \sqcap_2  \setc{y \in L_2}{x \sqleq_1 \gamma(y)} \ind[5] \gamma(y) = \sqcup_1 \setc{x \in L_1}{\alpha(x) \sqleq_2 y} \]
\noindent Also, $\alpha$ is a \textit{complete join-morphism}
$\forall S \subseteq L_1.\; \alpha(\sqcup_1 S) = \sqcup_2 \setc{\alpha(x)}{x \in S}, \alpha(\bot_1) = \bot_2$
and $\gamma$ is a \textit{complete meet-morphism}
$\forall S \subseteq L_2.\; \gamma(\sqcap_2 S) = \sqcap_1 \setc{\gamma(x)}{x \in S}, \gamma(\top_2) = \top_1.$ A similar result holds in the other 
direction: if $\alpha$ is a complete join-morphism and $\gamma$ is defined
as above, then $(\alpha, \gamma)$ is a Galois connection between $L_1$ and $L_2$. 
Likewise, if $\gamma$ is a complete meet-morphism and $\alpha$ is defined 
as above, then $(\alpha, \gamma)$ is a Galois connection~\cite{cousot1979systematic}.
\section{Parametric Data Flow Refinement Types}
\label{sec:types}

We now introduce our parametric data flow refinement type system. The
purpose of this section is primarily to build intuition. The remainder
of the paper will then formally construct the type system as an
abstract interpretation of our new data flow semantics. As a reward, we
will also obtain a parametric algorithm for inferring data flow
refinement types that is sound by construction.


\smartparagraph{Language.}
Our formal presentation considers a simple untyped lambda calculus:
\begin{align*}
  \qquad e \in \Exp ::= \;  \const \;
  \mid \; x \;
  \mid \;  e_1 \, e_2 \;
  \mid \; \lambda x.\, e
\end{align*}
The language supports constants $c \in \consts$ (e.g. integers and Booleans), variables $x \in
\vars$, lambda abstractions, and function applications.
An expression $e$ is \emph{closed} if all using occurrences of
variables within $e$ are bound in lambda abstractions $\lambda x.\, e$.
A \emph{program} is a closed expression. 

Let $e$ be an expression. Each subexpression of $e$ is uniquely 
annotated with a location drawn from the set $\elabels$. We denote
locations by $\elabel,\elabela,\elabelb$ and use subscript
notation to indicate the location identifying a (sub)expression as in
$(e_\elabela \, e_\elabelb)_\elabel$ and
$(\lambda x.e_\elabela)_\elabel$. The location annotations are
omitted whenever possible to avoid notational clutter.
Variables are also locations, i.e. $\vars \subseteq \elabels$.

In our example programs we often use let constructs. Note that
these can be expressed using lambda abstraction and function
application \tw{as usual:
\[\texttt{let}\; x = e_1 \;\texttt{in}\; e_2 \; \Def{=} \;
  (\lambda x.\,e_2)\,e_1 \enspace.\]
}
\smartparagraph{Types.}
Our data flow refinement type system takes two parameters: (1) a
finite set of \emph{abstract stacks} $\pstacks$, and (2) a (possibly
infinite) complete lattice of \emph{basic refinement types}
$\langle \bdomain, \bord, \bbot, \btop, \bjoin, \bmeet \rangle$. We
will discuss the purpose of abstract stacks in a moment. A basic
refinement type $\bval \in \bdomain$ comes equipped with an implicit
scope $\pscope \subseteq \vars$. Intuitively, $\bval$ represents a
relation between
primitive constant values (e.g. integers) and other values 
bound to the variables in $\pscope$. The partial order $\bord$ is an
abstraction of subset inclusion on such relations. We will make this
intuition formally precise later. For $\pscope \subseteq \vars$, we
denote by $\bdomain_\pscope$ the set of all basic refinement types with scope
$\pscope$.
\begin{example}
  \label{ex:lia-basic-types}
  Let $\phi(\pscope)$ stand for a convex linear constraint over (integer)
  variables in $\pscope \cup \set{\nu}$. Then define the set of basic
  refinement types
  \[\bdomain[\mathsf{lia}]_\pscope ::= \bbot \mid \btop \mid \setc{\nu:
      \tint}{\phi(\pscope)}
  \]
  An example of a basic type in $\bdomain[\mathsf{lia}]$ with scope
  $\set{x}$ is $\setc{\nu : \tint}{x \leq \nu \land 1 \leq \nu}$. The
  order $\bord$ on $\bdomain[\mathsf{lia}]$ is obtained by lifting
  the entailment ordering on linear constraints to
  $\bdomain[\mathsf{lia}]$ in the expected way. If we identify linear
  constraints up to equivalence, then $\bord$ induces a complete
  lattice on $\bdomain[\mathsf{lia}]$.
\end{example}

The basic refinement types are extended to data flow refinement types as follows:
\[\begin{array}{c}
\tval \in \tvalues ::= \tbot \mid \ttop \mid \bval \mid x:\ttable  \ind[1.9]
\bval \in \bdomain_\pscope \ind[1.9]
\absn:\ttable \in \ttables \Def= \Sigma \absn \in \vars
\setminus \pscope.\,
\pstacks \to \tvalues[\pscope] \times \tvalues[\pscope \cup \{\absn\}]
\end{array}\]
A data flow refinement type $\tval
\in \tvalues$
also has an implicit scope $\pscope$. We denote by $\tvalues[] =
\bigcup_{\pscope \subseteq \vars} \tvalues$ the set of all such types
for all scopes.  There are four kinds of types. First, the type $\tbot$ should be
interpreted as \emph{unreachable} or \emph{nontermination} and the type
$\ttop$ stands for a \emph{type error}. We introduce these types explicitly
so that
we can later endow $\tvalues[]$ with a partial order to form a complete lattice.
In addition to $\tbot$ and $\ttop$, we have basic refinement types $\bval$ and
(dependent) function types $x:\ttable$. The latter are pairs consisting
of a dependency variable $x \in \vars \setminus \pscope$ and a \emph{type table}
$\ttable$ that maps abstract stack $\pstack \in \pstacks$ to pairs of
types $\ttable(\pstack)=\pair{\tval_i}{\tval_o}$. That is, $x:\ttable$
can be understood as capturing a separate dependent function type $x:\tval_i
\to \tval_o$ per abstract stack $\pstack$. Abstract stacks
represent abstractions of concrete call stacks, enabling function
types to case split on different calling contexts of the
represented functions. In this sense, function types \zp{resemble}
intersection types \zp{with \emph{ad hoc} polymorphism on contexts}.
Note that the scope of the
output type $\tval_o$ includes the dependency variable $x$, thus
enabling the output type to capture input/output dependencies.

Let $t=x:\ttable$ be a function type and $\pstack \in \pstacks$ such
that $\ttable(\pstack)=\pair{\tval_i}{\tval_o}$ for some $\tval_i$ and
$\tval_o$. We denote $\tval_i$
by $\tin{\ttable(\pstack)}$ and $\tval_o$ by
$\tout{\ttable(\pstack)}$. We say that $\tval$ has been called at
$\pstack$, denoted $\pstack \in \tval$, if 
$\tin{\ttable(\pstack)}$ is not $\tbot$. We denote by $\ttable_{\bot}$
the \emph{empty type table} that maps every abstract stack to the pair
$\pair{\tbot}{\tbot}$ and write $[\tentry{\pstack}{\tval_i}{\tval_o}]$
as a short hand for the \emph{singleton type table}
$\ttable_{\bot}[\pstack \mapsto \pair{\tval_i}{\tval_o}]$. We extend
this notation to tables obtained by explicit enumeration of table
entries. Finally, we denote by $\restr{\tval}{\pstack}$ the function type
$x:[\tentry{\pstack}{\tval_i}{\tval_o}]$ obtained from $\tval$ by
restricting it to the singleton table for $\pstack$.

\begin{example}
  \label{ex:lia-types}
  \arraycolsep=0pt%
  Consider again the set $\bdomain[\lia]$ from
  Example~\ref{ex:lia-basic-types}.  In what follows, we write $\tint$
  for $\setc{\nu:\tint}{\btrue}$ and $\setc{\nu}{\phi(X)}$ for
  $\setc{\nu:\tint}{\phi(X)}$. Let further $\pstacks^0 = \set{\epsilon}$ be a
  trivial set of abstract stacks. We then instantiate $\tvalues[]$
  with $\bdomain[\lia]$ and $\pstacks^0$. Since tables only have a
  single entry, we use the more familiar notation
  $x:\tval_i \to \tval_o$ for function types, instead of
  $x:[\tentry{\epsilon}{\tval_i}{\tval_o}]$. We can then represent the
  type of \lstinline+fib+ in \S~\ref{sec:df-overview-liquid-algo} by the data flow
  refinement type
  \[x: \setc{\nu}{0 \leq \nu} \to \setc{\nu}{x \leq \nu \land 1 \leq \nu}\enspace.\]
  The most precise type that we can infer for function \lstinline+apply+ in
  Program~\ref{prg:apply} is
  \[f: (y: \tint \to \tint) \to x: \tint \to \tint\enspace.\]
  \tw{Here, the variables $f$, $x$, and $y$ refer to the corresponding
    parameters of the function
  \lstinline+apply+, respectively, the functions \lstinline+g+ and
  \lstinline+h+ that are passed to parameter $f$ of \lstinline+apply+.}
  Now, define the set of abstract stacks $\pstacks^1 =
  \elabels \cup \set{\epsilon}$. Intuitively, we can use the elements of
  $\pstacks^1$ to distinguish table entries in function types 
  based on the program locations where the represented functions are
  called. Instantiating our set of types $\tvalues[]$ with
  $\bdomain[\lia]$ and $\pstacks^1$, we can infer a more
  precise type for function \lstinline+apply+ in
  Program~\ref{prg:apply}:
  \[f: \begin{array}[t]{ll}
    [&i \triangleleft y:[j \triangleleft \setc{\nu}{0
      \leq \nu} \to \setc{\nu}{\nu = 2y}] 
    \to x:[j \triangleleft \setc{\nu}{0 \leq \nu} \to \setc{\nu}{\nu =
2x}],\\
   &k \triangleleft y:[\ell \triangleleft \setc{\nu}{0
      > \nu} \to \setc{\nu}{\nu = -2y}] 
    \to x:[\ell \triangleleft \setc{\nu}{0 > \nu} \to \setc{\nu}{\nu =
-2x}
]\enspace.
\end{array}
\]
  Note that the type provides sufficient information to distinguish
  between the calls to \lstinline+apply+ with functions \lstinline+g+
  and \lstinline+h+ at call sites $i$ and $k$, respectively. This
  information is sufficient to guarantee the correctness of the
  assertion on line~\ref{line:assert}.
\end{example}

\smartparagraph{Typing environments and operations on types.} 
Before we can define the typing rules, we first need to introduce a
few operations for constructing and manipulating types. We here only
provide the intuition for these operations through examples. In
\S~\ref{sec:arefinement} we will then explain how to define these
operations in terms of a few simple primitive operations provided by the
domain $\bdomain$.

First, the (basic) refinement type abstracting a single constant
$\const \in \consts$ is denoted by $\tconst$. For instance, for our
basic type domain $\bdomain[\lia]$ from
Example~\ref{ex:lia-basic-types} and an integer constant $c \in \ZZ$,
we define $\tconst = \setc{\nu}{\nu = \const}$. Next, given a type
$\tval$ over scope $\pscope$ and a type $\tval'$ over scope
$\pscope' \subseteq \pscope \setminus \set{x}$, we denote by
$\rupdate{\tval}{\absn}{\tval'}$ the type obtained from $\tval$ by
strengthening the relation to $x$ with the information provided by
type $\tval'$. Returning to Example~\ref{ex:lia-basic-types}, for two
basic types $\bval = \setc{\nu}{\phi(\pscope)}$ and
$\bval' = \setc{\nu}{\phi'(\pscope')}$, we have
$\rupdate{\bval}{\absn}{\bval'} = \setc{\nu}{\phi(X) \land
  \phi'(X')[x/\nu]}$. Finally, for a variable $x \in \pscope$ and
$\tval \in \tvalues$, we assume an operation $\tval\tvareq{x}$ that
strengthens $\tval$ by enforcing equality between the value bound to
$x$ and the value represented by $\tval$. \tw{For instance, for a base
  type $\bval = \setc{\nu}{\phi(\pscope)}$ from
  Example~\ref{ex:lia-basic-types} we have
  $\bval\tvareq{x}=\setc{\nu}{\phi(X) \land \nu = x}$.}

A \emph{typing environment} for scope $\pscope$ is a function
$\tenv \in (\Pi x \in \pscope.\, \tvalues[\pscope \setminus \{x\}])$. We lift
$\rupdate{\tval}{\absn}{\tval'}$ to an operation
$\rstrengthen{\tval}{\tenv}$ that strengthens $\tval$ with 
the constraints on variables in $\pscope$ imposed by the types bound to
these variables in $\tenv$.

We additionally assume that abstract stacks $\pstacks$ come equipped
with an \emph{abstract concatenation operation}
$\pconcat: \elabels \times \pstacks \to \pstacks$ that prepends a
call site location $\elabela$ onto an abstract stack $\pstack$,
denoted $\elabela \pconcat \pstack$. For instance, consider again the
sets of abstract stacks $\pstacks^0$ and $\pstacks^1$ introduced in
Example~\ref{ex:lia-types}. We define
$\pconcat$ on $\pstack \in \pstacks^0$ as
$\elabela \pconcat \pstack = \epsilon$ and on $\pstack \in \pstacks^1$
we define it as $\elabela \pconcat \pstack = \elabela$. The general
specification of $\pconcat$ is given in \S~\ref{sec:acollapsed}.

\smartparagraph{Typing rules.}  Typing judgements take the form
$\tenv,\pstack \typrel e : \tval$ and rely on the \emph{subtype
  relation} $\tval \subtype \tval'$ defined in
Fig.\ref{fig:subtyping-rules}. The rule \textsc{s-bot} states that
$\tbot$ is a subtype of all other types except $\ttop$. \tw{Since
  $\ttop$ denotes a type error, the rules must ensure that we
  do not have $\tval \subtype \top$ for any type $\tval$}. The rule \textsc{s-base} defines
subtyping on basic refinement types, which simply defers to the
partial order $\bord$ on $\bdomain$. Finally, the rule \textsc{s-fun} is
reminiscent of the familiar contravariant subtyping rule for dependent
function types, except that it quantifies over all
entries in the type tables. \tw{In \S~\ref{sec:typing-correctness}, we
  will establish a formal correspondence between subtyping and the
  propagation of values along data flow paths.}

\begin{figure}[t]
\begin{mathpar}
\inferrule*[lab=s-bot,vcenter]{\tval \neq \ttop}{\tbot \subtype
  \tval} \ind[2]
\inferrule*[lab=s-base,vcenter]{\bval_1 \bord \bval_2}{\bval_1 \subtype
  \bval_2} \ind[2]
\inferrule*[left=s-fun,right=$\forall \pstack \in \pstacks$,vcenter,width=8cm]
{
\ttable_1(\pstack) = \pair{\tval_{i1}}{\tval_{o1}} \\
\ttable_2(\pstack) = \pair{\tval_{i2}}{\tval_{o2}} \\
\tval_{i2} \subtype \tval_{i1} \\
\tupdate{\tval_{o1}}{x}{\tval_{i2}} \subtype
\tupdate{\tval_{o2}}{x}{\tval_{i2}}}
{x:\ttable_1 \subtype x:\ttable_2}
\end{mathpar}
\caption{Data flow refinement subtyping rules}
\label{fig:subtyping-rules}
\end{figure}

\begin{figure}[t]
\begin{mathpar}
  \inferrule*[lab=t-var]
{\rstrengthen{\tenv(x)\tvareq{x}}{\tenv} \subtype
  \rstrengthen{\tval\tvareq{x}}{\tenv} }
{\tenv,\pstack \typrel x : \tval} \ind[2]
\inferrule*[lab=t-app]
{  \tenv,\pstack \typrel e_\elabela : \tval_\elabela \\
  \tenv,\pstack \typrel e_\elabelb :
\tval_\elabelb \\
  \tval_\elabela \subtype x:[\elabela \pconcat \pstack \triangleleft
   \tval_\elabelb \to \tval]}
 {\tenv,\pstack \typrel e_\elabela \, e_\elabelb : \tval}\\
 \inferrule*[lab=t-const]
 {\tconst[\tenv] \subtype
\tval}{\tenv,\pstack \typrel \const : \tval} \ind[1]
\inferrule*[lab=t-abs,right=$\forall \pstack' \in \tval$]
{
\tenv_\elabela = \tenv.x\!:\tval_x \\ \tenv_\elabela,\pstack' \typrel
e_\elabela : \tval_\elabela \\ x\!:[\pstack' \triangleleft \tval_x \to
\tval_\elabela] \subtype \restr{\tval}{\pstack'}
}
{\tenv,\pstack \typrel \lambda x.\,e_\elabela : \tval}
\end{mathpar}
\caption{Data flow refinement typing rules}
\label{fig:typing-rules}
\end{figure}

The rules defining the typing relation
$\tenv,\pstack \typrel e : \tval$ are shown in
Fig.~\ref{fig:typing-rules}. We implicitly require that all free
variables of $e$ are in the scope of $\tenv$. The rule
\textsc{t-const} for typing constants requires that $\tconst$ is a
subtype of the type $\tval$, after strengthening it with all
constraints on the variables in scope obtained from $\tenv$. That is,
we \emph{push} all environmental assumptions into the types. This way,
subtyping can be defined without tracking explicit typing
environments. We note that this formalization does not preclude an
implementation of basic refinement types that tracks typing
environments explicitly. The rule \textsc{t-var} for typing variable
expressions $x$ is similar to the rule \textsc{t-const}. That is, we
require that the type $\tenv(x)$ bound to $x$ is a subtype of $\tval$,
modulo strengthening with the equality constraint $\nu = x$ and all
environmental constraints. Note that we here strengthen both sides of
the subtype relation, which is necessary for completeness of the
rules, due to the bidirectional nature of subtyping for function
types.

The rule \textsc{t-app} for typing function applications
$e_\elabela \, e_\elabelb$ requires that the type $\tval_\elabela$ of
$e_\elabela$ must be a subtype of the function type
$x:[\elabela \pconcat \pstack: \tval_\elabelb \to \tval]$ where
$\tval_\elabelb$ is the type of the argument expression $e_\elabelb$
and $\tval$ is the result type of the function application. Note that
the rule extends the abstract stack $\pstack$ with the call site
location $\elabela$ \tw{identifying $e_\elabela$}. The subtype relation then forces $\tval_\elabela$
to have an appropriate entry for the abstract call stack
$\elabela \pconcat \pstack$. The rule \textsc{t-abs} for typing lambda
abstraction is as usual, except that it universally quantifies over
all abstract stacks $\pstack'$ at which $\tval$ has been called. 
The side condition $\forall \pstack' \in \tval$ implicitly
constrains $\tval$ to be a function type.


\section{Data flow Semantics}
\label{sec:df-dataflow}

Our goal is to construct our type system from a concrete semantics of
functional programs, following the usual calculational approach taken
in abstract interpretation. This imposes some constraints on our
development. In particular, given that the typing rules are defined by
structural recursion over the syntax of the evaluated expression, the
same should be true for the concrete semantics that we take as the
starting point of our construction. This requirement rules out
standard operational semantics for higher-order programs (e.g. based on function closures)
because they evaluate function bodies at call sites rather than
definition sites, making these semantics nonstructural. A more natural
choice is a standard denotational semantics (e.g. one that interprets lambda
terms by mathematical functions). However, the problem with
denotational semantics is that it is inherently compositional;
functions are given meaning irrespective of the context in which they
appear. However, as we have discussed in \S~\ref{sec:df-overview}, the
function types inferred by algorithms such as Liquid types only
consider the inputs to functions that are observed in the program under
consideration. Denotational semantics are therefore ill-suited for capturing
the program properties abstracted by these type systems.

Hence, we introduce a new \emph{data flow refinement type semantics}.
Like standard denotational semantics, it is fully structural in the
program syntax but without being compositional. That is, it captures
the intuition behind refinement type inference algorithms that view a function value as a \emph{table}
that records all inputs the function will receive \emph{from this
  point onwards} as it continues to flow through the program.

\subsection{Semantic Domains}
\label{sec:dfs-sdomains}
We start with the semantic domains used for giving meaning to expressions $e \in \Exp$:
\begin{align*}
\cnode \in \cnodes \; \Def= {}  & \cenodes \cup \cvnodes &&&
\cenodes \; \Def= {}  & \elabels  \times \cenvironments &&&
\cvnodes \; \Def= {}  & \vars \times \cenvironments \times \cstacks \\
\cstack \in \cstacks \; \Def= {}  & \elabels^{*} &&&
\cenv \in \cenvironments \; \Def= {}  & \vars \pto_{\mi{fin}} \cvnodes &&&
\cmap \in \cmaps \; \Def= {}  &\cnodes \to \cvalues \\
\cval \in \cvalues \; ::= {} & \bot \mid \cerr \mid \const \mid \ctable &&&
\ctable \in \ctables \; \Def= {} & \ccallsites \to \cvalues \times \cvalues 
\end{align*}
{\textbf{Nodes, stacks, and environments.\;}} Every intermediate
point of a program's execution is uniquely identified by an
\textit{(execution) node}, $\cnode \in \cnodes$, a
concept which we adapt
from~\cite{Jagannathan:1995:UTF:199448.199536}.  We distinguish
\emph{expression nodes} $\cenodes$ and \emph{variable nodes}
$\cvnodes$. An expression node $\pair{\elabel}{\cenv}$, denoted
$\cn{\elabel}{\cenv}$, captures the execution point where a
subexpression $e_\elabel$ is evaluated in the environment $\cenv$.
An environment $\cenv$ is a (finite) partial map binding variables to
variable
nodes. 
A variable node $\pair{x}{\cenv,\cstack}$, denoted $\cns{x}{\cenv}{\cstack}$, is
created at each execution point where an argument value is bound to a
formal parameter $x$ of a function at a call site. Here, $\cenv$ is the environment
at the point where the function is defined and $\cstack$ is the \emph{call
site stack}\footnote{referred to as \emph{contour} in~\cite{Jagannathan:1995:UTF:199448.199536}}
of the variable node. \zp{Note that we define nodes and environments using
mutual recursion where the base cases are defined using the empty
environment $\epsilon$}. The call site stack captures the
sequence of program locations of all pending function calls before
this variable node was created. That is, intuitively, $\cstack$ can be
thought of as recording the return addresses of these pending
calls. We write $\elabel \cdot \cstack$ to denote the stack obtained
from $\cstack$ by prepending $\elabel$. Call site stacks are used to
uniquely identify each variable binding.

We explain the role of expression and variable nodes in more detail
later. For any node $\cnode$, we denote by $\nloc(\cnode)$ the
location of $\cnode$ and by $\nenv(\cnode)$ its environment. If $n$ is a variable node, we denote its stack
component by $\nstack(\cnode)$.
A pair $\pair{e}{\cenv}$ is called well-formed if $\cenv(x)$ is defined for
every variable $x$ that occurs free in $e$.

\smartparagraph{Values and execution maps.}
Similar to our data flow refinement types, there are four kinds of
(data flow) values $\cval \in \cvalues$. First, every constant
$\const$ is also a value. The value $\bot$ stands for nontermination
or unreachability of a node, and the value $\cerr$ models execution
errors. Functions are represented by \emph{tables}. A table $\ctable$
maintains an input/output value pair $T(S)=\pair{v_i}{v_o}$ for each
call site stack $\cstack$. We adapt similar notation for concrete
function tables as we introduced for type tables in
\S~\ref{sec:types}. In particular, we denote by $\ctable_{\bot}$
the table that maps every call site stack to the pair
$\pair{\bot}{\bot}$. We say that a value $\cval$ is safe, denoted
$\safe(\cval)$, if \tw{$\cerr$ does not occur anywhere in $\cval$, i.e.,} $\cval \neq \cerr$ and if $\cval \in \ctables$ then
for all $\cstack$, both $\tin{\cval(\cstack)}$ and
$\tout{\cval(\cstack)}$ are safe.

The data flow semantics computes \emph{execution maps}
$\cmap \in \cmaps$, which map nodes to values.
We write $\cmap_{\bot}$ ($\cmap_{\cerr}$) for
the execution map that assigns $\bot$ ($\cerr$) to every node.


As a precursor to defining the data flow semantics, we define a
\emph{computational order} $\cvord$ on values. In this order,
$\bot$ is the smallest element, $\cerr$ is the largest element, and
tables are ordered recursively on the pairs of values for each call
site:
\begin{align*}
  \cval_1 \cvord \cval_2 {\Def\iff} &  \cval_1 = \bot \;\lor\; \cval_2 = \cerr
                                       \;\lor\; (\cval_1, \cval_2 \in \consts
                                       \land \cval_1 = \cval_2) \;\lor\;
 (\cval_1,\cval_2 \in \ctables
                                   \land \forall \cstack.\,
                                   \cval_1(\cstack) \mathrel{\vecop{\cvord}}
                                   \cval_2(\cstack))
\end{align*}
Defining the error value as the largest element of the order is
nonessential but simplifies the presentation. This
definition ensures that the least upper bounds (lub) of arbitrary sets of
values exist, denoted by the join operator $\cvjoin$. In 
fact, $\cvord$ is a partial order that induces a complete lattice on values
which can be lifted point-wise to execution maps.

\begin{example}
  \label{ex:dataflow-semantics}
Let us provide intuition for the execution maps through an
example. To this end, consider
Program~\ref{ex:concrete-semantics-example}. \tw{The middle shows the
corresponding expression in our simple language with each
subexpression annotated with its unique location. E.g., the
using occurrence of \lstinline+x+ on line 1 is annotated with location
$k$.

The program's execution map is given on the right.
We abbreviate call stacks occurring in table entries by the last
location pushed onto the stack (e.g. writing just $a$ instead of $adh$).
\zp{This simplification preserves the uniqueness of call stacks for
  this specific program.}
We similarly denote a node just
by its location if this already uniquely identifies the node. During the execution of
Program~\ref{ex:concrete-semantics-example}, the lambda abstraction at
location
$o$ is called twice whereas all other functions are called once. Due
to the two calls to the function (\lstinline+id+) at $o$, the variable
$\texttt{x}$ is bound twice and the subexpression at location
$k$ is also evaluated two times. This is reflected in the execution
map by entries for two distinct variable nodes associated with
$\texttt{x}$ and two distinct expression nodes associated with
$k$. We use superscript notation to indicate the environment
associated with these nodes. For instance,
$k^q$ is the expression node that records the value obtained from the
subexpression at location
$k$ when it is evaluated in the environment binding
$\texttt{x}$ to the variable node $\texttt{x}^q$. In turn,
$\texttt{x}^q$ records the binding of
$\texttt{x}$ for the function call at location $q$. }

\begin{figure}[t]
\begin{minipage}[t]{.27\textwidth}
\begin{flushleft}
{\begin{lstlisting}[language=Caml,mathescape=true,
  columns=fullflexible]
let id x = x${}_{k}$ in
let u = (id${}_{q}$ 1${}_{f}$)${}_{g}$ in
(id${}_{a}$ 2${}_{b}$)${}_{c}$
\end{lstlisting}}
\end{flushleft}
\end{minipage}%
\begin{minipage}[t]{.27\textwidth}
  \small
  \arraycolsep=1pt
   \[\begin{array}{l}
       {\color{red}(}{\color{blue}(}\lambda \texttt{id}.\, {\color{purple}(}\lambda \texttt{u}.\, (\texttt{id}_a\;
2_b)_c{\color{purple})}_d\\[.5ex]
       \qquad \quad\; (\texttt{id}_q\; 1_f)_g{\color{blue})}_h\\[.5ex]
       \quad (\lambda x.\, x_k)_o{\color{red})}_p
     \end{array}\]
\end{minipage}
\begin{minipage}[t]{.45\textwidth}
  \footnotesize
  \begin{flushright}
    \[
      \begin{array}{@{}l@{}}
        h \mapsto [ \tentry{h}{[ \tentry{q}{1}{1},\; \tentry{a}{2}{2}]}{2}]\\[1ex]
          \texttt{id},o \mapsto [ \tentry{q}{1}{1},\; \tentry{a}{2}{2}]
        \quad\; d \mapsto [\tentry{d}{1}{2}]\\[1ex]
          q \mapsto [ \tentry{q}{1}{1} ] \quad\; a \mapsto [
\tentry{a}{2}{2}] \\[1ex]
 \texttt{u},f,g,\texttt{x}^q,k^q \mapsto 1 \quad\; b,c,d,h,\texttt{x}^a,k^a,p \mapsto 2
        \end{array}
    \]
  \end{flushright}
\end{minipage}
\captionof{lstlisting}{The right side shows the program's execution
  map~\label{ex:concrete-semantics-example}}

\small
\[\arraycolsep=10pt
\begin{array}{cccc}
        \begin{array}{@{}l@{}l@{}}
	\mathbf{(0)}  & {}\; 
        \begin{array}{@{}r@{}l@{}}
          \_ \mapsto &\;\; \bot
        \end{array}
	\end{array}
&
        \begin{array}{@{}l@{}l@{}}
	\mathbf{(1)}  & {}\; 
        \begin{array}{@{}r@{}l@{}}
          h \mapsto &\;\; [\tentry{h}{\ctable_\bot}{\bot}]\\
          o \mapsto &\;\; \ctable_\bot
        \end{array}
        \end{array}
&
        \begin{array}{@{}l@{}l@{}}
	\mathbf{(2)}  & {}\; 
        \begin{array}{@{}r@{}l@{}}
          \texttt{id} \mapsto &\;\; \ctable_\bot\\
          d \mapsto &\;\; \ctable_\bot
        \end{array}
        \end{array}
&
        \begin{array}{@{}l@{}l@{}}
	\mathbf{(3)}  & {}\; 
        \begin{array}{@{}r@{}l@{}}
          q \mapsto &\;\; [\tentry{q}{1}{\bot}]\\
          f \mapsto &\;\; 1
        \end{array}
        \end{array}
\end{array}
\]
\noindent\makebox[\linewidth]{\rule{.99\textwidth}{0.2pt}}
\[\arraycolsep=0pt
\begin{array}{cccc}
        \begin{array}{@{}l@{}l@{}}
	\mathbf{(4)}  & {}\; 
        \begin{array}{@{}r@{}l@{}}
          \texttt{id} \mapsto &\;\; [\tentry{q}{1}{\bot}]\\
          h \mapsto &\;\; [\tentry{h}{[\tentry{q}{1}{\bot}]}{\bot}]\\
          o \mapsto &\;\; [\tentry{q}{1}{\bot}]
        \end{array}
	\end{array}\;\;\;{}
&
        \begin{array}{@{}l@{}l@{}}
	\mathbf{(5)}  & {}\; 
        \begin{array}{@{}r@{}l@{}}
          \texttt{x}^q \mapsto &\;\; 1
        \end{array}
	\end{array}\;\;\;{}
&
        \begin{array}{@{}l@{}l@{}}
	\mathbf{(6)}  & {}\; 
        \begin{array}{@{}r@{}l@{}}
          k^q \mapsto &\;\; 1\\
          o \mapsto &\;\; [\tentry{q}{1}{1}]\\
          h \mapsto &\;\; [\tentry{h}{[\tentry{q}{1}{1}]}{\bot}]
        \end{array}
	\end{array}\;\;\;{}
&
        \begin{array}{@{}l@{}l@{}}
	\mathbf{(7)}  & {}\; 
        \begin{array}{@{}r@{}l@{}}
          \texttt{id} \mapsto &\;\; [\tentry{q}{1}{1}]\\
        \end{array}
	\end{array}
\end{array}
\]
\noindent\makebox[\linewidth]{\rule{.99\textwidth}{0.2pt}}
\[\arraycolsep=0pt
\begin{array}{cccc}
        \begin{array}{@{}l@{}l@{}}
	\mathbf{(8)}  & {}\; 
        \begin{array}{@{}r@{}l@{}}
          q \mapsto &\;\; [\tentry{q}{1}{1}]\\
          g \mapsto &\;\; 1\\
          d \mapsto &\;\; [\tentry{d}{1}{\bot}]
        \end{array}
	\end{array}\;\;\;{}
&
        \begin{array}{@{}l@{}l@{}}
	\mathbf{(9)}  & {}\; 
        \begin{array}{@{}r@{}l@{}}
          \texttt{u} \mapsto &\;\; 1\\
          b \mapsto &\;\; 2\\
          a \mapsto &\;\; [\tentry{a}{2}{\bot}]\\
        \end{array}
	\end{array}\;\;\;{}
&
        \begin{array}{@{}l@{}l@{}}
	\mathbf{(10)}  & {}\; 
        \begin{array}{@{}r@{}l@{}}
          \texttt{id} \mapsto &\;\; [\tentry{q}{1}{1},\tentry{a}{2}{\bot}]\\
          h \mapsto &\;\; [\tentry{h}{[\tentry{q}{1}{1},\tentry{a}{2}{\bot}]}{\bot}]\\
          o \mapsto &\;\; [\tentry{q}{1}{1},\tentry{a}{2}{\bot}]          
        \end{array}
	\end{array}\;\;\;{}
&
        \begin{array}{@{}l@{}l@{}}
	\mathbf{(11)}  & {}\; 
        \begin{array}{@{}r@{}l@{}}
          \texttt{x}^a \mapsto &\;\; 2
        \end{array}
	\end{array}
\end{array}
\]
\noindent\makebox[\linewidth]{\rule{.99\textwidth}{0.2pt}}
\[\arraycolsep=0pt
\begin{array}{ccc}
        \begin{array}{@{}l@{}l@{}}
	\mathbf{(12)}  & {}\; 
        \begin{array}{@{}r@{}l@{}}
          k^a \mapsto &\;\; 2\\
          o \mapsto &\;\; [\tentry{q}{1}{1},\tentry{a}{2}{2}]\\
          h \mapsto &\;\; [\tentry{h}{[\dots,\tentry{a}{2}{2}]}{\bot}]
        \end{array}
	\end{array}\quad{}
&
        \begin{array}{@{}l@{}l@{}}
	\mathbf{(13)}  & {}\; 
        \begin{array}{@{}r@{}l@{}}
          \texttt{id} \mapsto &\;\; [\tentry{q}{1}{1},
          \tentry{a}{2}{2}]
        \end{array}
	\end{array}\quad{}
&
        \begin{array}{@{}l@{}l@{}}
	\mathbf{(14)}  & {}\; 
        \begin{array}{@{}r@{}l@{}}
          c,g,p \mapsto &\;\; 2\\
          a \mapsto &\;\; [\tentry{a}{2}{2}]\\
          d \mapsto &\;\; [\tentry{d}{1}{2}]\\
          h \mapsto &\;\; [\tentry{h}{[\dots]}{2}]
        \end{array}
	\end{array}\;\;\;{}
\end{array}
\]
\caption{Fixpoint iterates of the data flow semantics for Program~\ref{ex:concrete-semantics-example}}
\label{fig:running-example-maps}
\vspace{-1em}
\end{figure}

The cumulative entry $\texttt{u},f,g,\texttt{x}^q,k^q \mapsto 1$ in the execution
map indicates that the result computed
at nodes $\texttt{u}$, $f$, $g$, and $k^q$ is $1$, respectively, that $x$ is
bound to $1$ for the call at location $q$. 
Some of the nodes are mapped to function values. The function
\lstinline+id+ is represented by the table
$[ \tentry{q}{1}{1},\; \tentry{a}{2}{2} ]$ that stores input-output
values for the two calls to \lstinline+id+ at $q$ and $a$.  The nodes
corresponding to the two usages of \lstinline+id+ are also mapped to
tables. However, these tables only have a single entry
each. Intuitively, \lstinline+id+ takes two separate data flow paths
in the program starting from its definition at $o$. For each node on
these two paths, the associated table captures how \lstinline+id+ will
be used at the nodes on any data flow path that continues from that
point onward. The tables stored at nodes $q$ and $a$ thus only contain
information about the input and output for the respective call site
whereas the table stored at \lstinline+id+ captures both call sites
because the node for the definition of \lstinline+id+ occurs on both
paths.

Some additional examples involving higher-order functions and recursion can be found
in \lessmore{the companion report~\cite[\S~A]{techreport}}{\S~\ref{sec:appendix-additional-examples}}.
\end{example}

\subsection{Concrete Semantics}

We define the data flow semantics of a higher-order program $e$ formally as the
least fixpoint of a concrete transformer, $\ctransname$, on execution maps.

\smartparagraph{Concrete Transformer.} The idea is that we start with
the map $\cmap_\bot$ and then use $\ctransname$ to consecutively
update the map with new values as more and more nodes are reached
during the execution of $e$. The
signature of $\ctransname$ is as follows:
\[
  \ctransname: \Exp \to \cenvironments \times \cstacks \to \cmaps \to
  \cvalues \times \cmaps
\]
It takes an expression $e$, an environment
$\cenv$, and a stack $\cstack$, and returns a transformer
$\ctrans[k]{e_\elabel}(\cenv,\cstack): \cmaps \to \cvalues \times \cmaps $ on
execution maps.  Given an input map $\cmap$, the transformer
$\ctrans[k]{e_\elabel}(\cenv,\cstack)$ returns the new value $\cval'$ computed at node
$\cn{\elabel}{\cenv}$ together with the updated execution map $\cmap'$. That is, we always have
$\cmap'(\cn{\elabel}{\cenv}) = \cval'$. We
could have defined the transformer so that it returns only the updated map $\cmap'$,
but returning $\cval'$ in addition yields a more concise definition:
observe that $\cmaps \to \cvalues \times \cmaps$ is the type of the
computation of a state monad~\cite{DBLP:conf/lfp/Wadler90}. We exploit
this observation and define $\ctransname$ using
monadic composition of primitive state transformers. This allows us to hide
the stateful nature of the definition and make it easier to see
the connection to the type system later. 

The primitive state transformers and composition operations are
defined in Fig.~\ref{fig:monadic-operations}.  For instance, the
transformer $\mread{\cnode}$ reads the value at node $\cnode$ in the
current execution map $\cmap$ and returns that value together with the
unchanged map. Similarly, the transformer $\mupd{\cnode}{\cval}$
updates the entry at node $\cnode$ in the current map $\cmap$ by
taking the join of the current value at $\cnode$ with $\cval$,
returning the obtained new value $\cval'$ and the updated map. We compress a
sequence of update operations $\mupd{\cnode_1}{\cval_1}, \dots, \mupd{\cnode_n}{\cval_n}$,
by using the shorter notation $\mupd{\cnode_1, \dots, \cnode_n}{\cval_1, \dots, \cval_n}$
to reduce clutter. We point out that the result of this sequenced
update operation is the result of the last update $\mupd{\cnode_n}{\cval_n}$.

The operation $\mbind{F}{G}$ defines the usual composition of stateful
computations $F$ and $G$ in the state monad, where
$F \in \cmaps \to \alpha \times \cmaps$ and
$G \in \alpha \to \cmaps \to \beta \times \cmaps$ for some $\alpha$
and $\beta$. Note that the composition is short-circuiting in the case where
the intermediate value $u$ produced by $F$ is $\top$ (i.e. an error
occurred).
We use Haskell-style monad comprehension syntax for applications of
$\kw{bind}$ at the meta-level. That is, we write
$\Mdo x \mto F \Mend G$ for $\mbind{F}{\Lambda x.\,G}$. We similarly
write $\Mdo x \mto F \,\Mif P \Mend G$ for $\mbind{F}{\Lambda x.\, \Mif
  P \,\Mthen G \, \Melse \Mreturn \bot}$ and we shorten
$\Mdo x \mto F \Mend G$ to just $\Mdo F \Mend G$ in cases where $x$ does not occur
free in $G$. Moreover, we write
$\Mdo x_1 \mto F_1 \Mend \dots \Mend x_n \mto F_n \Mend G$ for the
comprehension sequence
$\Mdo x_1 \mto F_1 \Mend (\dots \Mend (\Mdo x_n \mto F_n \Mend
G)\dots)$. We also freely mix the monad comprehension syntax with
standard let binding syntax and omit the semicolons whenever this
causes no confusion.


The definition of $\ctrans[k]{e}(\cenv,\cstack)$ is given in
Fig.~\ref{fig:concrete-transformer} using induction over the structure
of $e$. As we discussed earlier, the structural definition of the
transformer enables an easier formal connection to the data flow
refinement typing rules. Note that in the definition we implicitly
assume that $\pair{e}{\cenv}$ is well-formed. We discuss the cases of
the definition one at a time \tw{using
  Program~\ref{ex:concrete-semantics-example} as our running
  example. Figure~\ref{fig:running-example-maps} shows the fixpoint
  iterates of $\ctrans[k]{e}(\epsilon,\epsilon)$ starting from the
  execution map $\cmap_\bot$ where $e$
  is Program~\ref{ex:concrete-semantics-example} and $\epsilon$ refers
  to both the empty environment and empty stack. For each
  iterate $(i)$, we only show the entries in the execution map that
  change in that iteration. We will refer to this figure throughout our
  discussion below.}

\begin{figure}
  \begin{align*}
    \mread{\cnode} \Def= {} & \Lambda \cmap.\, \pair{\cmap(\cnode)}{\cmap}\\
    \mupd{\cnode}{\cval} \Def= {} & \Lambda \cmap.\, \Mlet \cval'
= \cmap(\cnode) \cvjoin \cval \,\Min \Mif \safe(\cval') \,\Mthen
\pair{\cval'}{\cmap[\cnode \mapsto \cval']} \,\Melse \pair{\cerr}{\cmap_\top}\\
    \menv(\cenv) \Def= {} & \Lambda \cmap.\, \pair{\cmap \circ \cenv}{\cmap}\\
    \massert{P} \Def= {} & \Lambda \cmap.\, \Mif P
\,\Mthen  \pair{\bot}{\cmap} \, \Melse \pair{\cerr}{\cmap_\top}\\
    \Mfor \ctable \, \Mdo F \Def= {} & \Lambda 
\cmap.\, \cmjoinb_{\cstack \in \ctable} F(\cstack)(\cmap)\\
   \Mreturn \cval \, \Def= {} & \Lambda \cmap.\, \pair{\cval}{\cmap}\\
   \mbind{F}{G} \Def= {} & \Lambda \cmap.\, \Mlet
\pair{u}{{\cmap}'} = F(\cmap) \, \Min \, \Mif u = \top \, \Mthen
\pair{\cerr}{\cmap_\top} \, \Melse G(u)(\cmap')
%
%
  \end{align*}
  \vspace{-1em}
  \caption{Primitive transformers on state monad for computations over
    execution maps}
  \label{fig:monadic-operations}
\[\begin{array}[t]{@{}l|l}
\begin{array}[t]{@{}l}
\ctrans{c_\elabel}(\cenv, \cstack) \Def= \\
\sind \Mdo \cnode = \cn{\elabel}{\cenv} \,\Mend\;
\cval' \mto \mupd{\cnode}{c} \,\Mend\; \Mreturn \cval'\\[2ex]
\ctrans{x_\elabel}(\cenv, \cstack) \Def= \\
\sind \Mdo  \cnode = \cn{\elabel}{\cenv} \, \Mend \cval \mto
\mread{\cnode} \, \Mend \, \cnode_x = \cenv(x)
\, \Mend \, \Gamma \mto \menv(\cenv) \, \\[0.4ex]
\sind[6] \cval' \mto \mupd{\cnode_x,\,
\cnode}{\Gamma(x) \ciprop \cval}\\
\sind \Mreturn \cval'\\[2ex]
%
\ctrans{(e_\elabela \, e_\elabelb)_\elabel}(\cenv, \cstack) \Def=  \\
\sind 	\Mdo
\arraycolsep=0pt
\begin{array}[t]{l}
\cnode = \cn{\elabel}{\cenv} \,\Mend\; \cnode_\elabela =
\cn{\elabela}{\cenv}
		\,\Mend\; \cnode_\elabelb = \cn{\elabelb}{\cenv} \,
\Mend\; \cval \mto \mread{\cnode}\\[0.4ex]
\cval_\elabela \mto \ctrans[k]{e_\elabela}(\cenv, \cstack) \,\Mif \cval_\elabela \neq \bot\\[0.4ex]
\massert{\cval_\elabela \in \ctables} \\[0.4ex]
\cval_\elabelb \mto \ctrans[k]{e_\elabelb}(\cenv, \cstack)\\[0.4ex]
\cval'_\elabela,\, [\tentry{\elabela \cdot \cstack}{\cval_\elabelb'}{\cval'}] = \cval_\elabela \ciprop
   [\tentry{\elabela \cdot \cstack}{\cval_\elabelb}{\cval}]\\[0.4ex]
\cval''	\mto \mupd{\cnode_\elabela,\, \cnode_\elabelb,\,
\cnode}{\cval'_\elabela,\, \cval'_\elabelb,\, \cval'}
\end{array}\\[.4ex]
\sind \Mreturn \cval''
\end{array}
&
\begin{array}[t]{@{}l}
\ctrans{(\lambda x.e_\elabela)_\elabel}(\cenv, \cstack) \Def= \\ 
\sind	\Mdo \cnode = \cn{\elabel}{\cenv} \,\Mend\; \ctable \mto
\mupd{\cnode}{\ctable_{\bot}} \\[0.5ex]
\sind[6] \ctable' \mto \Mfor \ctable \;\Mdo \cstepbody(x,
e_\elabela, \cenv, \ctable)\\[0.5ex]
\sind[6] \ctable'' \mto \mupd{\cnode}{\ctable'}\\[0.5ex] 
\sind[2] \Mreturn \ctable'' \\[2ex]
\cstepbody(x, e_\elabela, \cenv, \ctable)(\cstack') \Def= \\
\sind  \Mdo
\arraycolsep=0pt
\begin{array}[t]{l}
\cnode_x = \cns{x}{\cenv}{\cstack'} \,\Mend\;
                  \cenv_\elabela = \cenv.x\!:\!\cnode_x \,\Mend\;
\cnode_\elabela = \cn{\elabela}{\cenv_\elabela}\\[0.5ex]
\cval_x \mto \mread{\cnode_x} \\[0.5ex]
\cval_\elabela \mto \ctrans[k+1]{e_\elabela}(\cenv_\elabela,
\cstack') \\[0.5ex]
[\ttab[\cstack']{\cval_x'}{\cval_\elabela'}], \ctable' = [\ttab[\cstack']{\cval_x}{\cval_\elabela}] \ciprop \restr{\ctable}{\cstack'} \\[0.5ex]
\mupd{\cnode_x,\, \cnode_\elabela}{\cval'_x,\, \cval'_\elabela}
\end{array}\\[0.5ex]
\sind  \Mreturn \ctable'
\end{array}
\end{array}\]
\vspace{-1em}
\caption{Transformer for the concrete data flow semantics\label{fig:concrete-transformer}}
%
\begin{minipage}[t]{0.55\textwidth}
\[\begin{array}{r l}
\ctable_1 \ciprop \ctable_2 \Def=& \\ 
\multicolumn{2}{l}{\sind \Mlet \ctable' = \Lambda \cstack. } \\[0.5ex]
\multicolumn{2}{l}{\sind[2]\Mif \cstack \not\in \ctable_2 \;\Mthen \pair{\ctable_1(\cstack)}{\ctable_2(\cstack)} \;\Melse} \\[0.5ex]
\multicolumn{2}{l}{\sind[2]\Mlet \pair{\cval_{1i}}{\cval_{1o}} = \ctable_1(\cstack) \Mend \pair{\cval_{2i}}{\cval_{2o}} = \ctable_2(\cstack)} \\ [0.5ex]
\multicolumn{2}{l}{\sind[3]\;\;\pair{\cval'_{2i}}{\cval'_{1i}} = \cval_{2i} \ciprop \cval_{1i} \Mend \pair{\cval'_{1o}}{\cval'_{2o}} = \cval_{1o} \ciprop \cval_{2o}} \\[0.5ex]
\multicolumn{2}{l}{\sind[2] \Min (\pair{\cval'_{1i}}{\cval'_{1o}}, \pair{\cval'_{2i}}{\cval'_{2o}}) \quad\quad\quad \ind[4]} \\[0.5ex]
\multicolumn{2}{l}{\sind \Min \pair{\Lambda \cstack.\; \pi_1(\ctable'(\cstack))}{\Lambda \cstack.\; \pi_2(\ctable'(\cstack))}} 
\end{array}\]
\end{minipage}%
\begin{minipage}[t]{0.44\textwidth}
\begin{align*}
\ctable \ciprop \bot \Def= {} &  \pair{\ctable}{\ctable_{\bot}}\\[5.5ex]
\ctable \ciprop \cerr \Def= {} &  \pair{\cerr}{\cerr} \\[5.5ex]
\cval_1 \ciprop \cval_2 \Def= {} &  \pair{\cval_1}{\cval_1 \cvjoin
\cval_2} \quad \textbf{(otherwise)}
\end{align*}
\end{minipage}
\caption{Value propagation in the concrete data flow semantics\label{fig:concrete-prop}}
\vspace{-.8em}
\end{figure}



\smartparagraphnb{Constant $e = c^\elabel$.} Here, we simply set the current node
$\cnode$ to the join of its current value $\cmap(\cnode)$ and
the value $c$. \tw{For example, in
  Fig.~\ref{fig:running-example-maps}, when execution reaches the
  subexpression at location $f$ in iteration $(2)$, the corresponding
  entry for the node $n$ identified by $f$ is updated to
  $\cmap(n) \cvjoin 1 = \bot \cvjoin 1 = 1$.}

\smartparagraphnb{Variable $e = x_\elabel$.} This case implements the data flow propagation between the
variable node $\cnode_x$ binding $x$ and the current expression node
$\cnode$ where $x$ is used. This is realized using the
\emph{propagation function} $\ciprop$. Let
$v_x = \Gamma(x)=\cmap(\cnode_x)$ and $v = \cmap(\cnode)$ be the
current values stored at the two nodes in $\cmap$. The function
$\ciprop$ takes these values as input and propagates information
between them, returning two new values $v_x'$ and $v'$ which are then
stored back into $\cmap$. The propagation function is defined in
Fig.~\ref{fig:concrete-prop} and works as follows. If $v_x$ is a
constant or the error value and $v$ is still $\bot$, then we simply
propagate $v_x$ forward, replacing $v$ and leaving $v_x$
unchanged. The interesting cases are when we propagate information
between tables. The idea is that inputs in a table $v$ flow backward
to $v_x$ whereas outputs for these inputs flow forward from $v_x$ to
$v$. \tw{For example, consider the evaluation of the occurrence of
  variable $\texttt{id}$ at location $a$ in step $(10)$ of
  Fig.~\ref{fig:running-example-maps}. Here, the expression node $\cnode$ is
  identified by $a$ and the variable node $\cnode_x$ by
  $\texttt{id}$. Moreover, we have $v=[\tentry{a}{2}{\bot}]$
  from step $(9)$ and $v_x=[\tentry{q}{1}{1}]$ from step $(7)$. We
  then obtain
  \[v_x \tiprop v = [\tentry{q}{1}{1}] \tiprop [\tentry{a}{2}{\bot}] =
    \pair{[\tentry{q}{1}{1},
      \tentry{a}{2}{\bot}]}{[\tentry{a}{2}{\bot}]}\] That is, the
  propagation causes the entry in the execution map for the node
  identified by $\texttt{id}$ to be updated to
  $[\tentry{q}{1}{1}, \tentry{a}{2}{\bot}]$.  }

In general, if $v_x$
is a table but $v$ is still $\bot$, we initialize $v$ to the empty
table $\ctable_\bot$ and leave $v_x$ unchanged (because we have not
yet accumulated any inputs in $v$). If both $v_x$ and $v$ are tables,
$v_x \tiprop v$ propagates inputs and outputs as described above by
calling $\ciprop$ recursively for every call site $\cstack \in
\cval$. Note how the recursive call for the propagation of the inputs
$v_{2i} \ciprop v_{1i}$ inverts the direction of the
propagation. \tw{This has the affect that information about argument
  values propagate from the call site to the definition site of the
  function being called, as is the case for the input value $2$ at
  call site $a$ in our example above. Conversely, output values are
  propagated in the other direction from function definition sites to
  call sites. For example, in step $(14)$ of
  Fig.~\ref{fig:running-example-maps}, the occurrence of $\texttt{id}$
  at location $a$ is evaluated again, but now we have
  $v_x=[\tentry{q}{1}{1},\tentry{a}{2}{2}]$ from step $(13)$ whereas
  $v$ is as before. In this case, the propagation yields
  \[v_x \tiprop v = [\tentry{q}{1}{1},\tentry{a}{2}{2}] \tiprop [\tentry{a}{2}{\bot}] =
    \pair{[\tentry{q}{1}{1}, \tentry{a}{2}{2}]}{[\tentry{a}{2}{2}]}\]
  That is, the information about the output value $2$ for the input
  value $2$ has finally arrived at the call site $a$.
As we shall see, the dataflow propagation between tables closely
relates to contravariant subtyping of function types.}

\smartparagraphnb{Function application $e = (e_\elabela \, e_\elabelb)_\elabel$.}We 
first evaluate $e_\elabela$ to obtain the updated map and extract
the new value $v_\elabela$ stored at the corresponding expression node $\cnode_\elabela$.
If $v_\elabela$ is not a table, then we must be attempting
an unsafe call, in which case the monadic operations return the error map $\cmap_{\cerr}$. 
If $v_\elabela$ is a table, we continue evaluation of $e_\elabelb$ obtaining the new value $v_\elabelb$
at the associated expression node $\cnode_\elabelb$.
We next need to propagate the information between this call site $\cnode_\elabela$ and $v_\elabela$. To this
end, we use the return value $v$ for the node $\cnode$ of $e$
computed thus far and create a singleton table 
$[\cstack':v_\elabelb \rightarrow v]$ where $\cstack'= \elabela \cdot
\cstack$ is the extended call stack that will be used for the
evaluation of this function
call. We then propagate between $v_\elabela$
and this table to obtain the new value $v'_\elabela$ and table $\ctable'$. Note that this propagation
boils down to (1) the propagation between $v_\elabelb$ and the input of $v_\elabela$ at 
$\cstack'$ and (2) the propagation between the output of $v_\elabela$ at
$\cstack'$ and the return value $v$.
The updated table $\ctable'$ hence contains
the updated input and output values $v'_\elabelb$ and $v'$ at $\cstack'$. All of these
values are stored back into the execution map. Intuitively, $v'_\elabela$ contains the
information that the corresponding function received an input coming 
from the call site identified by $\cstack'$. This information is ultimately
propagated back to the function definition where the call is actually evaluated.

\tw{As an example, consider the evaluation of the function application
  at location $g$ in step $(3)$ of
  Fig.~\ref{fig:running-example-maps}. That is, node $\cnode$ is
  identified by $g$ and we have $\elabela=q$, 
  $\elabelb=f$, $S=h$, and $v=\bot$. Moreover, we initially have 
  $\cmap(\cnode_\elabela)=\bot$ and $\cmap(\cnode_\elabelb)=\bot$. The recursive evaluation of
  $e_\elabela$ will propagate the information that $\texttt{id}$ is a
  table to location $q$, i.e., the recursive call to $\ctransname$
  returns $\cval_\elabela = \ctable_\bot$. Since $\cval_\elabela$ is a
  table, we proceed with the recursive evaluation of $1_f$,
  after which we obtain $\cmap(\cnode_f)=\cval_\elabelb=1$. Next we
  compute
  \[\cval_\elabela \ciprop
    [\tentry{\elabela \cdot \cstack}{\cval_\elabelb}{\cval}]
    = \ctable_\bot \ciprop [\tentry{qh}{1}{\bot}]
    = \pair{[\tentry{qh}{1}{\bot}]}{[\tentry{qh}{1}{\bot}]}
  \]
  That is, $\cmap(\cnode_\elabela)$ is updated to
  $\cval_\elabela'=[\tentry{qh}{1}{\bot}]$\footnote{Recall that in
    Fig.~\ref{fig:running-example-maps} we abbreviate the call stack
    $qh$ by just $q$.}. Note that we still have $\cval'=\bot$ at this
  point. The final value $1$ at location $g$ is obtained later when
  this subexpression is reevaluated in step $(8)$.  
}

\smartparagraphnb{Lambda abstraction $e = (\lambda
  x.e_\elabela)_\elabel$.}
We first extract the table $\ctable$ computed for the function thus
far. Then, for every call stack $\cstack'$ for which an input has
already been back-propagated to $\ctable$, we analyze the body by evaluating
$\cstepbody(x, e_\elabela, \cenv, \ctable)(\cstack')$, as follows.
First, we create a variable node $\cnode_x$ that will store the input value
that was recorded for the call site stack $\cstack'$. Note that if this value
is a table, indicating that the function being evaluated was called
with another function as input, then any inputs to the argument
function that will be seen while evaluating the body $e_\elabela$ will be
back-propagated to the table stored at $n_x$ and then further to
the node generating the call stack $\cstack'$. By incorporating
$\cstack'$ into variable nodes, we guarantee that there is a unique node for each call.

We next evaluate the function body $e_\elabela$ with the input
associated with stack $\cstack'$. To this end, we extend the
environment $\cenv$ to $\cenv_\elabela$ by binding $x$ to the variable
node $\cnode_x$.
We then propagate the information between the values
stored at the node $\cnode_\elabela$, i.e., the result of evaluating the
body $e_\elabela$, the nodes $\cnode_x$ for the bound variable, and the table $\ctable$. That is, we propagate
information from (1) the input of $\ctable$ at $\cstack'$ and the
value at node $\cnode_x$, and (2) the value assigned to the function body under 
the updated environment $\cenv_\elabela$ and the output of $\ctable$
at $\cstack'$
Finally, the updated tables $\ctable'$ for all call site stacks $\cstack'$ are joined
together and stored back at node $\cnode$.

\tw{
As an example, consider the evaluation of the lambda abstraction at
location $o$ in step $(5)$ of
Fig.~\ref{fig:running-example-maps}. Here, $\cnode$ is identified by
$o$ and we initially have
$\cmap(\cnode)=\ctable=[\tentry{qh}{1}{\bot}]$. Thus, we evaluate a single call to
$\cstepbody$ for $e_\elabela=\texttt{x}_k$ and $\cstack'=qh$. In this
call we initially have $\cmap(\cnode_x)=\cmap(\texttt{x}^q)=\bot$. Hence, the
recursive evaluation of $e_\elabela$ does not yet have any effect and
we still obtain $\cval_\elabela=\bot$ at this point. However, the
final propagation step in $\cstepbody$ yields:
\[[\ttab[\cstack']{\cval_x}{\cval_\elabela}] \ciprop
  \ctable(\cstack') = [\ttab[qh]{\bot}{\bot}] \ciprop
  [\ttab[qh]{1}{\bot}] = \pair{[\ttab[qh]{1}{\bot}]}{[\ttab[qh]{1}{\bot}]}
\]
and we then update $M(\cnode_x)$ to $\cval_x'=1$. In step $(6)$, when
the lambda abstraction at $o$ is once more evaluated, we now have
initially $\cmap(\cnode_x)=\cmap(\texttt{x}^q)=1$ and the recursive
evaluation of $e_i=\texttt{x}_k$ will update the entry for $k^q$ in
the execution map to $1$. Thus, we also obtain $\cval_\elabela=1$. The
final propagation step in $\cstepbody$ now yields:
\[[\ttab[\cstack']{\cval_x}{\cval_\elabela}] \ciprop
  \ctable(\cstack') = [\ttab[qh]{1}{1}] \ciprop
  [\ttab[qh]{1}{\bot}] = \pair{[\ttab[qh]{1}{1}]}{[\ttab[qh]{1}{1}]}
\]
which will cause the execution map entry for $\cnode$ (identified by
$o$) to be updated to $[\ttab[qh]{1}{1}]$.

Observe that the evaluation of a lambda abstraction for a new input
value always takes at least two iterations of $\ctransname$. This can
be optimized by performing the propagation in $\cstepbody$ both before and after the
recursive evaluation of $e_i$. However, we omit this optimization here
for the sake of maintaining a closer resemblance to the typing rule
for lambda abstractions.}

\begin{lemma}
  \label{prop-monotone}
  The function $\ciprop$ is monotone and increasing.
\end{lemma}

\begin{lemma}
  \label{step-monotone}
  For every  $e \in \Exp$, $\cenv \in \cenvironments$, and $\cstack
  \in \cstacks$ such that $\pair{e}{\cenv}$ is well-formed,
  $\ctrans{e}(\cenv,\cstack)$ is monotone and increasing.
\end{lemma}

We define the semantics $\csem \llbracket e \rrbracket$ of a program $e$
as the least fixpoint of $\ctransname$ over the complete lattice of
execution maps:
\[
  \csem \llbracket e \rrbracket \Def= \mathbf{lfp}^{\cmord}_{
     \cmap_{\bot}} \Lambda \cmap.\, \Mlet \pair{\_}{\cmap'} = \ctrans[k]{e}(\epsilon,\epsilon)(\cmap) \;\Min \cmap'
\]

Lemma~\ref{step-monotone} guarantees that
$\csem \llbracket e \rrbracket$ is well-defined.
We note that the above semantics does not precisely model
  certain non-termina\-ting
programs where tables grow infinitely deep. The semantics of such programs
is simply $\cmap_{\cerr}$. A more precise semantics 
can be defined using
step-indexing~\cite{DBLP:journals/toplas/AppelM01}. However, we omit
this for ease of presentation and note that our semantics is adequate
for capturing refinement type systems and inference algorithms
\emph{\`a la} Liquid types, which do not support
infinitely nested function types.
\smartparagraph{Properties and collecting semantics.}
As we shall see, data flow refinement types abstract programs
  by \emph{properties} $P \in \cdomain$, which are sets of execution
  maps: $\cdomain \Def= \powerset(\cmaps)$. Properties form the
  concrete lattice of our abstract interpretation and are ordered by
  subset inclusion. That is, the concrete semantics of our abstract
  interpretation is the \emph{collecting semantics}
  $\ccoll: \Exp \to \cdomain$ that maps programs to properties:
  $\ccoll \llbracket e \rrbracket \Def= \{\csem \llbracket e \rrbracket\}.$
An example of a property is \emph{safety}: let
$\psafe$ be the property consisting of all execution maps that map all
nodes to safe values. Then a
program $e$ is \emph{safe} if $\ccoll \llbracket e \rrbracket \subseteq \psafe$.

\section{Intermediate Abstract Semantics}

We next present two abstract semantics that represent crucial
abstraction steps in calculationally obtaining our data flow refinement
type system from the concrete data flow semantics. Our formal exposition focuses
mostly on the aspects of these semantics that are instrumental in
understanding the loss of precision introduced by these intermediate
abstractions. Other technical details can be found in the designated
appendices included in the supplementary materials.

\subsection{Relational Semantics}
\label{sec:arelational}
%
%
As we shall see later, a critical abstraction step performed by the
type system is to conflate the information in tables that are propagated
back from different call sites to function definitions. That is, a pair of
input/output values that has been collected from one call site will
also be considered as a possible input/output pair at other call sites
to which that function flows. To circumvent a catastrophic loss of
precision caused by this abstraction step, we first introduce an
intermediate semantics that explicitly captures the relational
dependencies between input and output values of functions. 
We refer to this semantics as the \textit{relational data flow semantics}.

\smartparagraph{Abstract domains.}
The definitions of nodes, stacks, and environments in the relational
semantics are the same as in the concrete semantics. Similar to
data flow refinement types, the relational
abstractions of values, constants, and tables are defined with respect
to \emph{scopes} $\rscope \subseteq_{\mathit{fin}} \vars$. The variables in scopes are
used to track how a value computed at a specific node in the execution
map relates to the other values bound in the current environment. We
also use scopes to capture how function output values depend on
the input values, similar to the way input-output relations are
captured in dependent function types. The scope of a node $\cnode$, denoted
$\rscope_{\cnode}$, is the domain of $\cnode$'s environment:
$\rscope_{\cnode} \Def= \dom(\nenv(\cnode))$.  The new semantic domains are
defined as follows:
\begin{align*}
\rval \in \rvalues & ::= {} \; \rbot \mid \rerr \mid \rel \mid x:\rtable &
\dmap \in \dmaps_\rscope & \Def= {} \; \rscope \cup \set{\nu} \to \consts \cup \{\rdf\} \\
\rel \in \rels & \Def= {}\; \powerset(\dmaps_{\rscope}) &
\absn:\rtable \in \rtables & \Def= {} \; \Sigma \absn \in (\vars
\setminus \rscope).\,
\cstacks \to \rvalues[\rscope] \times \rvalues[\rscope \cup \{\absn\}]\\
&&\rmap \in \rmaps & \Def= \; \Pi{\cnode \in \cnodes}.\;\rvalues[\rscope_{\cnode}]
\end{align*}
Relational values $\rval \in \rvalues$ model how concrete values, stored at
some node, depend on the concrete values of nodes in the current scope
$\rscope$. The relational value $\rbot$ again models nontermination or
unreachability and imposes no constraints on the values in its
scope. The relational value $\rerr$ models every possible concrete
value, including $\cerr$. Concrete constant values are abstracted by
relations $\rel$, which are sets of \emph{dependency vectors}
$\dmap$. A dependency vector tracks the dependency between a
constant value associated with the special
symbol $\nu$, and the values bound to the variables in
scope $\rscope$. Here, we assume that $\nu$ is never contained in $\rscope$.

We only track dependencies between constant values precisely: if a
node in the scope stores a table, we abstract it by the symbol $\rdf$
which stands for an arbitrary concrete table. We assume $\rdf$ to be
different from all other constants $\consts$. We also require that for all
$\dmap \in \rel$, $\dmap(\nu) \in \consts$.
Relational tables $\pair{\absn}{\rtable}$, denoted $\absn:\rtable$,
are defined analogously to concrete tables except that $\rtable$ now
maps call site stacks $\cstack$ to pairs of relational values
$\pair{\rval_i}{\rval_o}$. As for dependent function types, we add a
dependency variable $\absn$ to every relational table to track
input-output dependencies.
Note that we consider relational tables to be equal up to
$\alpha$-renaming of dependency variables.
%
Relational execution maps $\rmap$ assign each node $\cnode$ a relational
value with scope $\rscope_\cnode$. 


The relational semantics of a program is the relational execution map
obtained as the least fixpoint of a Galois abstraction of the concrete
transformer $\ctransname$. The formal definition is mostly straightforward, so we delegate it to 
\lessmore{the companion report~\cite[\S~B]{techreport}}{\S~\ref{appendix:relational-semantics}}.


\begin{example}
  \label{ex:relational-semantics}
  The relational execution map obtained for
  Program~\ref{ex:concrete-semantics-example} is as follows (we only
  show the entries for the nodes $\texttt{id}$, $q$, and $a$):
\begin{small}
\begin{align*}
\texttt{id} \mapsto {} & \absn:[ \tentry{q}{\{(\nu\!:1)\}}{\{(\absn\!:1, \nu\!:1)\}},\;
			 \tentry{a}{\{(\nu\!:2)\}}{\{(\absn\!:2, \nu\!:2)\}} ] \\
q \mapsto {} & \absn:[ \tentry{q}{\{(\texttt{id}\!:\rdf, \nu\!:1)\}}{\{(\texttt{id}\!:\rdf, \absn\!:1, \nu\!:1)\}}] \\
a \mapsto {} & \absn:[\tentry{a}{\{(\texttt{id}\!:\rdf, \texttt{u}\!:1, \nu\!:2)\}}{\{(\texttt{id}\!:\rdf, \texttt{u}\!:1, \absn\!:2, \nu\!:2)\}} ]
\end{align*}
\end{small}%
Each concrete value $\cval$ in the concrete execution map shown to the
right of Program~\ref{ex:concrete-semantics-example} is abstracted by
a relational value that relates $\cval$ with the values bound to the
variables that are in the scope of the node where $\cval$ was
observed.  Consider the entry for node \lstinline+id+. As expected,
this entry is a table that has seen inputs at call site stacks
identified with $q$ and $a$. The actual input stored for call site
stack $q$ is now a relation consisting of the single row $(\nu:1)$,
and similarly for call site stack $a$. As the node \lstinline+id+ has
no other variables in its scope, these input relations are simply
representing the original concrete input values $1$ and $2$.  We
associate these original values with the symbol $\nu$. The output
relation for $q$ consists of the single row $(\absn:1, \nu:1)$,
stating that for the input value $1$ (associated with $\absn$), the
output value is also $1$ (associated with $\nu$). Observe how we use
$\absn$ to capture explicitly the dependencies between the input and
output values.

The entry in the relational execution map for the node identified by
$q$ is similar to the one for \lstinline+id+, except that
the relation also has an additional entry $\texttt{id}:\rdf$.  This is
because \lstinline+id+ is in the scope of $q$. The value of
$\texttt{id}$ in the execution map is a table, which the relational
values abstract by the symbolic value $\rdf$. That is, the
relational semantics only tracks relational dependencies between
primitive values precisely whereas function values are abstracted by
$\rdf$. The relational table stored at node $a$ is similar, except
that we now also have the variable $\texttt{u}$ which is bound to value $1$. As
in the concrete execution map, the relational tables for $q$
and $a$ contain fewer entries than the table stored at
\lstinline+id+.
\end{example}

\smartparagraph{Abstraction.}
We formalize the meaning of relational execution maps in terms of a
Galois connection between $\rmaps$ and the complete lattice of sets of
concrete execution maps $\powerset(\cmaps)$.
The details of this construction and the resulting abstract
transformer of the relational semantics can be found
in~\lessmore{\cite[\S~B]{techreport}}{\S~\ref{appendix:relational-semantics}}.
%
We here focus on the key idea of the abstraction by formalizing the
intuitive meaning of relational values given above. Our formalization
uses a family of concretization functions 
$\rgamma[\rscope]: \rvalues \to \powerset((\rscope \to
\cvalues) \times \cvalues)$,
parameterized by scopes $\rscope$, that map relational values to
sets of pairs $\pair{\scmap}{\cval}$ where $\scmap$ maps the
variables in scope $\rscope$ to values in $\cvalues$:
\begin{align*}
\rgamma[\rscope](\rbot) \Def=\; & (\rscope \to \cvalues) \times \set{\bot} \ind[7] \quad \rgamma[\rscope](\rerr) \Def=\; (\rscope \to \cvalues) \times \cvalues \\
\rgamma[\rscope](\rel) \Def=\; & \setc{\pair{\scmap}{c}}{\dmap {\in} \rel \wedge \dmap(\nu) = c \land  
 \forall x {\in} \rscope.\; \scmap(x) \in \dgamma(\dmap(x))} \cup \rgamma[\rscope](\rbot)\\
\rgamma[\rscope](\absn: \rtable) \Def=\; & \{ \pair{\scmap}{\ctable}\;|\; \forall \cstack.\;
  \ctable(\cstack) = \pair{\cval_i}{\cval_o} \wedge \rtable(\cstack) = \pair{\rval_i}{\rval_o} 
  \wedge \\ 
  &\ind[4] \pair{\scmap}{\cval_i} \in \rgamma[X](\rval_i) \wedge \pair{\scmap[\absn \mapsto \cval_i]}{\cval_o} \in \rgamma[X \cup \{\absn\}](\rval_o)\} \cup \rgamma[\rscope](\rbot)
\end{align*}
Here, the function $\dgamma$, which we use to give meaning to dependency
relations, is defined by $\dgamma(c) = \set{c}$ and $\dgamma(\rdf)=\ctables$.
The meaning of relational execution maps is then given by the
function
\[\rgamma[](\rmap) \Def= \setc{\cmap}{\forall \cnode \in \cnodes.\,
    \pair{\scmap}{\cval} \in \rgamma[\rscope_\cnode](\rmap(\cnode))
    \land \scmap = \cmap \circ \nenv(\cnode) \land \cval =
    \cmap(\cnode)} \enspace.\]


\subsection{Collapsed Semantics}
\label{sec:acollapsed}

We now describe a key abstraction step in the construction of our
data flow refinement type semantics. We formalize this step in terms of
a \textit{collapsed (relational data flow) semantics}, which collapses function tables to a bounded number of entries while controlling how much stack information is being
lost, thereby allowing for different notions of call-site context
sensitivity.

\smartparagraph{Abstract domains.}
The collapsed semantics is parameterized by a finite set of
\emph{abstract stacks} $\pstack \in \pstacks$, a \emph{stack
  abstraction} function $\rho: \cstacks \to \pstacks$, and an \emph{abstract
stack concatenation} operation $\pconcat : \elabels \times \pstacks
\to \pstacks$. Stack abstraction must be homomorphic with
respect to concatenation: for all $\elabel \in \elabels$ and $\cstack
\in \cstacks$, $\rho(\elabel \cdot \cstack) = \elabel \pconcat \rho(\cstack)$.

Abstract stacks induce sets of abstract nodes and abstract
environments following the same structure as in the concrete semantics
\begin{align*}
\pnode \in \pnodes \Def= {} & \penodes \cup \pvnodes &&&
\penodes \Def= {} & \locs \times \penvironments &&&
\pvnodes \Def= {} & \vars \times \penvironments \times \pstacks &&&
\penv \in \penvironments \Def= {} & \vars \pto_{\mi{fin}} \pvnodes 
\end{align*}
We lift $\rho$ from stacks to nodes and environments in the expected
way. In particular, for variable nodes
$\cnode_x = \cns{\elabel}{\cenv}{\cstack}$, we recursively define
$\rho(\cnode_x) \Def= \cns{\elabel}{(\rho \circ
  \cenv)}{\rho(\cstack)}$. Analogous to concrete nodes, we define the
scope of an abstract node as $\rscope_{\pnode} = \dom(\nenv(\pnode))$.

The definition of values and execution maps remains largely
unchanged. In particular, dependency vectors and relational values are
inherited from the relational semantics. Only the definition of tables
changes, which now range over abstract stacks rather than concrete stacks:
\[\begin{array}{c}
\pval \in \pvalues[\rscope] ::=  \pbot \mid \perr \mid \rel \mid x:\ptable \ind 
x:\ptable \in \ptables[\rscope] \Def= \Sigma \absn \in (\vars
\setminus \pscope).\, \pstacks
\to \pvalues[\pscope] \times \pvalues[\rscope \cup \{\absn\}] \\[1.2ex]
\pmap \in \pmaps \Def= \Pi{\pnode \in \pnodes}.\;\pvalues[\pscope_{\pnode}]
\end{array}\]
%


\begin{example}
\label{ex:collapsed-semantics}
We use Program~\ref{ex:concrete-semantics-example} again to provide
intuition for the new semantics. To this end, we first define a family
of sets of abstract stacks which we can use to instantiate the
collapsed semantics. Let $k \in \NN$ and define
$\pstacks^k = \elabels^k$ where $\elabels^k$ is the set of all
sequences of locations of length at most $k$. Moreover, for
$\pstack \in \pstacks^k$ define
$\ell \pconcat \pstack = (\ell \cdot \pstack)[0,k]$ where
$(\ell_1\dots\ell_n)[0,k]$ is $\epsilon$ if $k=0$, $\ell_1\dots\ell_k$
if $0 < k < n$, and $\ell_1\dots\ell_n$ otherwise. Note that this
definition generalizes the definitions of $\pstacks^0$ and
$\pstacks^1$ from Example~\ref{ex:lia-types}. An abstract stack in
$\pstacks^k$ only maintains the return locations of the $k$ most recent
pending calls, thus yielding a $k$-context-sensitive analysis. In particular, instantiating our collapsed semantics with $\pstacks^0$ yields a
context-insensitive analysis. Applying this analysis to
Program~\ref{ex:concrete-semantics-example}, we obtain the following
\emph{collapsed execution map}, which abstracts the relational
execution map for this program shown in
Example~\ref{ex:relational-semantics}:
%
\begin{align*}
\small
%
%
\texttt{id} \mapsto\; & \ptab[\absn]{\{(\nu\!:1), (\nu\!:2)\}}{\{(\absn\!:1, \nu\!:1), (\absn\!:2, \nu\!:2)\}} \\
q \mapsto\; &\ptab[\absn]{\{(\texttt{id}\!:\rdf,\nu\!:1)\}}{\{(\texttt{id}\!:\rdf,\absn\!:1, \nu\!:1)\}} \\
a \mapsto\; &\ptab[\absn]{\{(\texttt{id}\!:\rdf,\texttt{u}:1,\nu\!:2)\}}{\{(\texttt{id}\!:\rdf,\texttt{u}:1,\absn\!:2, \nu\!:2)\}}
\end{align*}
%
Again, we only show some of the entries and for economy of
notation, we omit the abstract stack $\epsilon$ in the singleton tables. Since the
semantics does not maintain any stack information, the
\emph{collapsed} tables no longer track where functions are being
called in the program. For instance, the entry for \lstinline+id+
indicates that \lstinline+id+ is a function called at some concrete
call sites with inputs 1 and 2. While the precise call site stack
information of \lstinline+id+ is no longer maintained, the symbolic
variable $\absn$ still captures the relational dependency between the
input and output values for all the calls to \lstinline+id+.

If we chose to maintain more information in $\pstacks$, the collapsed
semantics is also more precise. For instance, a 1-context-sensitive
analysis is obtained by instantiating the collapsed semantics with
$\pstacks^1$. Analyzing Program~\ref{ex:concrete-semantics-example}
using this instantiation of the collapsed semantics yields a collapsed
execution map that is isomorphic to the relational execution map shown
in Example~\ref{ex:relational-semantics} (i.e., the analysis does not
lose precision in this case).

\end{example}


\smartparagraph{Abstraction.}
Similar to the relational semantics, we formalize the meaning of the
collapsed semantics in terms of a Galois connection between the
complete lattices of relational execution maps $\rmaps$ and collapsed
execution maps $\pmaps$. Again, we only provide the definition of the
right adjoint $\pgammac: \pmaps \to \rmaps$ here. Similar to the relational
semantics, $\pgammac$ is defined in terms of a family of concretization
functions $\pgamma[\rscope]: \pvalues[\rscope] \to \rvalues[\rscope]$
mapping collapsed values to relational values:
\begin{gather*}
\pmgamma(\pmap) \Def = \Lambda \cnode \in \cnodes.\,
  (\pgamma[\rscope_{\rho(\cnode)}] \circ \pmap \circ \rho)(\cnode) \ind
\pgamma[\rscope](\pbot) \Def= \rbot \ind \pgamma[\rscope](\perr) \Def=
\rerr \ind
\pgamma[\rscope](\rel) \Def=\; \rel
\\
\pgamma[\rscope](\absn:\ptable) \Def=\; \absn: \Lambda \cstack \in
\cstacks.\, \Mlet \, \pair{\pval_i}{\pval_o} = (\ptable \circ
\rho)(\cstack) \, \Min
\pair{\pgamma[\rscope](\pval_i)}{\pgamma[\rscope \cup \{\absn\}](\pval_o)}
\end{gather*}
%
More details on the collapsed
semantics including its abstract transformer are provided in
\lessmore{the companion report~\cite[\S~C]{techreport}}{\S~\ref{appendix:collapsed-semantics}}.
\section{Parametric Data Flow Refinement Type Semantics}
\label{sec:arefinement}

At last, we obtain our parametric data flow refinement
type semantics from the collapsed relational semantics by abstracting
dependency relations over concrete constants by abstract relations
drawn from some relational abstract domain. In
\S~\ref{sec:typing-correctness}, we will then show that our data flow
refinement type system introduced in \S~\ref{sec:types} is sound and
complete with respect to this abstract semantics. Finally, in
\S~\ref{sec:type-inference} we obtain a generic type
inference algorithm from our abstract semantics by using
widening to enforce finite convergence of the fixpoint iteration sequence.


\subsection{Type Semantics}

\smartparagraph{Abstract domains.}
The abstract domains of our data flow refinement type semantics
build on the set of types $\tvalues$ defined in
\S~\ref{sec:types}. Recall that $\tvalues$ is
parametric with a set of abstract stacks $\pstacks$ and a complete lattice of basic
refinement types
$\langle \bdomain, \bord, \bbot, \btop, \bjoin, \bmeet
\rangle$, which can be viewed as a union of sets $\bdomain_\pscope$ for each
scope $\pscope$. We require that each $\bdomain_\pscope$ forms a
complete sublattice of $\bdomain$ and that there exists a family of Galois connections\footnote{In fact, 
we can relax this condition and only require a concretization function, thus supporting
abstract refinement domains such as polyhedra~\cite{DBLP:conf/popl/CousotH78}.}
$\pair{\balpha_\pscope}{\bgamma_\pscope}$ between
$\bdomain_\pscope$ and the complete lattice of dependency
relations
$\langle \rels[], \subseteq, \emptyset, \dmaps_{\pscope}, \cup, \cap
\rangle$.  For instance, for the domain
$\bdomain[\lia]$ from Example~\ref{ex:lia-basic-types}, the concretization
function $\bgamma_\pscope$ is naturally obtained from the satisfaction
relation for linear integer constraints.

We lift the partial order $\bord$ on basic refinement types to a
preorder $\tvord[]$ on types as follows:
\begin{align*}
\tval_1 \tvord[] \tval_2 \Def\iff {} & \tval_1 = \tbot \vee \tval_2 = \terr \vee
	(\tval_1, \tval_2 \in \bdomain \land \tval_1 \bord \tval_2)
\vee  (\tval_1, \tval_2 \in \ttables[] \wedge \forall \pstack.\; 
			\tval_1(\pstack) \;\tvecord\; \tval_2(\pstack)) 
\end{align*}
By implicitly taking the quotient of types modulo $\alpha$-renaming of dependency variables
in function types we obtain a partial order that induces a complete
lattice $\langle\tvalues, \tvord[], \tbot, \terr, \tvjoin[],
\tvmeet[]\rangle$. We lift this partial order point-wise to
refinement type maps $\tmaps \Def= \Pi{\pnode \in
  \pnodes}.\;\tvalues[\pscope_{\pnode}]$ and obtain a complete lattice
$\langle\tmaps, \tmord, \tmap_{\bot}, \tmap_{\top}, \tmjoin, \tmmeet\rangle$.

\smartparagraphnb{Galois connection.}  The meaning of refinement types
is given by a function $\tvgamma: \tvalues \to \pvalues$ that
extends $\bgamma$ on basic refinement types.  This function is then
lifted to type maps as before:
\[\begin{array}{c}
\tvgamma[\pscope](\tbot) \Def= \pbot \ind[1] \tvgamma[\pscope](\terr) \Def= \perr  
\ind[1]
\tvgamma[\pscope](\absn:\ttable) \Def= \absn:\Lambda
\pstack.\; \Mlet \, \pair{\tval_i}{\tval_o} = \ttable(\pstack) \, \Min
\pair{\tvgamma[\pscope](\tval_i)}{\tvgamma[\pscope
\cup \{\absn\}](\tval_o)} \\[1ex]
\tvgamma[](\tmap) \Def= \Lambda
\pnode.\, (\tvgamma[X_{\pnode}] \circ \tmap)(\pnode)
\end{array}
\]


\smartparagraphnb{Abstract domain operations.}
We briefly revisit the abstract domain operations on types introduced
in \S~\ref{sec:types} and provide their formal specifications needed
for the correctness of our data flow refinement type semantics. 

We define these operations in terms of three simpler operations on
basic refinement types. First, for $x,y \in \pscope
\cup \{\nu\}$ and $\bval \in \bdomain_\pscope$, let
$\tequality{\bval}{x}{y}$ be an abstraction of the concrete operation
that strengthens the dependency relations described by $\bval$
with an equality constraint $x=y$. That is, we require
$\bgamma_\pscope(\tequality{\bval}{x}{y}) \supseteq \setc{\dmap \in \bgamma_\pscope(\bval)}{\dmap(x) = \dmap(y)}$.
Similarly, for $c \in \consts \cup \{\rdf\}$ we assume that
$\tequality{\bval}{x}{c}$ is an abstraction of the concrete operation
that strengthens $\bval$ with the equality $x=c$, i.e. we require
$\bgamma_\pscope(\tequality{\bval}{x}{c}) \supseteq \setc{\dmap \in \bgamma_\pscope(\bval)}{\dmap(x) = c}$. Lastly,
we assume an \emph{abstract variable substitution} operation,
$\tsubst{\bval}{x}{\nu}$, which must be an abstraction of variable
substitution on dependency relations: $\bgamma_\pscope(\tsubst{\bval}{x}{\nu})
\supseteq \setc{\dmap[\nu \mapsto c, x
  \mapsto \dmap(\nu)]}{\dmap \in \bgamma_\pscope(\bval), c \in \consts}$. We lift these operations
to general refinement types $\tval$ in the expected way. For instance,
we define
\[\tequality{\tval}{x}{c} \Def= \begin{cases}
    \tequality{\tval}{x}{c} & \kw{if}\; \tval \in \bdomain[]\\
    z:\Lambda \pstack.\, \pair{
      \tequality{\tin{\ttable(\pstack)}}{x}{c}}{\tequality{\tout{\ttable(\pstack)}}{x}{c}}
    & \kw{if}\; \tval = z:\ttable \land x \neq \nu\\
    \tval & \kw{otherwise}
  \end{cases}\]
Note that in the second case of the definition, the fact that $x$ is
in the scope of $\tval$ implies $x \neq z$.

We then define the function that yields the abstraction of a constant $c
\in \consts$ as $\tconst \Def= \tequality{\btop}{\nu}{c}$.
The strengthening operation $\tupdate{\tval}{x}{\tval'}$ is defined
recursively over the structure of types as follows:
\[
  \tupdate{\tval}{x}{\tval'} \Def= \begin{cases}
    \tbot & \kw{if}\; \tval' = \tbot\\
    \tequality{\tval}{x}{\rdf} & \kw{else\;if}\; \tval' \in \ttables[]\\
    \tval \bmeet \tsubst{\tval'}{x}{\nu} & \kw{else\;if}\; \tval \in \bdomain[]\\
    z:\Lambda \pstack.\, \pair{ \tupdate{\tin{\ttable(\pstack)}}{x}{\tval'}}{\tupdate{\tout{\ttable(\pstack)}}{x}{\tval'}} & \kw{else\;if}\;
    \tval = z:\ttable\\
    \tval & \kw{otherwise}
  \end{cases}
\]
Finally, we lift $\tupdate{\tval}{x}{\tval'}$ to the operation
$\tstrengthen{\tval}{\tenv}$ that
strengthens $\tval$ with respect to a type environment
$\tenv$ by defining $\tstrengthen{\tval}{\tenv} \Def= \tvmeetb_{x \in
  \dom(\tenv)} \tupdate{\tval}{x}{\tenv(x)}$.

\smartparagraph{Abstract propagation and transformer.}
%
%
The propagation operation $\tiprop$ on refinement types, shown in
Fig.~\ref{fig:type-propagation}, is then obtained from
$\ciprop$ in Fig.~\ref{fig:concrete-prop} by
replacing all operations on concrete values with their counterparts
on types.
\begin{figure}[t]
\[\begin{array}{@{}l@{}l|l}
\absn:\ttable_1 \tiprop \absn:\ttable_2 \Def={} & 
\Mlet \ttable = \Lambda \pstack.	\\[0.4ex]
&\sind
 \Mlet \pair{\tval_{1i}}{\tval_{1o}} = \ttable_1(\pstack) \Mend 
       \pair{\tval_{2i}}{\tval_{2o}} = \ttable_2(\pstack) 
& x:\ttable \tiprop \tbot \Def={} \pair{x:\ttable}{x:\ttable_\bot} \\[0.4ex]
&\sind[2] \pair{\tval'_{2i}}{\tval'_{1i}} = \tval_{2i} \tiprop \tval_{1i} &
x:\ttable \tiprop \terr \Def={} \pair{\terr}{\terr} \\[0.4ex]  
&\sind[2] \pair{\tval'_{1o}}{\tval'_{2o}} = \rupdate{\tval_{1o}}{\absn}{\tval_{2i}} \tiprop \rupdate{\tval_{2o}}{\absn}{\tval_{2i}} &
\tval_1 \tiprop \tval_2 \Def={} \pair{\tval_1}{\tval_1 \tvjoin[] \tval_2} \\[0.4ex]
&\sind \Min \pair{ \pair{\tval'_{1i}}{\tval_{1o} \tvjoin[] \tval'_{1o}}  }{ \pair{\tval'_{2i}}{\tval_{2o} \tvjoin[] \tval'_{2o}  }} 
& \ind\textbf{(otherwise)} \\[0.4ex]
&\;
\Min \pair{\absn:\Lambda \pstack.\; \pi_1(\ttable(\pstack))}
          {\absn:\Lambda \pstack.\; \pi_2(\ttable(\pstack))}  &
\end{array}\]
\caption{Abstract value propagation in the data flow refinement type semantics}
\label{fig:type-propagation}
\vspace{-1em}
\end{figure}
In a similar fashion, we obtain the new abstract transformer
$\ttransname$ for the refinement type semantics from the concrete
transformer $\ctransname$. We again use a state monad to hide the
manipulation of type execution maps. The corresponding operations are
variants of those used in the concrete transformer, which are
summarized in Fig.~\ref{fig:refinement-transformer}. The abstract
transformer closely resembles the concrete one. The only major
differences are in the cases for constant values and variables. Here,
we strengthen the computed types with the relational information about
the variables in scope obtained from the current environment $\tenv =
\tmap \circ \penv$.

\begin{figure}[t]
\[\begin{array}[t]{@{}l|l}
\begin{array}[t]{@{}l}
\ttrans{c_\elabel}(\penv, \pstack) \Def= \\
\sind \Mdo \pnode = \cn{\elabel}{\penv} \,\Mend \tenv \mto
\menv(\penv) \,\Mend
\tval' \mto \mupd{\pnode}{\rstrengthen{\tconst}{\tenv}} \\[0.4ex]
\sind \Mreturn \tval'\\[2ex]
\ttrans{x_\elabel}(\penv, \pstack) \Def= \\
\sind \Mdo  \pnode = \cn{\elabel}{\penv} \, \Mend \, \tval \mto \mread{\pnode}
\, \Mend \, \pnode_x = \penv(x)
\, \Mend \, \tenv \mto \menv(\penv) \, \\[0.4ex]
\sind[6] \tval' \mto \mupd{\pnode_x,\,
\pnode}{\rstrengthen{\tenv(x)\tvareq{x}}{\tenv} \tiprop \rstrengthen{\tval\tvareq{x}}{\tenv}}\\
\sind \Mreturn \tval'\\[2ex]
%
\ttrans{(e_\elabela \, e_\elabelb)_\elabel}(\penv, \pstack) \Def=  \\
\sind 	\Mdo  \pnode = \cn{\elabel}{\penv} \,\Mend \pnode_\elabela =
\cn{\elabela}{\penv}
		\,\Mend \pnode_\elabelb = \cn{\elabelb}{\penv} \,
\Mend \tval \mto \mread{\pnode}\\[0.4ex]
\sind[6] \tval_\elabela \mto \ttrans[k]{e_\elabela}(\penv, \pstack) \,\Mend
\massert{\tval_\elabela \in \ttables[]}\\[0.4ex]
\sind[6] \tval_\elabelb \mto \ttrans[k]{e_\elabelb}(\penv, \pstack)\\[0.4ex]
\sind[6] \tval'_\elabela,\, \absn:[\tentry{\elabela \pconcat \pstack}{\tval_\elabelb'}{\tval'}] = \tval_\elabela \tiprop
\absn:[\tentry{\elabela \pconcat \pstack}{\tval_\elabelb}{\tval}]\\[0.4ex]
\sind[6] \tval''	\mto \mupd{\pnode_\elabela,\, \pnode_\elabelb,\,
\pnode}{\tval'_\elabela,\, \tval'_\elabelb,\, \tval'}\\[0.4ex]
\sind \Mreturn \tval''
\end{array}
&
\begin{array}[t]{@{}l}
\ttrans{(\lambda x.e_\elabela)_\elabel}(\penv, \pstack) \Def= \\ 
\sind	\Mdo \pnode = \cn{\elabel}{\penv} \,\Mend \tval \mto
\mupd{\pnode}{x:\ttable_{\bot}} \\[0.5ex]
\sind[6] \tval' \mto \Mfor \tval \;\Mdo \tstepbody(x,
e_\elabela, \penv, \tval)\\[0.5ex]
\sind[6] \tval'' \mto \mupd{\pnode}{\tval'}\\[0.5ex] 
\sind[2] \Mreturn \tval'' \\[2ex]
\tstepbody(x, e_\elabela, \penv, \tval)(\pstack') \Def= \\
\sind  \Mdo  \pnode_x = \cns{x}{\penv}{\pstack'} \,\Mend
                  \penv_\elabela = \penv.x\!:\!\pnode_x \,\Mend
\pnode_\elabela = \cn{\elabela}{\penv_\elabela}\\[0.5ex]
\sind[6] \tval_x \mto \mread{\pnode_x} \,\Mend \tval_\elabela \mto \ttrans[k+1]{e_\elabela}(\penv_\elabela,
\pstack') \\[0.5ex]
\sind[6]  x:[\ttab[\pstack']{\tval_x'}{\tval_\elabela'}], \tval' =
\\[0.5ex]
\sind[9]  x:[\ttab[\pstack']{\tval_x}{\tval_\elabela}] \tiprop \restr{\tval}{\pstack'} \\[0.5ex]
\sind[6]  \mupd{\pnode_x,\, \pnode_\elabela}{\tval'_x,\, \tval'_\elabela} \\[0.5ex]
\sind  \Mreturn \tval'
\end{array}
\end{array}\]
\caption{Abstract transformer for the data flow refinement type semantics\label{fig:refinement-transformer}}
\vspace{-1em}
\end{figure}

\smartparagraph{Abstract semantics.}
We identify abstract properties $\tdomain$ in the data flow refinement type
semantics with type maps, $\tdomain \Def= \tmaps$ and define $\gamma:
\tdomain \to \cdomain$, which maps abstract to concrete properties by
$\gamma \Def= \rgamma[] \circ \pgammac \circ \tvgamma[]$. 

\begin{lemma}~\label{thm:t-meet-morph-values}
The function $\gamma$ is a complete meet-morphism between $\tdomain$ and
$\cdomain$.
\end{lemma}
It follows that $\gamma$ induces a Galois connection between concrete
and abstract properties.
The data flow refinement semantics $\tcoll: \Exp \to \tdomain$
is then defined as the least fixpoint of $\ttransname$:
\[\tcoll[e] \Def= \mathbf{lfp}^{\tmord}_{\tmap_{\bot}} \Lambda
  \tmap.\; \Mlet \pair{\_}{{\tmap}'} =
  \ttrans[]{e}(\epsilon,\epsilon)(\tmap) \,\Min {\tmap}'\]
%

\begin{theorem}\label{thm:refinement-soundness}
The refinement type semantics is sound, i.e.
$\ccoll\llbracket e \rrbracket \subseteq \gamma(\tcoll[e])$.
\end{theorem}
The soundness proof follows from the calculational design of our
abstraction and the properties of the involved Galois connections.

We say that a type $\tval$ is \emph{safe} if it does not contain
$\ttop$, i.e. $\tval \neq \ttop$ and if $\tval=x:\ttable$ then for all
$\pstack \in \pstacks$, $\tin{\ttable(\pstack)}$ and
$\tout{\ttable(\pstack)}$ are safe. Similarly, a type map $\tmap$ is
safe if all its entries are safe. The next lemma states that safe type
maps yield safe properties. It follows immediately from the
definitions of the concretizations.

\begin{lemma}
  For all safe type maps $\tmap$, $\gamma(\tmap) \subseteq \psafe$.
\end{lemma}

A direct corollary of this lemma and the soundness theorems for
our abstract semantics is that any safe approximation of the
refinement type semantics can be used to prove program safety.

\begin{corollary}
  \label{cor:safety}
  For all programs $e$ and safe type maps $\tmap$, if $\tcoll[e]
  \tmord \tmap$, then $e$ is safe.
\end{corollary}

\subsection{Soundness and Completeness of Type System}
\label{sec:typing-correctness}

It is worth to pause for a moment and appreciate the resemblance
between the subtyping and typing rules introduced in
\S~\ref{sec:types} on one hand, and the abstract propagation operator
$\tiprop$ and abstract transformer $\ttransname$ on the other hand. We
now make this resemblance formally precise by showing that the type
system exactly captures the safe fixpoints of the abstract
transformer. This implies the soundness and completeness of our type
system with respect to the abstract semantics.

We start by formally relating the subtype relation and type
propagation. The following lemma states that subtyping precisely captures the safe
fixpoints of type propagation.

\begin{lemma}\label{lem:subtyping-prop-fixpoint}
  For all $\tval_1,\tval_2 \in \tvalues$,
  $\tval_1 \,{\subtype}\, \tval_2$ iff $\pair{\tval_1}{\tval_2} =
  \tval_1 \tiprop \tval_2$ and $\tval_1,\tval_2$ are safe.
\end{lemma}

\noindent We use this fact to show that any derivation of a
typing judgment $\tenv,\pstack \typrel e : \tval$ represents a
safe fixpoint of $\ttransname$ on $e$, and vice versa, for any safe fixpoint
of $\ttransname$ on $e$, we can obtain a typing derivation.
To state the soundness theorem we need one more definition: we say
that a typing environment is \textit{valid} if it does not map any
variable to $\tbot$ or $\ttop$.

\begin{theorem}[Soundness]
  \label{thm:typing-soundness}
  Let $e$ be an expression, $\tenv$ a valid typing
  environment, $\pstack$ an abstract stack, and $\tval \in \tvalues[]$. If
  $\tenv,\pstack \typrel e: \tval$, then there exist $\tmap,\penv$ such
  that $\tmap$ is safe, 
  $\tenv = \tmap \circ \penv$ and 
  $\pair{\tval}{\tmap} =
  \ttrans{e}(\penv,\pstack)(\tmap)$.
\end{theorem}

\begin{theorem}[Completeness]
  \label{thm:typing-completeness}
  Let $e$ be an expression, $\penv$ an environment,
  $\pstack \in \pstacks$, $\tmap$ a type map, and $\tval \in
  \tvalues[]$. If $\ttrans{e}(\penv,\pstack)(\tmap)=\pair{\tval}{\tmap}$
  and $\tmap$ is safe, then $\tenv,\pstack \typrel e : \tval$ where
  $\tenv = \tmap \circ \penv$.
\end{theorem}

\subsection{Type Inference}
\label{sec:type-inference}

The idea for the generic type inference algorithm is to iteratively
compute $\tcoll[e]$, which captures the most precise typing for $e$ as
we have established above. Unfortunately, there is no guarantee that
the fixpoint iterates of $\ttransname$ converge towards $\tcoll[e]$ in
finitely many steps. The reasons are two-fold. First, the domain
$\bdomain$ may not satisfy the \emph{ascending chain condition}
(i.e. it may have infinite height). To solve this first issue, we
simply assume that $\bdomain$ comes equipped with a family
of widening operators
$\bwid: \bdomain_\pscope \times \bdomain_\pscope \to
\bdomain_\pscope$ for its scoped sublattices. Recall that a widening operator for a complete lattice
$\langle L, \sqleq, \bot, \top, \sqcup, \sqcap \rangle$ is a function
$\wid: L \times L \to L$ such that: (1) $\wid$ is an upper bound
operator, i.e., for all $x,y \in L$, $x \sqcup y \sqleq x \wid y$, and
(2) for all infinite ascending chains $x_0 \sqleq x_1 \sqleq \dots$ in $L$, the
chain $y_0 \sqleq y_1 \sqleq \dots$ eventually stabilizes, where
$y_0 \Def= x_0$ and $y_i \Def= y_{i - 1} \wid x_i$ for all $i > 0$~\cite{cousot1977abstract}.

The second issue is that there is in general no bound on the depth
of function tables recorded in the fixpoint iterates. This phenomenon
can be observed, e.g., in the following program\footnote{We here assume that
  recursive functions are encoded using the $Y$ combinator.}:
\begin{lstlisting}[language=Caml,aboveskip=0.5em,belowskip=0.2em]
let rec hungry x = hungry and loop f = loop (f ()) in
loop hungry
\end{lstlisting}
To solve this second issue, we introduce a \emph{shape widening}
operator that enforces a bound on the depth of tables. The two
widening operators will then be combined to yield a widening operator
for $\tvalues$.
In order to define shape widening, we first define the \emph{shape} of
a type using the function $\shape: \tvalues \to \tvalues$ 
\[
\begin{array}{l}
  \shape(\tbot) \Def= \tbot \quad \shape(\terr) \Def= \terr \quad
  \shape(\tyrelv) \Def= \tbot \quad 
  \shape(\absn:\ttable) \Def= \absn:\Lambda \pstack.\;
  	\pair{\shape(\tin{\ttable(\pstack)})}
		{\shape(\tout{\ttable(\pstack)})}
\end{array}
\]

A shape widening operator is a function
$\shapewid: \tvalues \times \tvalues \to \tvalues$ such that (1)
$\shapewid$ is an upper bound operator and (2) for every infinite ascending
chain $\tval_0 \tvord \tval_1 \tvord \dots$, the chain
$\shape(\tval_0') \tvord \shape(\tval_1') \tvord \dots $ stabilizes,
where $\tval_0' \Def= \tval_0$ and $\tval_i' \Def= \tval_{i-1}' \shapewid
\tval_i$ for $i > 0$.

\tw{
\begin{example}
  The occurs check performed when unifying type variables in type
  inference for Hindley-Milner-style type systems serves a similar
  purpose as shape widening. In fact, we can define a shape widening operator
  that mimics the occurs check. To this end, suppose that each
  dependency variable $x$ is tagged with a finite set of pairs
  $\pair{\elabel}{\pstack}$. We denote this set by
  $\vtag(x)$. Moreover, when computing joins over function types
  $x:\ttable$ and $y:\ttable$, first $\alpha$-rename $x$,
  respectively, $y$ by some fresh $z$ such that
  $\vtag(z) = \vtag(x) \cup \vtag(y)$. We proceed similarly when
  applying $\tiprop$ to function types. Finally, assume that each
  fresh dependency variable $x$ generated by $\ttransname$ for the
  function type at a call expression $e_\elabela \, e_\elabelb$ has
  $\vtag(x)=\{\pair{\elabela}{\pstack}\}$ where $\pstack$ is
  the abstract call stack at this point. Then to obtain
  $\tval_1 \shapewid \tval_2$, first define
  $\tval=\tval_1 \tvjoin \tval_2$. If $\tval$ contains two distinct
  bindings of dependency variables $x$ and $y$ such that
  $\vtag(x)=\vtag(y)$, define $\tval_1 \shapewid \tval_2=\ttop$ and
  otherwise $\tval_1 \shapewid \tval_2=\tval$. Clearly this is a
  shape widening operator if we only consider the finitely many tag
  sets that can be constructed from the locations in the analyzed
  program.
\end{example}}

In what follows, let $\shapewid$ be a shape widening operator. First, we
lift the $\bwid$ on $\bdomain$ to an upper bound operator $\trelwid$ on $\tvalues$:
\[
  \tval \trelwid \tval' \Def= \left\{
    \begin{array}{l l}
\tval \bwid \tval'
& \textbf{if } \tval,\tval' \in \bdomain \\[0.5ex]
\absn:\Lambda \pstack. \begin{array}[t]{l}\Mlet \pair{\tval_{i}}{\tval_{o}} =
\ttable(\pstack) \Mend \pair{\tval'_{i}}{\tval'_{o}} =
{\ttable}'(\pstack)\\
\Min \; \pair{\tval_i \trelwid[\pscope] \tval'_i}
	{\tval_o \trelwid[\pscope \cup \set{\absn}] \tval'_o }
\end{array}
& \textbf{if } \tval = \absn:\ttable \land \tval' = \absn:{\ttable}'\\[0.5ex]
\tval \tvjoin[] \tval' & \textbf{otherwise}
\end{array}
\right.
\]
We then define $\twid: \tvalues \times \tvalues \to \tvalues$ as the
composition of $\shapewid$ and $\trelwid$, that is,
$\tval \twid \tval' \Def= \tval \trelwid (\tval \shapewid \tval')$.

\begin{lemma}
  \label{lem:type-widening}
  $\twid$ is a widening operator.
\end{lemma}

We lift the widening operators $\twid$ pointwise to an upper bound
operator
$\mathrel{\dot{\wid}^{\tdesignation}}: \tmaps \times \tmaps \to
\tmaps$ and define a widened data flow refinement semantics
$\tcollwid: \Exp \to \tdomain$ as the least fixpoint of
the widened iterates of $\ttransname$:
\begin{align}\tcollwid[e] \Def= {} & \mathbf{lfp}^{\tmord}_{\tmap_{\bot}} \Lambda
  \tmap.\; \Mlet \pair{\_}{{\tmap}'} =
  \ttrans[]{e}(\epsilon,\epsilon)(\tmap) \,\Min 
                                     (\tmap \mathrel{\dot{\wid}^{\tdesignation}} {\tmap}')
  \label{eq:algo}
\end{align}
%

The following theorem then follows directly from
Lemma~\ref{lem:type-widening} and~\cite{cousot1977abstract}.

\begin{theorem}
\label{thm:widening-soundness}
The widened refinement type semantics is sound and terminating, i.e., for all
programs $e$, $\tcollwid[e]$ converges in finitely many iterations.
Moreover, $\tcoll[e] \tmord \tcollwid[e]$.
\end{theorem}

Our parametric type inference algorithm thus computes $\tcollwid[e]$
iteratively according to (\ref{eq:algo}). By
Theorem~\ref{thm:widening-soundness} and Corollary~\ref{cor:safety},
if the resulting type map is safe, then so is $e$.

%




\section{Implementation and Evaluation}
\label{sec:implementation}
\newcolumntype{?}{!{\vrule width 1pt}}

We have implemented a prototype of our parametric data flow refinement
type inference analysis in a tool called
\ourtool\footnote{\tw{\url{https://github.com/nyu-acsys/drift/}}}. The
tool is written in OCaml and builds on top of the \tool{Apron}
library~\cite{DBLP:conf/cav/JeannetM09} to support various numerical
abstract domains of type refinements. We have implemented two versions
of the analysis: a context-insensitive version in which all entries in
tables are collapsed to a single one (as in Liquid type inference) and
a 1-context-sensitive analysis that distinguishes table entries based
on the most recent call site locations. For the widening on basic
refinement types, we consider two variants: plain widening and
widening with
thresholds~\cite{DBLP:journals/corr/abs-cs-0701193}. Both variants use
\tool{Apron}'s widening operators for the individual refinement
domains.

The tool takes programs written in a subset of OCaml as input. This
subset supports higher-order recursive functions, operations on
primitive types such as integers and Booleans, as well as lists and arrays. We note
that the abstract and concrete
transformers of our semantics can be easily extended to handle
recursive function definitions directly. As of now,
the tool does not yet support user-defined algebraic data
types. \ourtool automatically checks whether all array accesses are
within bounds. In addition, the tool supports the verification of
user-provided assertions. Type refinements for lists can express
constraints on both the list's length and its elements.

To evaluate \ourtool we conduct two experiments that aim to answer the
following questions:

\begin{enumerate}
\item What is the trade-off between efficiency and precision for
  different instantiations of our parametric analysis framework?
\item How does our new analysis compare with other state-of-the-art
  automated verification tools for higher-order programs?
\end{enumerate}

\subsubsection*{Benchmarks and Setup}
We collected a benchmark suite of OCaml programs by combining several
sources of programs from prior work and augmenting it with our own new
programs. Specifically, we included the programs used in the
evaluation of the \tool{DOrder}~\cite{DBLP:conf/pldi/ZhuPJ16} and
\tool{R\_Type}~\cite{DBLP:conf/tacas/ChampionC0S18} tools, excluding
only those programs that involve algebraic data types \tw{or certain
  OCaml standard library
functions that our tool currently does not yet support}. We generated
a few additional variations of some of these programs by adding or modifying
user-provided assertions, or by replacing concrete test inputs for the
program's \texttt{main} function by unconstrained parameters. In
general, the programs are small (up to 86 lines) but intricate.
We partitioned the programs into five categories: first-order arithmetic
programs (FO), higher-order arithmetic programs (HO), \tw{higher-order
  programs that were obtained by reducing program termination to 
  safety checking (T)}, array programs
(A), and list programs (L). \tw{All programs except two in the T category
are safe. We separated these two erroneous programs out into a sixth
category (E) which we augmented by additional unsafe programs obtained by
modifying safe programs so that they 
contain implementation bugs or faulty specifications.}
\tw{The benchmark suite is available in the tool's Github repository.}
All our experiments were conducted on a desktop computer with an
Intel(R) Core(TM) i7-4770 CPU and 16 GB memory running Linux.

\subsubsection*{Experiment 1: Comparing different configurations of \ourtool}
We consider the two versions of our tool (context-insensitive and
1-context-sensitive) and instantiate each with two different
relational abstract domains implemented in \tool{Apron}:
Octagons~\cite{DBLP:journals/corr/abs-cs-0703084} (Oct), and Convex
Polyhedra and Linear Equalities~\cite{DBLP:conf/popl/CousotH78}
(Polka). For the Polka domain we consider both its \emph{loose}
configuration, which only captures non-strict inequalities, as well as
its \emph{strict} configuration which can also represent strict
inequalities. We note that
the analysis of function calls critically relies on the abstract
domain's ability to handle equality constraints precisely. We
therefore do not consider the interval domain as it cannot express
such relational constraints. For each abstract domain, we further
consider two different widening configurations: standard widening (w) and widening with
thresholds (tw). For widening with
thresholds~\cite{DBLP:journals/corr/abs-cs-0701193}, we use a
simple heuristic that chooses the conditional expressions in the
analyzed program as well as pair-wise inequalities between the variables
in scope as threshold constraints.



\begin{table}[t]
\centering
\caption{Summary of Experiment 1. For each benchmark category, we
  provide the number of programs within that category in
  parentheses. For each benchmark category and configuration, we list:
  the number of programs successfully analyzed \textbf{(succ)} and the
  total accumulated running time across all benchmarks in seconds
  \textbf{(total)}. \tw{In the E category, an analysis run is
    considered successful if it flags the type error/assertion
    violation present in the benchmark. The numbers given in
    parenthesis in the full columns indicate the number of benchmarks
    that failed due to timeouts (if any). These benchmarks
    are excluded from the calculation of the cumulative running
    times. The timeout threshold was 300s per benchmark.}
}
\label{table:1}
\resizebox{\textwidth}{!}{
\begin{tabular}{ |c||c|c|c|c|c|c|c?c|c|c|c|c|c| }
\hline
\multirow{4}*{\shortstack{Benchmark\\category}} & & \multicolumn{12}{c|}{Configuration} \\
\cline{2-14}
& Version & \multicolumn{6}{c?}{{context-insensitive}} & \multicolumn{6}{c|}{{1-context-sensitive}} \\
\cline{2-14}
& Domain & \multicolumn{2}{c|}{Oct} & \multicolumn{2}{c|}{Polka strict} & \multicolumn{2}{c?}{Polka loose} & \multicolumn{2}{c|}{Oct} & \multicolumn{2}{c|}{Polka strict} & \multicolumn{2}{c|}{Polka loose} \\
\cline{2-14}
& Widening & w & tw & w & tw & w & tw & w & tw & w & tw & w & tw \\
\hline\hline
FO (73) & succ  & 25  & 39  & 39  & 51  & 39  & 51  & 33  & 46  & 47  & 58  & 47  & 59  \\
loc: 11 & total & 10.40 & 46.94 & 17.37 & 46.42(1)  & 15.82 & 42.88(1)  & 67.38 & 129.80  & 92.37 & 138.00(1) & 87.27 & 138.33  \\
\hline
HO (62) & succ  & 33  & 48(1) & 42  & 55  & 42  & 55  & 42  & 51  & 48  & 60  & 48  & 60  \\
loc: 10 & total & 8.53  & 49.97 & 14.97 & 60.67 & 14.03 & 54.71 & 83.56 & 282.57  & 119.10  & 345.03  & 112.05  & 316.18  \\
\hline
T (80)  & succ  & 72  & 72  & 72  & 72  & 72  & 72  & 79  & 79  & 78  & 79  & 78  & 79  \\
loc: 44 & total & 806.24  & 842.70  & 952.53  & 994.52  & 882.21  & 924.00  & 1297.13(1)  & 1398.25(1)  & 1467.09(1)  & 1566.11(1)  & 1397.71(1)  & 1497.28(1)  \\
\hline
A (13)  & succ  & 6 & 6 & 8 & 8 & 8 & 8 & 8 & 8 & 11  & 11  & 11  & 11  \\
loc: 17 & total & 4.30  & 23.43 & 7.66  & 25.04 & 7.19  & 23.38 & 17.99 & 41.66 & 28.07 & 47.77 & 26.84 & 45.94 \\
\hline
L (20)  & succ  & 8(2)  & 14  & 10  & 14  & 10  & 14  & 11  & 18  & 10  & 18  & 10  & 18  \\
loc: 16 & total & 1.62  & 12.02 & 4.00  & 12.52 & 3.45  & 10.97 & 5.73  & 27.01 & 11.46 & 29.01 & 10.32 & 26.43 \\
\hline
E (17)  & succ  & 17  & 17  & 17  & 17  & 17  & 17  & 17  & 17  & 17  & 17  & 17  & 17  \\
loc: 21 & total & 8.04  & 14.89 & 12.72 & 19.44 & 12.01 & 18.28 & 17.90 & 28.76 & 24.96 & 34.24 & 23.75 & 32.71 \\
\hline
\end{tabular}}
\end{table}

Table~\ref{table:1} summarizes the results of the experiment. First,
note that all configurations successfully flag all erroneous
benchmarks (as one should expect from a sound analysis). Moreover, the
context-sensitive version of the analysis is in general more precise
than the context-insensitive one. The extra precision comes at the
cost of an increase in the analysis time by a factor of \ys{1.8} on average. The
1-context-sensitive version with Polka loose/tw performed best, solving 244
out of 265 benchmarks. \ys{There are two programs for which some of the
  configurations produced timeouts. However, each of these programs
  can be successfully verified by at least one configuration}. 
As expected, using Octagon is more
efficient than using loose polyhedra, which in turn is more efficient than
strict polyhedra. We anticipate that the differences in running times
for the different domains will be more pronounced on larger
programs. In general, one can use different domains for different
parts of the program as is common practice in static analyzers such as
\tool{Astr\'ee}~\cite{DBLP:journals/fmsd/CousotCFMMR09}.


We conducted a detailed analysis of the 20 benchmarks that \ourtool
could not solve using any of the configurations that we have
considered. To verify the 16 failing benchmarks in the FO and HO
categories, one needs to infer
type refinements that involve either non-linear or non-convex
constraints, neither of which is
currently supported by the tool. This can be addressed e.g. by further increasing
context-sensitivity, by using more expressive domains such as interval
polyhedra~\cite{DBLP:conf/sas/ChenMWC09}, and by incorporating
techniques such as trace partitioning to reduce loss of precision due
to joins~\cite{DBLP:conf/esop/MauborgneR05}. The two failing
benchmarks in the array category require the analysis to capture
universally quantified invariants about the elements stored in
arrays. However, our implementation currently only tracks array length
constraints. It is relatively straightforward to extend the analysis
in order to capture constraints on elements as has been proposed
e.g. in~\cite{DBLP:conf/esop/VazouRJ13}.


We further conducted a more detailed analysis of the running times by
profiling the execution of the tool. This analysis determined that
most of the time is spent in the strengthening of output types
with input types when propagating recursively between function types
in $\tiprop$. This happens particularly often when analyzing programs
that involve applications of curried functions, which are currently
handled rather naively\tw{, causing a quadratic blowup that can be avoided
with a more careful implementation. Notably, the programs in the T category
involve deeply curried functions. This appears to be the primary
reason why the tool is considerably slower on these programs.} Moreover, the implementation of the fixpoint
loop is still rather naive as it calls $\tiprop$ even if the
arguments have not changed since the previous fixpoint
iteration. We believe that by avoiding redundant calls to $\tiprop$
the running times can be notably improved.



\subsubsection*{Experiment 2: Comparing with other Tools}
Overall, the results of Experiment 1 suggest that the
1-context-sensitive version of \ourtool instantiated with the loose Polka
domain and threshold widening provides a good balance between
precision and efficiency. In our second experiment, we compare this
configuration with several other existing tools. We consider three
other automated verification tools: \tool{DSolve}, 
\tool{R\_Type}, and \tool{MoCHi}. \tool{DSolve} is the original
implementation of the Liquid type inference algorithm proposed
in~\cite{DBLP:conf/pldi/RondonKJ08}
(cf. \S~\ref{sec:df-overview}). \tool{R\_Type}~\cite{DBLP:conf/tacas/ChampionC0S18} improves
upon \tool{DSolve} by \tw{replacing the
Houdini-based fixpoint algorithm of~\cite{DBLP:conf/pldi/RondonKJ08}
with the Horn clause solver
\tool{HoIce}~\cite{DBLP:conf/aplas/Champion0S18}. \tool{HoIce} uses an
ICE-style machine learning
algorithm~\cite{DBLP:conf/cav/0001LMN14} that, unlike Houdini, can
also infer disjunctive refinement predicates.}
We note that \tool{R\_Type} does
not support arrays or lists and hence we omit it from these categories. Finally,
\tool{MoCHi}~\cite{DBLP:conf/pldi/KobayashiSU11} is a software model checker
based on higher-order recursion schemes\tw{, which also uses
  \tool{HoIce} as its default back-end Horn clause solver. We used the
  most recent version of each tool at the time when the experiments were
  conducted and we ran all tools in their default configurations. More
precise information about the specific tool versions used can be found in
the tools' Github repository.}


We initially also considered
\tool{DOrder}~\cite{DBLP:conf/pldi/ZhuPJ16} in our
comparison. This tool builds on the same basic algorithm as
\tool{DSolve} but also learns candidate predicates from concrete
program executions via machine learning. However, the tool primarily
targets programs that manipulate algebraic data types. Moreover, \tool{DOrder}
relies on user-provided test inputs for its predicate inference. As
our benchmarks work with unconstrained input parameters and we
explicitly exclude programs manipulating ADTs from our benchmark set, this puts
\tool{DOrder} decidedly at a disadvantage. To keep the comparison
fair, we therefore excluded \tool{DOrder} from the experiment.

\begin{table}[t]
\centering
\caption{Summary of Experiment 2. In addition to the cumulative
  running time for each category, we provide the average
  \textbf{(avg)} and median \textbf{(med)} running time per benchmark
  (in s). Timeouts are reported as in Table~\ref{table:1}. The timeout
  threshold was 300s per benchmark across all tools.
  In the
  \textbf{(succ)} column, we additionally provide in
  parentheses the number of benchmarks that were only solved by that
  tool (if any).}
\label{table:2}
\resizebox{\textwidth}{!}{
\begin{tabular}{|c||c|c|c|c?c|c|c|c?c|c|c|c?c|c|c|c|} 
\hline
\multirow{2}*{\shortstack{Bench-\\mark cat.}} &
                                                \multicolumn{4}{c?}{{\ourtool}} & \multicolumn{4}{c?}{\tool{R\_Type}} 
  & \multicolumn{4}{c?}{\tool{DSolve}} & \multicolumn{4}{c|}{\tool{MoCHi}}\\
\cline{2-17}
  & succ & full & avg & med & succ & full & avg & med & succ & full & avg & med & succ & full & avg & med \\
\hline\hline
FO (73) 
& 59(5) & 138.33  & 1.89  & 0.26 & 44 & 5.78(15)  & 0.08  & 0.06  
& 49  & 16.27 & 0.22  & 0.12 & \bf 62 & 332.39(9) & 4.55  & 21.50  \\ 
\hline
HO (62) 
& \bf 60  & 316.18  & 5.10  & 1.77 & 49 & 9.19(5) & 0.15  & 0.03
& 41  & 9.76  & 0.16  & 0.24 & 58 & 276.37(4) & 4.46  & 15.18 \\
\hline
T (80)
& 79  & 1497.28(1)  & 18.72 & 0.00 & 73 & 13.18 & 0.16  & 0.04
& 30  & 26.26 & 0.33  & 0.45 & \bf 80 & 41.56 & 0.52  & 0.11 \\
\hline
A (13)
& \bf 11  & 45.94 & 3.53  & 3.70 & -  & - & - & -
& 8 & 7.15  & 0.55  & 0.77 & 9(1) & 2.44  & 0.19  & 0.10 \\
\hline
L (20)
& \bf 18(2) & 26.43 & 1.32  & 1.23 & -  & - & - & -
& 8 & 6.10  & 0.30  & 0.26 & 17 & 455.28(3) & 22.76 & 19.81 \\
\hline
E (17)
& \bf 17  & 32.71 & 1.92  & 0.39 & \bf 17 & 1.22  & 0.07  & 0.05
& 14  & 8.22  & 0.48  & 0.17	& 14 & 95.95(3)  & 5.64  & 43.41 \\
\hline
\end{tabular}}
\end{table}

Table~\ref{table:2} summarizes the results of our comparison. \tw{\ourtool
and \tool{MoCHi} perform similarly overall and significantly better than
the other two tools. In particular, we note that only \ourtool and
\tool{MoCHi} can verify the
second program discussed in \S~\ref{sec:df-overview}. The
evaluation also indicates complementary strengths of these tools. In terms of number of
verified benchmarks, \ourtool
performs best in the HO, A, and L categories with \tool{MoCHi} scoring
a close second place.  For the FO and T
categories the roles are reversed. In the FO category, \tool{MoCHi}
benefits from its ability to infer non-convex type refinements, which are needed
to verify some of the benchmarks in this category. Nevertheless
there are five programs in this category that only \ourtool can
verify. Unlike our current implementation, \tool{MoCHi} does not
appear to suffer from inefficient handling of deeply curried
functions, which leads to significantly better cumulative running times in the T
category. On the other hand, \ourtool is faster than \tool{MoCHi} in the L
category.

We note that there are two benchmarks in the FO category and four
benchmarks in the A category for which \tool{MoCHi} produces false
alarms. However, this appears to be due to the use of certain language
features that are unsupported by the tool (such as OCaml's
\lstinline+mod+ operator).}



None of the tools produced unsound results in the E category. The
failing benchmarks for \tool{MoCHi} are due to timeouts and the ones
for \tool{DSolve} are due to crashes. The considered timeout was 300
seconds per benchmark for all tools. Across all benchmarks,
\tool{R\_Type} timed out on 20 and \tool{MoCHi} on 19
programs. \ourtool timed out on only one benchmark in the T
category.
We attribute the good
responsiveness of \ourtool to the use of infinite refinement domains
with widening in favor of the counterexample-guided abstraction
refinement approach used by \tool{R\_Type} and \tool{MoCHi}, which does
not guarantee termination of the analysis.

\section{Related Work}
\label{sec:related}

\smartparagraph{Refinement type inference.} Early work on refinement
type systems supported dependent types with unions and intersections
based on a \emph{modular} bidirectional type checking
algorithm~\cite{DBLP:conf/pldi/FreemanP91,DBLP:conf/popl/DunfieldP04,DBLP:conf/fossacs/DunfieldP03}. These
algorithms require some type annotations. Instead, the focus of this
paper is on fully automated refinement type inference algorithms that
perform a whole program analysis. Many existing algorithms in this
category can be obtained by instantiating our parametric
framework. Specifically, Liquid type
inference~\cite{DBLP:conf/pldi/RondonKJ08,DBLP:conf/esop/VazouRJ13,DBLP:conf/icfp/VazouSJVJ14,DBLP:conf/icfp/VazouBJ15}
performs a context-insensitive analysis over a monomial predicate
abstraction domain of type refinements. Similarly,
\citet{DBLP:conf/vmcai/ZhuJ13} propose a 1-context sensitive analysis
with predicate abstraction, augmented with an additional
counterexample-guided refinement loop\tw{, an idea that has also
  inspired recent techniques for analyzing pointer programs~\cite{DBLP:conf/esop/TomanSSI020}}. Our work generalizes these
algorithms to arbitrary abstract domains of type refinements
(including domains of infinite height) and provides parametric and
constructive soundness and completeness proofs for the obtained type
systems.
Orthogonal to our work are extensions of static inference algorithms with
data-driven approaches for inferring refinement predicates~\cite{DBLP:conf/icfp/ZhuNJ15,DBLP:conf/pldi/ZhuPJ16, DBLP:conf/tacas/ChampionC0S18}. 
Gradual Liquid type inference~\cite{DBLP:journals/pacmpl/VazouTH18}
addresses the issue of how to apply whole program analysis to infer modular
specifications and improve error reporting. Our framework is in
principle compatible with this approach. We note that Galois connections are used
in~\cite{DBLP:conf/pldi/VekrisCJ16,DBLP:conf/vmcai/KazerounianVBFT18,DBLP:conf/popl/LehmannT17,DBLP:conf/popl/GarciaCT16,DBLP:journals/pacmpl/VazouTH18}
to relate dynamic gradual refinement types and static refinement types. 
However, the resulting gradual type systems are not calculationally
constructed as abstract interpretations of concrete program
semantics.

\smartparagraph{Semantics of higher-order programs.}
The techniques introduced in~\cite{Jagannathan:1995:UTF:199448.199536} and~\cite{Plevyak95iterativeflow} use
flow graphs to assign concrete meaning to higher-order programs. Nodes in these graphs
closely relate to the nodes in our data flow semantics. Their semantics represents functions as expression 
nodes storing the location of the function expression. Hence, this semantics has to make non-local
changes when analyzing function applications. Our data flow semantics treats 
functions as tables and the concrete transformer is defined structurally on program syntax, which is 
more suitable for deriving type inference analyses. The idea of modeling functions
as tables that only track inputs observed during program
execution was first explored in the minimal function graph
semantics~\cite{DBLP:journals/jflp/JonesR97,DBLP:conf/popl/JonesM86}. However, that
semantics does not explicitly model data flow in a program,
and is hence not well suited 
for constructing refinement type inference algorithms. There is a large body of work on control and 
data flow analysis of higher-order programs and the concrete semantics they 
overapproximate~\cite{DBLP:journals/csur/Midtgaard12,DBLP:conf/popl/NielsonN97,DBLP:journals/cacm/HornM11,DBLP:journals/njc/Mossin98,Jones:2007:FAL:1235896.1236115,DBLP:conf/iccl/CousotC94}. 
However, these semantics either do not capture data flow properties or rely on 
non-structural concrete transformers. An interesting direction for
future work is to reconcile our collapsed semantics with control flow analyses
that enjoy \textit{pushdown precision} such
as~\cite{DBLP:journals/infsof/Reps98,DBLP:conf/popl/GilrayL0MH16,DBLP:journals/corr/abs-1102-3676,DBLP:conf/icfp/VardoulakisS11}.

\zp{Temporal logic for higher-order programs}~\cite{DBLP:conf/sas/OkuyamaT019} and
higher-order recursion schemes~\cite{DBLP:conf/pldi/KobayashiSU11,
  DBLP:conf/lics/Ong15} provide alternative bases for verifying
functional programs. Unlike our framework, these approaches use finite
state abstractions and model checking. In particular, they rely
on abstraction refinement to support infinite height data domains,
giving up on guaranteed termination of the analysis. On the other
hand, they are not restricted to proving safety properties.

\smartparagraph{Types as abstract interpretations.}
\lessmore{The formal connection between types and abstract
interpretation is studied by~\citet{cousot1997types}. The paper shows how to construct
standard (modular) polymorphic type systems via a sequence of abstractions of
denotational call-by-value semantics. Similarly, \citet{monsuez1992polymorphic,monsuez1993polymorphic,monsuez1995system,monsuez1995using}
uses abstract interpretation to design polymorphic type systems for call-by-name semantics 
and model advanced type systems such as System 
F. \citet{DBLP:conf/vmcai/GoriL03,DBLP:conf/vmcai/GoriL02} 
use abstract interpretation to design new type inference algorithms for ML-like languages 
by incorporating more precise widening operators for analyzing
recursive functions. None of these works address refinement type inference.
Harper introduces a framework for constructing dependent type systems from operational 
semantics based on the PER model of types~\cite{DBLP:journals/jsc/Harper92}. Although this 
work does not use abstract interpretation, it views types as an overapproximation
of program behaviors derived using a suitable notion of abstraction.
}{
The formal connection between types and abstract
interpretation is studied by~\cite{cousot1997types}. His paper shows how to construct
standard polymorphic type systems via a sequence of abstractions starting from a concrete 
denotational call-by-value semantics. \cite{monsuez1993polymorphic2}
similarly shows how to use abstract interpretation to design
polymorphic type systems for call-by-name semantics and model advanced type systems
such as System 
F~\cite{monsuez1992polymorphic,monsuez1993polymorphic,monsuez1995system,monsuez1995using}. \cite{DBLP:conf/vmcai/GoriL03,DBLP:conf/vmcai/GoriL02} 
use the abstract interpretation view of types
to design new type inference for ML-like languages by incorporating
widening operators that infer more precise types for recursive
functions.
\cite{DBLP:conf/sas/JensenMT09} use abstract interpretation
techniques to develop a type system for JavaScript that is able to guarantee absence
of common errors such as invoking a non-function value or accessing an undefined 
property.~\cite{DBLP:conf/popl/GarciaCT16} use Galois connections to relate
gradual types with abstract static types to build a gradual
type system from a static one. \cite{DBLP:journals/jsc/Harper92} introduces a framework
for constructing dependent type systems from operational semantics
based on the PER model
of types. Although this work does not build on abstract
interpretation, it relies on the idea that types overapproximate
program behaviors and that they can be derived using a suitable notion of abstraction.
}

\section{Conclusion}
In this work, we systematically develop
a parametric refinement type systems as an abstract interpretation of
a new concrete data flow semantics.  This development unveils the
design space of refinement type analyses that infer data flow
invariants for functional programs. Our prototype implementation and
experimental evaluation \tw{indicate that our framework can be used to implement new
  refinement type inference algorithms that are both robust and precise.}

\begin{acks}                            
This work is funded in parts by the
\grantsponsor{GS100000001}{National Science Foundation}{http://dx.doi.org/10.13039/100000001} under
grants~\grantnum{GS100000001}{CCF-1350574} and
\grantnum{GS100000001}{CCF-1618059}. We thank the anonymous reviewers
for their feedback on an earlier draft of this paper.
\end{acks}

\bibliographystyle{ACM-Reference-Format}
\def\doi[#1]{\href{https://doi.org/#1}{\url{#1}}}
\bibliography{references,dblp}

\more{
\clearpage

\appendix
\section{Additional Examples of the Concrete Dataflow Semantics}
\label{sec:appendix-additional-examples}

\subsection{Higher-order Example}
\label{sec:appendix-ho-example}

Let us consider an example involving higher-order functions.
The code is again given below.
\begin{lstlisting}[language=Caml,mathescape = true]
$\;$let dec y = y - 1 in
$\;$let f x g = if x >= 0 then g${}_o$ x${}_p$ else x${}_q$ 
$\;$in f${}_a$ ((f${}_b$ 1${}_c$)${}_d$ dec${}_i$)${}_j$ dec${}_m$
\end{lstlisting}
The concrete execution map for this program, restricted to
the execution nodes of interest is as follows:
\[
\begin{array}{l}
q^a \mapsto \bot \ind[1] q^b \mapsto \bot \ind[1] c \mapsto 1 \ind[1] 
p^a \mapsto 0 \ind[1] p^b \mapsto 1 \ind[1] i \mapsto 0 \\[2ex]
i \mapsto [ \tentry{od}{1}{0}] \ind[1] m \mapsto [ \tentry{oj}{0}{-1}] 
\ind[1] \mathtt{dec} \mapsto [\tentry{od}{1}{0},\;\; 
                              \tentry{oj}{0}{-1} ]\\[2ex]
a \mapsto [\tentry{a}{0}{[ \tentry{j}{[ \tentry{oj}{0}{-1}]}{-1}]}] \ind[1]
b \mapsto [\tentry{c}{1}{[\tentry{d}{[\tentry{od}{1}{0}]}{0}]}] \\[2ex]
f \mapsto  [\tentry{a}{0}{[ \tentry{j}{[ \tentry{oj}{0}{-1}]}{-1}]},\;\;
            \tentry{b}{1}{[\tentry{d}{[\tentry{od}{1}{0}]}{0}]}  ] 
\end{array}
\]
For this specific program, the execution will always reach the program expressions with locations
$\mathtt{a,b,c,d,i,j,}$ and $\mathtt{m}$ with the same environment. For this reason, we again identify
these expression nodes with their locations. Expressions with labels $\mathtt{o,p,}$ and
$\mathtt{q}$ will be reached with two environments and stacks, each of which corresponds to executing
the body of \texttt{f} with the value passed to it at call sites with locations $\mathtt{a}$ and $\mathtt{b}$.
We use this call site information to distinguish the two evaluations of the body of \texttt{f}.
For instance, the node associated with the $\texttt{x}_{q}$ expression when
performing the analysis for the call to \texttt{f} made at location $a$ is denoted by
$\mathtt{q}^{a}$ in the above table. 

\subsection{Recursion Example Concrete Map}
\label{sec:appendix-recursion-example}

We next exemplify how our concrete data flow semantics models recursion through 
an OCaml program that can be found to the left in Fig.~\ref{ex:recursion-program}.
The program first defines a function \lstinline+f+ that returns 0 when 
given 1 as input, and it otherwise just recurses with a decrement of the 
input. This function is then called two times with inputs 2 and 0, 
respectively. The latter call never terminates. To avoid clutter, we only 
annotate program expressions of particular interest with their locations. 

\begin{figure}[h!]
\begin{minipage}[t]{.35\textwidth}
\begin{lstlisting}[language=Caml,mathescape=true, columns=fullflexible, numbers=none]
let rec f n = 
  if n=1 then 0 else f${}_{a}$ (n-1) 
in (f${}_{b}$ 2) + (f${}_{c}$ 0)
\end{lstlisting}
\end{minipage}%
\begin{minipage}[t]{.64\textwidth}
\[\begin{array}{l}
b \mapsto [ \tentry{b}{2}{0} ] \ind[2] c \mapsto [ \tentry{c}{0}{\bot} ] \\
f \mapsto   [ \tentry{b}{2}{0},\;\;  \tentry{c}{0}{\bot},\;\;
              \tentry{ab}{1}{0},\\ 
\ind[1]\sind	      \tentry{ac}{-1}{\bot},\;\;\tentry{aac}{-2}{\bot},\;\; \dots]
\end{array}\]
\end{minipage}
\caption{A program exhibiting recursion and excerpt of its execution map~\label{ex:recursion-program}}
\end{figure}
Let us first discuss the node $b$, which corresponds to the call to
\lstinline+f+ with integer 2. The meaning assigned to $b$ is a table
indicating that the corresponding function is invoked at the call site
stack associated with $b$ with integer 2 as input, returning 0 as the
output. The computation of this output clearly requires a recursive
call to \lstinline+f+ with 1 as the input. The existence and effect of
this call can be observed by the entry $\tentry{ab}{1}{0}$ in the
table associated with \lstinline+f+. Intuitively, this entry states
that \lstinline+f+ was recursively called at the point identified by
$a$ while executing the call at $b$.

Consider now the node $c$, which corresponds to the call to \lstinline+f+ 
with integer 0. The meaning assigned to $c$ is a table indicating that the 
corresponding function is invoked at the call site stack associated with $c$ with integer 0 but 
with no computed output $\bot$. Clearly,
there is no useful output computed since \lstinline+f+ does not terminate with 0 as input.
The recursive calls to \lstinline+f+ observed after calling \lstinline+f+ with 0 and their effects 
can be again observed in the table computed for \lstinline+f+. For instance, the table entry $\tentry{ac}{-1}{\bot}$
captures that the recursive call to \lstinline+f+ at the call site
point $a$ after entering \lstinline+f+ from $c$ has seen $-1$ as the input, 
but no output has been computed. 
Similarly, the entry $aac \mapsto \pair{-2}{\bot}$ indicates that \lstinline+f+ was recursively called
at $a$ with $-2$ as input after the previous recursive call, and again
no output was observed. Due to the nonterminating execution, the table
stored for \lstinline+f+ has infinitely many entries of the shape
$\tentry{a^kd}{-k}{\bot}$.

\section{Relational Data Flow Semantics}
\label{appendix:relational-semantics}


\subsection{Notation}

Before proceeding to the definitions of the ordering
and abstract transformer for the relational semantics,
we first formally introduce auxiliary notation and
definitions.

We use $\absn:\rtable_{\bot}$ 
for the empty relational table that maps every call site to the pair $\pair{\rbot}{\rbot}$.
We denote the fact that a relational table $\rtable$ has been invoked at a
call site stack $\ccsite$ with dependency variable $\absn$ by
$\ccsite, \absn \in \rtable$.
By $\rmap_{\bot}$ ($\rmap_{\cerr}$) we denote 
the relational execution maps that assign $\rbot$ ($\rerr$) to every node.

We denote by $\rvalues[]$ the union of all $\rvalues$ and use similar
notation for the sets of all dependency relations and relational
tables. For the members of these sets we sometimes omit the subscript
indicating the scope when it is irrelevant or clear from the context to avoid clutter.

\subsection{Domain Ordering}

We define a partial order $\rvord$ on $\rvalues[]$,
which is analogous to the partial order $\cvord$ on concrete
values, except that equality on constants is replaced by subset
inclusion on dependency relations:
\begin{align*}
\rval_1 \rvord \rval_2 \Def\iff {} & \rval_1 {=} \rbot \vee \rval_2 {=} \rerr \vee
	(\rval_1, \rval_2 \in \rels \wedge \rval_1 \subseteq \rval_2)
                                   \;\vee \\
&    (\rval_1,\rval_2 \in \rtables \wedge \forall \ccsite.\;
		\rval_1(\ccsite) \;\rvecord\; \rval_2(\ccsite)) 
\end{align*}
We note that aside from $\rbot$ and $\rerr$, only the relational
values of identical scope $\rscope$ are comparable. This ordering
induces for every $\rscope$, a complete lattice $(\rvalues, \rvord,
\rbot, \rerr, \rvjoin, \rvmeet)$. We assume the set of relational
values is implicitly quotiented under dependency variable renaming,
which can always be achieved.
\more{ The definition of join over relational
values is given below.
%
\[\begin{array}{cc}
\rbot \rvjoin \rval \Def= \rval \quad\quad \rval \rvjoin \rbot \Def= \rval & \rel_1 \rvjoin \rel_2 \Def= \rel_1 \cup \rel_2 \\
\rtable_1 \rvjoin \rtable_2 \Def= \lambda \ccsite.\, \rtable_1(\ccsite) \;\vecop{\rvjoin}\; \rtable_2(\ccsite) 
& \quad \rval_1 \rvjoin \rval_2 \Def= \rerr \quad \text{(otherwise)}
\end{array}\]
%
The meet $\rvmeet$ is defined similarly.
\begin{lemma}\label{thm:rvjoin}
Let $V \in \powerset(\rvalues)$. Then, $\rvjoin V = \mi{lub} (V)$ and 
$\rvmeet V = \mi{glb} (V)$ subject to $\rvord$.
\end{lemma}
}
We lift the lattices on relational values pointwise
to a complete lattice on relational maps $(\rmaps, \rmord,
\rmap_{\bot}, \rmap_{\cerr}, \rmjoin, \rmmeet)$.
%
%

\subsubsection*{Galois connections}

We can show that there exists a Galois connection between the complete lattice
of relational values and the superset of pairs of concrete maps and concrete values.
\begin{lemma}
  \label{lem:rgamma-meet-morphism}
  For all scopes $\rscope$, $\rgamma[]$ is a complete meet-morphism between
  $\rvalues$ and $\powerset((\rscope \to \cvalues) \times \cvalues)$ .
\end{lemma}

It follows then that there exists a Galois connection between
$\powerset((\rscope \to \cvalues) \times \cvalues)$ and $\rvalues$
for every scope
$\rscope$ which has $\rgamma[]$ as its right adjoint~\cite{cousot1979systematic}. Let $\ralpha:
\powerset((\rscope \to \cvalues) \times \cvalues) \to \rvalues$ be
the corresponding left adjoint, which is uniquely determined by
$\rgamma[]$.
From Lemma~\ref{lem:rgamma-meet-morphism} it easily follows that
$\rmgamma$ is also a complete meet-morphism. We denote by
$\rmalpha: \powerset(\cmaps) \to \rmaps$ the left adjoint of the
induced Galois connection.

\subsubsection*{Abstract domain operations} 
Before we can define the abstract transformer, we need to introduce a
few primitive operations for constructing and manipulating relational
values and execution maps (Fig.~\ref{fig:relational-ops}).
%
%
%
\begin{figure}[t]
\small
\begin{gather*}
\rupdate{\rval_{\rscope}}{\absn}{\rval'} \Def= 
\ralpha(\rgamma[](\rval_\rscope) \;\cap \setc{ \pair{\Gamma'[\absn \mapsto \cval']}{\cval}}{\pair{\Gamma'}{\cval'} \in \rgamma[](\rval')})
\\[0.3ex]
\requality{\rval_{\rscope}}{\absn_1}{\absn_2} \Def=
\ralpha(\rgamma[\iota](\rval_\rscope) \cap
\setc{\pair{\Gamma}{\cval}}{\Gamma(\absn_1) = \Gamma(\absn_2)})
\\[0.3ex]
\requality{\rval_{\rscope}}{\absn}{\cval} \Def=
\ralpha(\rgamma[](\rval_\rscope) \cap \setc{\pair{\Gamma}{\cval}}{\Gamma(\absn) = \cval})
\\[0.3ex]
\rstrengthen{\rval_{\rscope}}{\Gamma} \Def= 
\rvmeetb \setc{\rupdate{\rval}{\cnode}{\Gamma(\cnode)}}{\cnode \in \rscope}
\\[0.3ex]
\rrescope{\rval_{\rscope}}{\rscope'} \Def= 
\ralpha[\rscope'](\rgamma[](\rval_\rscope) \;\cap \setc{\pair{\Gamma}{\cval}}{\forall n \in \rscope' {\setminus} \rscope.\,
\Gamma(n) \notin \set{\bot,\cerr}})
\end{gather*}
\caption{Operations on relational abstract domains\label{fig:relational-ops}}
\end{figure}
The operation $\requality{\rval_{\rscope}}{\absn_1}{\absn_2}$
strengthens $\rval_{\rscope}$ by enforcing equality between $\absn_1$
and $\absn_2$ in the dependency vectors. Similarly,
$\requality{\rval_{\rscope}}{\absn}{\cval}$ strengthens
$\rval_{\rscope}$ by requiring that the value desribed by $\absn$ 
is $\cval$. We use the operation $\rupdate{\rval_{\rscope}}{\absn}{\rval'}$ to strengthen
$\rval_{\rscope}$ with $\rval'$ for $\absn$. Logically, this can be
viewed as computing $\rval_{\rscope} \land \rval'[\absn/\nu]$. 
\more{The operation $\rupdate{\rmap}{\cnode}{\rval_{\rscope}}$ lifts
$\rupdate{\rval_{\rscope}}{\cnode}{\rval'}$ pointwise to relational
execution maps. Note that here for every $\cnode'$ we implicitly
rescope $\cmap(\cnode')$ so that its scope includes $\cnode$.} 
The operation $\rstrengthen{\rval_{\rscope}}{\Gamma}$, conversely,
strengthens $\rval_{\rscope}$ with all the relations describing the
values at nodes $\rscope$ in $\Gamma:\rscope \to \rval$. Intuitively,
this operation will be used
to push scope/environment assumptions to a relational value. Finally, the \emph{rescoping}
operation $\rrescope{\rval_{\rscope}}{\rscope'}$ changes the scope of
$\rval_\rscope$ from $\rscope$ to $\rscope'$ while strengthening it with
the information that all nodes in $\rscope' \setminus \rscope$ are
neither $\bot$ nor $\cerr$. We use the operation $\exists \cnode.\,\rval_{\rscope} \Def= \rstrengthen{\rval}{\rscope\setminus\{\cnode\}}$
to simply project out, i.e., remove a node $\cnode$ from the
scope of a given relational value $\rval_\rscope$.
We (implicitly) use the rescoping operations to enable precise 
propagation between relational values that have incompatible scopes.

The relational abstraction of a constant $c$ in scope $\rscope$ is
defined as $\rconst_\rscope \Def= \requality{\rerr}{\nu}{\const} = \pset{\dmap \in \dmaps_\rscope}{\dmap(\nu) = c}$.

\subsection{Abstract Transformer}

The abstract transformer, $\rtransname$, of the relational semantics
is shown in Fig.~\ref{fig:relational-transformer}. It closely
resembles the concrete transformer, $\ctransname$, where relational
execution maps $\rmap$ take the place of concrete execution maps
$\cmap$ and relational values $\rval$ take the place of concrete
values $\cval$. 

The case for constant values $c^{\elabel}$ effectively replaces the join on values in 
the concrete transformer using a join on relational values. The
constant $c$ is abstracted by the relational value
$\rstrengthen{\rconst}{\Gamma}$, which establishes the relation between
$c$ and the other values stored at the nodesin $\cmap$ that are bound in the
environment $\cenv$. In general, the transformer \emph{pushes} environment
assumptions into relational values rather than leaving the assumptions implicit.
We decided to treat environment assumptions this way so we can interpret 
relational values independently of the environment.
%
%
This idea is implemented in the rules for variables and constants, as explained above, 
while it is done inductively for other constructs. 
In the case for variables $x^\elabel$, the abstract transformer
replaces $\ciprop$ in the concrete transformer by the abstract
propagation operator $\riprop$ defined in Fig.~\ref{fig:relational-propagation}.
Note that the definition of
$\riprop$ assumes that its arguments range over the same
scope. Therefore, whenever we use $\riprop$ in $\rtransname$, we
implicitly use the rescoping operator to make the arguments
compatible. For instance, in the call to $\riprop$ for variables
$x^\elabel$, the scope of $\rmap(\cenv(x))$ is $\rscope_{\cenv(x)}$
which does not include the variable node $\cenv(x)$ itself. We
therefore implicitly rescope $\rmap(\cenv(x))$ to
$\rrescope{\rmap(\cenv(x))}{\rscope_{\cnode}}$ before
strengthening it with the equality $\nu = \cenv(x)$. The other cases
are fairly straightforward. 
\more{ Observe how in the case for conditionals
$(\ifte{e_0^{\elabel_0}}{e_1^{\elabel_1}}{e_2^{\elabel_2}})^{\elabel}$
we analyze both branches and combine the resulting relational
execution maps with a join. In each branch, the map $\rmap_0$ is
strengthened with the information about the truth value stored at
$\cn{\cenv}{\elabel_0}$, reflecting the path-sensitive nature of
Liquid types.}
\begin{figure}[t]
\[\begin{array}[t]{@{}l|l}
\begin{array}[t]{@{}l}
\rtrans{c_\elabel}(\cenv, \cstack) \Def= \\
\sind \Mdo \cnode = \cn{\elabel}{\cenv} \,\Mend \Gamma^\rdesignation \mto
\menv(\cenv) \,\Mend
\cval' \mto \mupd{\cnode}{\rstrengthen{\rconst}{\Gamma^\rdesignation}} \\[0.4ex]
\sind \Mreturn \rval'\\[2ex]
\rtrans{x_\elabel}(\cenv, \cstack) \Def= \\
\sind \Mdo  \cnode = \cn{\elabel}{\cenv} \, \Mend \, \rval \mto \mread{\cnode}
\, \Mend \, \cnode_x = \cenv(x)
\, \Mend \, \Gamma^\rdesignation \mto \menv(\cenv) \, \\[0.4ex]
\sind[6] \rval' \mto \mupd{\cnode_x,\,
\cnode}{\rstrengthen{\Gamma^\rdesignation(x)\tvareq{x}}{\Gamma^\rdesignation} \riprop \rstrengthen{\rval\tvareq{x}}{\Gamma^\rdesignation}}\\
\sind \Mreturn \rval'\\[2ex]
%
\rtrans{(e_\elabela \, e_\elabelb)_\elabel}(\cenv, \cstack) \Def=  \\
\sind 	\Mdo  \cnode = \cn{\elabel}{\cenv} \,\Mend \cnode_\elabela =
\cn{\elabela}{\cenv}
		\,\Mend \cnode_\elabelb = \cn{\elabelb}{\cenv} \,
\Mend \rval \mto \mread{\cnode}\\[0.4ex]
\sind[6] \rval_\elabela \mto \rtrans[k]{e_\elabela}(\cenv, \cstack) \,\Mend
\massert{\rval_\elabela \in \ttables[]}\\[0.4ex]
\sind[6] \rval_\elabelb \mto \rtrans[k]{e_\elabelb}(\cenv, \cstack)\\[0.4ex]
\sind[6] \rval'_\elabela,\, \absn:[\tentry{\elabela \cdot \cstack}{\rval_\elabelb'}{\rval'}] = \rval_\elabela \riprop
\absn:[\tentry{\elabela \cdot \cstack}{\rval_\elabelb}{\rval}]\\[0.4ex]
\sind[6] \rval''	\mto \mupd{\cnode_\elabela,\, \cnode_\elabelb,\,
\cnode}{\rval'_\elabela,\, \rval'_\elabelb,\, \rval'}\\[0.4ex]
\sind \Mreturn \rval''
\end{array}
&
\begin{array}[t]{@{}l}
\rtrans{(\lambda x.e_\elabela)_\elabel}(\cenv, \cstack) \Def= \\ 
\sind	\Mdo \cnode = \cn{\elabel}{\cenv} \,\Mend \rval \mto
\mupd{\cnode}{x:\rtable_{\bot}} \\[0.5ex]
\sind[6] \rval' \mto \Mfor \rval \;\Mdo \rstepbody(x,
e_\elabela, \cenv, \rval)\\[0.5ex]
\sind[6] \rval'' \mto \mupd{\cnode}{\rval'}\\[0.5ex] 
\sind[2] \Mreturn \rval'' \\[2ex]
\rstepbody(x, e_\elabela, \cenv, \rval)(\cstack') \Def= \\
\sind  \Mdo  \cnode_x = \cns{x}{\cenv}{\cstack'} \,\Mend
                  \cenv_\elabela = \cenv.x\!:\!\cnode_x \,\Mend
\cnode_\elabela = \cn{\elabela}{\cenv_\elabela}\\[0.5ex]
\sind[6] \rval_x \mto \mread{\cnode_x} \,\Mend \rval_\elabela \mto \rtrans[k+1]{e_\elabela}(\cenv_\elabela,
\cstack') \\[0.5ex]
\sind[6]  x:[\ttab[\cstack']{\rval_x'}{\rval_\elabela'}], \rval' =
\\[0.5ex]
\sind[9]  x:[\ttab[\cstack']{\rval_x}{\rval_\elabela}] \riprop \restr{\rval}{\cstack'} \\[0.5ex]
\sind[6]  \mupd{\cnode_x,\, \cnode_\elabela}{\rval'_x,\, \rval'_\elabela} \\[0.5ex]
\sind  \Mreturn \rval'
\end{array}
\end{array}\]
\caption{Abstract transformer for relational data flow semantics\label{fig:relational-transformer}}
\end{figure}
\begin{figure}[t]
\vspace{-0.4cm}
\begin{minipage}[t]{.6\textwidth}
\[\begin{array}{@{}l}
\absn:\rtable_1 \riprop \absn:\rtable_2 \Def= \; \\
\sind \Mlet \rtable = \Lambda \ccsite.	\\[0.4ex]
\sind[2] \Mlet \pair{\rval_{1i}}{\rval_{1o}} = \rtable_1(\ccsite) \Mend \pair{\rval_{2i}}{\rval_{2o}} = \rtable_2(\ccsite) \\[0.4ex]
\sind[3] \;\;\pair{\rval'_{2i}}{\rval'_{1i}} = \rval_{2i} \riprop \rval_{1i} \\[0.4ex]
\sind[3] \;\; \pair{\rval'_{1o}}{\rval'_{2o}} = 
		\rupdate{\rval_{1o}}{\absn}{\rval_{2i}} \riprop \rupdate{\rval_{2o}}{\absn}{\rval_{2i}}\\[0.4ex]
\sind[2] \Min (\pair{\rval'_{1i}}{\rval_{1o} \rvjoin \rval'_{1o}}, \pair{\rval'_{2i}}{\rval_{2o} \rvjoin \rval'_{2o}}) \quad\quad\quad \ind[4] \\[0.4ex]
\sind[1] \Min \pair{\absn:\Lambda \ccsite.\; \pi_1(\rtable(\ccsite))}{\absn:\Lambda \ccsite.\; \pi_2(\rtable(\ccsite))}
\end{array}\]
\end{minipage}%
\begin{minipage}[t]{.39\textwidth}
\[\begin{array}{@{}l}
\rtable \riprop \rbot \Def=\; \pair{\rtable}{\rtable_{\bot}} \\[5ex]
\rtable \riprop \rerr \Def=\; \pair{\rerr}{\rerr} \\[5ex]
\rval_1 \riprop \rval_2 \Def=\; \pair{\rval_1}{\rval_1 \rvjoin \rval_2} \\
\multicolumn{1}{c}{\textbf{(otherwise)}}
\end{array}\]
\end{minipage}
\caption{Value propagation in relational data flow semantics\label{fig:relational-propagation}}
\end{figure}

\begin{lemma}
  \label{thm:r-prop-increasing-monotone}
  The function $\riprop$ is monotone and increasing.
\end{lemma}

\begin{lemma}
  \label{thm:r-step-increasing-monotone}
  For every  $e \in \Exp$, $\cstack \in \cstacks$, $\cenv \in \cenvironments$
  such that $\pair{e}{\cenv}$ is well-formed, then
  $\rtrans{e}(\cenv, \cstack)$ is monotone and increasing.
\end{lemma}

\subsection{Abstract Semantics}

The abstract domain of the relational abstraction is given by
\emph{relational properties} 
$\rdomain \Def= \rmaps$.
The relational abstract semantics $\rcoll: \Exp \to \rdomain$ is then defined as the least
fixpoint of $\rtransname$ as follows:
\[\rcoll \llbracket e \rrbracket \Def= \mathbf{lfp}^{\rmord}_{\rmap_{\bot}} 
\Lambda \rmap.\; \Mlet \pair{\_}{\rmap_1} = \rtrans[\fuel]{e}(\epsilon, \epsilon)(\rmap) \,\Min \rmap_1 \]
\begin{theorem}
  \label{thm:relational-soundness}
  The relational abstract semantics is sound, i.e., for all programs
  $e$, $\ccoll \llbracket e \rrbracket \subseteq
  \rmgamma(\rcoll \llbracket e \rrbracket)$.
\end{theorem}
\section{Collapsed Relational Data Flow Semantics}
\label{appendix:collapsed-semantics}

\subsection{Notation}
%
The notations and definitions for the collapsed semantics are
as for relational values except that abstract nodes take the
place of concrete nodes.

\subsection{Domain Ordering}

The ordering $\pvord$ on collapsed values $\pvalues$ resembles the
ordering on relational values:
\begin{align*}
\pval_1 \pvord[] \pval_2 \Def\iff & \pval_1 = \pbot \vee \pval_2 = \perr \vee
	(\pval_1, \pval_2 \in \prels \wedge \pval_1 \subseteq \pval_2) \;\vee\\ 
	&(\pval_1, \pval_2 \in \ptables \wedge \forall \pcsite.\;\pval_1(\pcsite) \;\pvecord\; \pval_2(\pcsite)) 
\end{align*}
We assume that when we compare two collapsed tables, we implicitly
apply $\alpha$-renaming so that they range over the same dependency
variables. Again, this ordering induces, for every $\pscope$, a
complete lattice $(\pvalues, \pvord[], \pbot, \perr,\pvjoin[], \pvmeet[])$\footnote{Technically, 
$\pvord[]$ is only a pre-order due to $\alpha$-renaming of tables.}. 
\more{ The definition of the join $\pvjoin[]$ operator is as follows.
\[\begin{array}{cc}
\pbot \pvjoin[] \pval \Def= \pval \quad\quad \pval \pvjoin[] \pbot \Def= \pval &
           \prel_1 \pvjoin[] \prel_2 \Def= \prel_1 \cup \prel_2 \\
\ptab{\pval_1}{\pval_2} \pvjoin[] \ptab{\pval_3}{\pval_4} \Def= 
	\ptab{\pval_1\pvjoin[]\pval_3}{\pval_2\pvjoin[]\pval_4} 
& \pval_1 \pvjoin[] \pval_2 \Def= \perr \quad \text{(otherwise)}
\end{array}\]
The meet $\pvmeet[]$ is defined similarly.
\begin{lemma}\label{thm:clvjoin}
Let $V \in \powerset(\pvalues)$. Then, $\pvjoin[] V = \mi{lub} (V)$ and 
$\pvmeet[] V = \mi{glb} (V)$ subject to $\pvord[]$.
\end{lemma}
}
\noindent We implicitly identify $\pvalues$
  with its quotient subject to the equivalence relation induced by $\pvord[]$.
These lattices are lifted pointwise to a complete lattice on collapsed
execution maps
$(\pmaps, \pmord, \pmap_{\bot}, \pmap_{\cerr}, \pmjoin, \pmmeet)$.

%
We first formally state that our complete lattice of collapsed values
forms a Galois connection to the complete lattice of relational values.
\begin{lemma}
\label{lemma:coll-v-galois}
For every abstract scope $\pscope$,
$\pgamma[\pscope]$ is a complete meet-morphism between $\pvalues[\pscope]$ and $\rvalues[]$.
\end{lemma}
It follows then that there exists a unique Galois connection
$\pair{\palphac}{\pgammac}$ between $\rvalues$ and
$\pvalues$. 
From Lemma~\ref{lemma:coll-v-galois} it easily follows that
$\pmgamma$ is also a complete meet-morphism. We denote by
$\pmalpha: \rmaps \to \pmaps$ the left adjoint of the
induced Galois connection.

\subsubsection*{Abstract domain operations}
In the new abstract transformer we replace the strengthening 
operations on relational values 
by their most precise abstractions on collapsed values.
$\pgammac_\delta$ are simply lifted pointwise. For instance,
strengthening with equality between dependency variables,
$\requality{\pval_{\pscope}}{\absn_1}{\absn_2}$, 
is defined as
\begin{align*}
\requality{\pval_{\pscope}}{\absn_1}{\absn_2} \Def= {} & \palpha(\requality{\pgamma(\pval)}{\absn_1}{\absn_2})
\end{align*}

\more{This definition in terms of most precise abstraction should not
obscure the simple nature of this operation. E.g., for collapsed
dependency relations $\prel$ we simply have:
\[\pstrengthen{\prel_{\pscope}}{\pnode_1{=}\pnode_2} =
  \pset{\pdmap {\in} \prel_{\pscope}}{\pnode_1,\pnode_2 {\in} \pscope \Rightarrow
    \pdmap(\pnode_1){=}\pdmap(\pnode_2)}
\]
as one might expect. The other strengthening operations on relational
values are abstracted in a similar manner.
}

\subsection{Abstract Transformer}

The abstract transformer for the collapsed semantics,
$\ctransname^{\pdesignation}: \Exp \to \penvironments \times \pstacks \to
\pmaps \to \pvalues \times \pmaps$, resembles the relational abstract
transformer $\rtransname$. That is, 
the operations on relational values in $\rtransname$
are simply replaced by their abstractions on collapsed values in $\ptransname$.
We omit the definitions of the transformer and data propagation as they
follow straightforwardly from the relational ones.
\begin{lemma}
  \label{thm:p-prop-increasing-monotone}
  The function $\pprop$ is monotone and increasing.
\end{lemma}

\begin{lemma}
  \label{thm:p-step-increasing-monotone}
  For every  $e \in \Exp$, $\pstack \in \pstacks$, and
  $\penv \in \penvironments$
  such that $\pair{e}{\penv}$ is well-formed, then
  $\ptrans{e}(\penv, \pstack)$ is monotone and increasing.
\end{lemma}
In the above, the well-formedness of a pair $(e, \penv)$ of an expression
$e$ and abstract environment $\penv$ is defined
in a similar way as for the concrete environments.

\subsection{Abstract Semantics}

The abstract domain of the collapsed abstract semantics is given by
\textit{collapsed properties} $\pdomain \Def= \pmaps$. The collapsed semantics
$\pcoll: \Exp \to \pdomain$ is then defined as
\[\pcoll[e] \Def= \mathbf{lfp}^{\dot{\sqleq}^{\pdesignation}}_{\pmap_{\bot}} 
\Lambda \pmap.\; \Mlet \pair{\_}{\pmap_1} = \ptrans[\fuel]{e}(\epsilon, \epsilon)(\pmap) \,\Min \pmap_1.\]
%
%
%
\begin{theorem}
  \label{thm:collapsed-semantics-soundness}
  The collapsed abstract semantics is sound, i.e., for all
  programs $e$, $\rcoll \llbracket e \rrbracket \rmord \pmgamma(\pcoll[e]).$
\end{theorem}

\section{Proofs}
\label{appendix:proofs}
\subsection{Preliminaries}

We first repeat some of the basic results regarding abstract
interpretation that we use in our proofs.
\begin{proposition}[\cite{cousot1979systematic}]
  \label{prop:galois-connections}
  The following statements are equivalent:
  \begin{enumerate}
  \item $\pair{\alpha}{\gamma}$ is a Galois connection
  \item $\alpha$ and $\gamma$ are monotone,
    $\alpha \circ \gamma$ is reductive: $\forall y \in
    L_2.\, \alpha(\gamma(y)) \sqleq_2 y$, and
    $\gamma \circ \alpha$ is extensive: $\forall x \in
    L_1.\, x \sqleq_1 \gamma(\alpha(x))$
  \item 
    $\gamma$ is a complete meet-morphism and
    $\alpha = \Lambda x \in L_1.\, \bigsqcap_2 \pset{y \in L_2}{x
      \sqleq_1 \gamma(y)}$
  \item $\alpha$ is a complete join-morphism and
    $\gamma = \Lambda y \in L_2. \bigsqcup_1 \pset{x \in L_1}{\alpha(x)
      \sqleq_2 y}$.
  \end{enumerate}
\end{proposition}

\begin{proposition}[\cite{cousot1979systematic}]
  \label{prop:galois-connections-2}
  $\alpha$ is onto iff $\gamma$ is one-to-one iff $\alpha \circ \gamma = \lambda y.\,y$.
\end{proposition}

We additionally use $\locs(e)$ to denote the set of all
locations of $e$. Next, given an expression $e$ with 
unique locations and a location  $\elabel \in \locs(e)$,
we use $e(\elabel)$ in our proofs to designate 
the subexpression of $e$ with the location $\elabel$.



\subsection{Concrete Semantics}

\begin{proof}[Proof of Lemma~\ref{prop-monotone}]
By mutual structural induction on both input concrete values.
\end{proof}

\begin{proof}[Proof of Lemma~\ref{step-monotone}]
By structural induction on $e$ and Lemma~\ref{prop-monotone}.
\end{proof}

\subsection{Relational Semantics}

\begin{proof}[Proof of Lemma~\ref{lem:rgamma-meet-morphism}]
Let $U \in \powerset(\rvalues)$. We show 
$\rgamma(\rvmeet U) = \bigcap_{\rval \in U} \rgamma(\rval)$.

We carry the proof by case analysis on elements of $U$ and 
induction on the minimum depth thereof when $U$ consists of
tables only. The depth of a relational value is defined in the 
expected way: the depth of bottom, relational, and top values
is 0, whereas depth of a table is defined by the maximum depth
of any value stored in the table.

We first consider the trivial (base) cases where $U$ is not a set of multiple relational tables.
\begin{itemize}
\small
\setlength{\abovedisplayskip}{0pt}
\setlength{\belowdisplayskip}{0pt}

\item $U = \emptyset$. Here, $\rgamma(\rvmeet \emptyset) = 
		\rgamma(\rerr) = (\rscope \to \cvalues) \times \cvalues = \bigcap \emptyset$.
\item $U = \{\rval\}$. Trivial.
\item $\rel \in U, \rtable \in U$. We have $\rvmeet \{\rel, \rtable\} = \rbot$
        and since $\rbot$ is the bottom element, $\rvmeet U = \rbot$. Similarly,
	$\rgamma(\rel) \cap \rgamma(\rtable) = \rgamma(\rbot)$ and since $\forall \rval.\, \rgamma(\rbot) \subseteq \rgamma(\rval)$,
	it follows $\bigcap_{\rval \in U} \rgamma(\rval) = \rgamma(\rbot) = \rgamma(\rvmeet U)$.

\item $U \subseteq \rels$. Here,
{\footnotesize \begin{gather*}
\rgamma(\rvmeet U) = \rgamma(\bigcap_{\rel \in U}\,\rel) \\
= \setc{\pair{\scmap}{c}}{\dmap \in (\bigcap_{\rel \in U}\,\rel)) \wedge \dmap(\nu) = c \wedge \forall x \in \rscope.\scmap(x) \in \dgamma(\dmap(x))} \cup \rgamma(\rbot) \\
\ind[15] \textbf{\scriptsize [by def. of $\rgamma$]} \\
= \bigcap_{\rel \in U}  (\setc{\pair{\scmap}{c}}{\dmap \in \rel \wedge \dmap(\nu) = c \wedge \forall x \in \rscope.\scmap(x) \in \dgamma(\dmap(x))} \cup \rgamma(\rbot)) \\
\ind[12] \textbf{\scriptsize [by uniqueness/non-overlapping of $\dgamma$]} \\
=  \bigcap_{\rel \in U} \rgamma(\rel)
\ind[13] \textbf{\scriptsize [by def. of $\rgamma$]}
\end{gather*}}%

\item $\rerr \in U$. Here, $ \rvmeet U = \rvmeet (U/\{\rerr\})$ and $\bigcap_{\rval \in U} \rgamma(\rval) = \bigcap_{\rval \in U, \rval \not = \rerr} \rgamma(\rval)$ since $\rgamma(\rerr) = (\rscope \to \cvalues) \times \cvalues$. The set $U/\{\rerr\}$ either falls into one of the above cases or consists of multiple tables, which is the case we show next.
\end{itemize}

Let $U \subseteq \rtables$ and $|U| > 1$. Let $d$ be the minimum depth of any table in $U$.
{\footnotesize
\begin{gather*}
\bigcap_{x:\rtable \in U} \rgamma(x:\rtable) = \{ \pair{\scmap}{\ctable}\;|\; \forall \ccsite.\;
  \ctable(\ccsite) {=} \pair{\cval_i}{\cval_o} \wedge \pair{\scmap_i}{\cval_i} \in \bigcap_{\rtable \in U} \rgamma[](\pi_1(\rtable(\ccsite))) \;\wedge \\
   \pair{\scmap_o}{\cval_o} \in \bigcap_{\rtable \in U} \rgamma[](\pi_2(\rtable(\ccsite))) \wedge \scmap_o {=} \scmap_i[x \mapsto \cval_i]\} \cup \rgamma(\rbot) \\
\ind[10] \textbf{\scriptsize [by def. of $\rgamma$ and $x$ not in the scope of $\pi_1(\rtable(\ccsite))$]} \\
= \{ \pair{\scmap}{\ctable}\;|\; \forall \ccsite.\;
  \ctable(\ccsite) {=} \pair{\cval_i}{\cval_o} \wedge \pair{\scmap_i}{\cval_i} \in \rgamma(\rvmeet_{\rtable \in U} \pi_1(\rtable(\ccsite))) \;\wedge \\
   \pair{\scmap_o}{\cval_o} \in \rgamma(\rvmeet_{\rtable \in U} \pi_2(\rtable(\ccsite))) \wedge \scmap_o {=} \scmap_i[x \mapsto \cval_i]\} \cup \rgamma(\rbot) \\
\textbf{\scriptsize [by i.h. on the min. depth of $\setc{\pi_1(\rtable(\ccsite)))}{\rtable \in U}$ and $\setc{\pi_2(\rtable(\ccsite)))}{\rtable \in U}$]} \\
= \rgamma(\rvmeet_{x:\rtable \in U} x:\rtable) \ind[20] \textbf{\scriptsize [by def. of $\rgamma$]}
\end{gather*}}%
\end{proof}

\begin{proof}[Proof of Lemma~\ref{thm:r-prop-increasing-monotone}]
By mutual structural induction on both input relational values.
\end{proof}

\begin{proof}[Proof of Lemma~\ref{thm:r-step-increasing-monotone}]
By structural induction on $e$ and Lemma~\ref{thm:r-prop-increasing-monotone}.
\end{proof}


\subsection{Collapsed Relational Semantics}

\begin{proof}[Proof of Lemma~\ref{lemma:coll-v-galois}]
Similar to the proof of Lemma~\ref{lem:rgamma-meet-morphism}.
\end{proof}

\begin{proof}[Proof of Lemma~\ref{thm:p-prop-increasing-monotone}]
By mutual structural induction on both input types.
\end{proof}

\begin{proof}[Proof of Lemma~\ref{thm:p-step-increasing-monotone}]
By structural induction on $e$ and Lemma~\ref{thm:p-prop-increasing-monotone}.
\end{proof}


\subsection{Parametric Refinement Type Semantics}

%

\begin{proof}[Proof of Lemma~\ref{thm:t-meet-morph-values}]
Similar to the proofs for Lemma~\ref{lem:rgamma-meet-morphism} and
Lemma~\ref{lemma:coll-v-galois}.
The argument is based on the definition of $\tvgamma[]$ and the fact
that type refinements from a Galois connections with the complete
lattice of dependency relations.
\end{proof}





\subsection{Soundness and Completeness of Typing Rules}

\begin{lemma}\label{lem:strengthening}
  For all $\pscope \subseteq \vars$, $\tval_1,\tval_2,\tval_1'\tval_2',\tval \in \tvalues$ and
  $x,y \in \pscope$, the following are true:
  \begin{enumerate}[label=(\arabic*)]
  \item strengthening is monotone: if $\tval_1 \tvord[] \tval_2$ then $\tupdate{\tval_1}{x}{\tval} \tvord[] \tupdate{\tval_2}{x}{\tval}$\label{lem:strengthening-monotone}
  \item strengthening is reductive: $\tupdate{\tval_1}{x}{\tval}
    \tvord[] \tval_1$\label{lem:strengthening-reductive}
  \item strengthening is idempotent: $\tupdate{\tupdate{\tval_1}{x}{\tval}}{x}{\tval} = \tupdate{\tval_1}{x}{\tval}$\label{lem:strengthening-idempotent}
  \item strengthening is commutative: $\tupdate{\tupdate{\tval}{x}{\tval_1}}{y}{\tval_2} = \tupdate{\tupdate{\tval}{y}{\tval_2}}{x}{\tval_1}$\label{lem:strengthening-commutative}
  \end{enumerate}
\end{lemma}

\begin{proof}
The proof goes straightforwardly by structural induction over
types. In the base case all properties follow directly from the
properties of meets.
\end{proof}

\begin{lemma}\label{lem:prop-strengthening}
  For all $\pscope \subseteq \vars$, $\tval_1,\tval_2,\tval_1',\tval_2',\tval \in \tvalues$ and $x
  \in \pscope$, if $\pair{\tval_1'}{\tval_2'} = \tupdate{\tval_1}{x}{\tval} \tiprop
    \tupdate{\tval_2}{x}{\tval}$ and $\tval_1'$ and $\tval_2'$ are safe, then $\tval_1' = \tupdate{\tval_1'}{x}{\tval}$ and $\tval_2' = \tupdate{\tval_2'}{x}{\tval}$.
\end{lemma}

\begin{proof}
  The proof goes by simultaneous induction over the depth of $\tval_1$
  and $\tval_2$.
  Assume $\pair{\tval_1'}{\tval_2'} = \tupdate{\tval_1}{x}{\tval}
  \tiprop \tupdate{\tval_2}{x}{\tval}$ and that $\tval_1'$ and $\tval_2'$ are safe.
  We case split on the definition of $\piprop$.
  
  \paragraph{Case $\tupdate{\tval_2}{x}{\tval} = \tbot$ and
    $\tupdate{\tval_1}{x}{\tval} \notin \bdomain$.} We have
  $\tval_2' = \ttable_\bot = \tupdate{\ttable_\bot}{x}{\tval}$ by
  definition of $\tiprop$ and because strengthening is
  reductive. Moreover, we have $\tval_1' =
  \tupdate{\tval_1}{x}{\tval}$ and by idempotency of strengthening
  we obtain $\tval_1' = \tupdate{\tval_1}{x}{\tval'}$.

  \paragraph{Case $\tupdate{\tval_2}{x}{\tval} = \ttop$ and $\tupdate{\tval_1}{x}{\tval}  \notin  \bdomain$.}
  By definition of $\piprop$ we have $\tval_2' = \btop$, contradicting
  the assumption that $\tval_2'$ is safe.
  
  \paragraph{Case $\tupdate{\tval_1}{x}{\tval} \notin \bdomain$ and
    $\tupdate{\tval_2}{x}{\tval}  \notin  \bdomain$.}
  We must have $\tval_1 = z:\ttable_1$,
  $\tval_2 = z:\ttable_2$, $\tval_1' = z:{\ttable_1}'$, and
  $\tval_2' = z:{\ttable_2}'$ for some
  $\ttable_1,\ttable_2,{\ttable_1}',{\ttable_2}'$ and $z$. 

  Let $\pstack \in \pstacks$ and define:
  \begin{align*}
    \pair{\tval_{i1}}{\tval_{o1}} = {} &  \ttable_1(\pstack) &
    \pair{\tval_{i2}'}{\tval_{i1}'} = {} & \tupdate{\tval_{i2}}{x}{\tval} \tiprop
\tupdate{\tval_{i1}}{x}{\tval} & (a)\\
    \pair{\tval_{i2}}{\tval_{o2}} = {} & \ttable_2(\pstack) &
    \pair{\tval_{o1}'}{\tval_{o2}'} = {} &
    \tupdate{\tupdate{\tval_{o1}}{x}{\tval}}{z}{\tval_{i2}}  \tiprop
\tupdate{\tupdate{\tval_{o2}}{x}{\tval}}{z}{\tval_{i2}} & (b)
    \end{align*}

  Note that we have ${\ttable_1}'(\pstack) = \pair{\tval_{i1}'}{
    \tupdate{\tval_{o1}}{x}{\tval} \tvjoin[] \tval_{o1}'}$ and
  likewise ${\ttable_2}'(\pstack) = \pair{\tval_{i2}'}{
    \tupdate{\tval_{o2}}{x}{\tval} \tvjoin[] \tval_{o2}'}$. From the
  fact that $\tval_1'$ and $\tval_2'$ are safe, it follows that also
  $\tval_{i1}',\tval_{i2}',\tval_{o1}',\tval_{o2}'$ must be safe. Applying
  the induction hypothesis to $(a)$, we can directly conclude that
  $\tval_{i1}' = \tupdate{\tval_{i1}'}{x}{\tval}$ and $\tval_{i2}' =
  \tupdate{\tval_{i2}'}{x}{\tval}$.

  It remains to show that
  $\tupdate{\tval_{o1}}{x}{\tval} \tvjoin[]
    \tval_{o1}' = \tupdate{(\tupdate{\tval_{o1}}{x}{\tval} \tvjoin[]
    \tval_{o1}')}{x}{\tval}$ and $\tupdate{\tval_{o2}}{x}{\tval} \tvjoin[]
    \tval_{o2}' = \tupdate{(\tupdate{\tval_{o2}}{x}{\tval} \tvjoin[]
      \tval_{o2}')}{x}{\tval}$. The right-to-left direction of these equalities follows
    from
    Lemma~\ref{lem:strengthening}\ref{lem:strengthening-reductive}. For
    the other direction we first apply commutativity of
    strengthening
    (Lemma~\ref{lem:strengthening}\ref{lem:strengthening-commutative})
    to $(b)$ and then use the induction hypothesis to obtain 
    $\tval_{o1}' = \tupdate{\tval_{o1}'}{x}{\tval}$ and $(d)$ $\tval_{o2}' =
    \tupdate{\tval_{o2}'}{x}{\tval}$. The desired inclusions then
    follow from the derived equalities as well as
    Lemma~\ref{lem:strengthening}\ref{lem:strengthening-monotone}
    and Lemma~\ref{lem:strengthening}\ref{lem:strengthening-idempotent}.
    
    \paragraph{Otherwise.} We must have
    $\tupdate{\tval_2}{x}{\tval} \in \bdomain$ and by definition of
    $\tiprop$, we have $\tval_1' = \tupdate{\tval_1}{x}{\tval}$ and
    $\tval_2' = \tupdate{\tval_1}{x}{\tval} \tvjoin[]
    \tupdate{\tval_2}{x}{\tval}$. Note that by idempotency of
    strengthening
    (Lemma~\ref{lem:strengthening}\ref{lem:strengthening-idempotent})
    we directly have $\tval_1' = \tupdate{\tval_1'}{x}{\tval}$.

    To also show the desired equality for $\tval_2'$, let us first
    assume that $\tupdate{\tval_1}{x}{\tval} \notin \bdomain$. Then by
    we must either have $\tupdate{\tval_1}{x}{\tval} = \tbot$ or, by
    definition of $\tvord[]$ and $\tvjoin[]$, $\tval_2' = \ttop$. In
    the first case, we directly obtain
    $\tval_2' = \tupdate{\tval_1'}{x}{\tval}$. The second case contradicts the
    assumption that $\tval_2'$ is safe. Hence, consider the remaining
    case that $\tupdate{\tval_1}{x}{\tval} \in \bdomain$. Then we have:
  \begin{align*}
     && \tval_2' = {} & \tupdate{\tval_1}{x}{\tval} \tvjoin[]
  \tupdate{\tval_2}{x}{\tval}\\
     &\Rightarrow & \tval_2' & {} = \tupdate{\tval_1}{x}{\tval} \bjoin
  \tupdate{\tval_2}{x}{\tval}\\
     &\Rightarrow & \tupdate{\tval_2'}{x}{\tval} & {} = \tupdate{(\tupdate{\tval_1}{x}{\tval} \bjoin
  \tupdate{\tval_2}{x}{\tval})}{x}{\tval}\\
     &\Rightarrow & \tupdate{\tval_2'}{x}{\tval} & {} \bordinv
  \tupdate{\tupdate{\tval_1}{x}{\tval}}{x}{\tval} \bjoin
  \tupdate{\tupdate{\tval_2}{x}{\tval}}{x}{\tval}
  & \text{(Lemma~\ref{lem:strengthening}\ref{lem:strengthening-monotone})}\\
      &\Rightarrow & \tupdate{\tval_2'}{x}{\tval} & {} \bordinv
  \tupdate{\tval_1}{x}{\tval} \bjoin \tupdate{\tval_2}{x}{\tval}
  & \text{(Lemma~\ref{lem:strengthening}\ref{lem:strengthening-idempotent})}
  \\
  &\Rightarrow & \tupdate{\tval_2'}{x}{\tval} & {} \bordinv
  \tval_2'  
  \end{align*}
  We obtain $\tval_2' \bord \tupdate{\tval_2'}{x}{\tval}$ directly
  from the fact that strengthening is reductive (Lemma~\ref{lem:strengthening}\ref{lem:strengthening-reductive}). Thus, we conclude
  $\tval_2' = \tupdate{\tval_2'}{x}{\tval}$.
\end{proof}

\begin{proof}[Proof of Lemma~\ref{lem:subtyping-prop-fixpoint}]
  Let $\tval_1,\tval_2 \in \tvalues$.
  The proof goes by simultaneous induction over the depth of
  $\tval_1$ and $\tval_2$.

  \paragraph{Case $\tval_1 = \tbot$.} For the left-to-right direction,
  assume $\tval_1 \subtype \tval_2$. Only rule \textsf{s-bot} applies,
  so we must have $\tval_2 \neq \ttop$. Hence, both $\tval_1$ and
  $\tval_2$ are safe. Moreover, by the definition of $\tiprop$ we have
  $\tbot \tiprop \tval_2 = \pair{\tbot}{\tbot \tvjoin[] \tval_2} = \pair{\tbot}{\tval_2}$.

  For the other direction, assume that $\pair{\tval_1}{\tval_2} = \tval_1 \tiprop \tval_2$ and
  $\tval_1,\tval_2$ are safe. Then we have $\tval_2 \neq
  \ttop$. Hence, using rule \textsf{s-bot} we immediately conclude $\tbot \subtype \tval_2$.

  \paragraph{Case $\tval_1 \neq \tbot$ and $\tval_2 \in \bdomain$.}
  For the left-to-right direction, assume $\tval_1 \subtype \tval_2$.
  Since $\tval_2 \in \bdomain$ only rule \text{s-base} applies. Hence,
  we must have $\tval_1 \in \bdomain$, $\tval_1 \bord \tval_2$. 
  Since $\tval_1,\tval_2 \in \bdomain$, both are safe.
  It further
  follows from the definition of $\tiprop$ that
  $\tval_1 \tiprop \tval_2 = \pair{\tval_1}{\tval_1 \tvjoin[]
    \tval_2}$. Moreover, $\tval_1 \bord \tval_2$ implies
  $\tval_2 = \tval_1 \bjoin \tval_2 = \tval_1 \tvjoin[] \tval_2$.

  For the other direction, assume that
  $\pair{\tval_1}{\tval_2} = \tval_1 \tiprop \tval_2$ and
  $\tval_1,\tval_2$ are safe. In particular, it follows that
  $\tval_2 \neq \ttop$. From the definition of $\tiprop$ we can further conclude
  that we must have $\tval_2 = \tval_1 \tvjoin[] \tval_2$. The facts 
  $\tval_2 \neq \ttop$, $\tval_1 \neq \tbot$, and the definitions of $\tvord[]$
  and $\tvjoin[]$ imply that $\tval_1 \in \bdomain$. Hence, we have
  $\tval_2 = \tval_1 \bjoin \tval_2$ which implies $\tval_1 \bord \tval_2$.

  \paragraph{Case $\tval_1 \neq \tbot$ and $\tval_2 \notin \bdomain$.}
  For the left-to-right direction, assume again
  $\tval_1 \subtype \tval_2$. Only rule \textsc{s-fun} applies. Hence,
  we must have $\tval_1 = x:\ttable_1$ and
  $\tval_2 = x:\ttable_2$. Let
  $\pair{x:{\ttable_1}'}{x:{\ttable_2}'} = \tval_1 \tiprop
  \tval_2$. We show for all $\pstack \in \pstack$,
  ${\ttable_1}'(\pstack)=\ttable_1(\pstack)$ and
  ${\ttable_2}'(\pstack)=\ttable_2(\pstack)$. Thus, let
  $\pstack \in \pstacks$ and define
    \begin{align*}
    \pair{\tval_{i1}}{\tval_{o1}} = {} &  \ttable_1(\pstack) &
    \pair{\tval_{i2}'}{\tval_{i1}'} = {} & \tval_{i2} \tiprop
\tval_{i1}\\
    \pair{\tval_{i2}}{\tval_{o2}} = {} & \ttable_2(\pstack) &
    \pair{\tval_{o1}'}{\tval_{o2}'} = {} &
\tupdate{\tval_{o1}}{x}{\tval_{i2}}  \tiprop \tupdate{\tval_{o2}}{x}{\tval_{i2}}
    \end{align*}
    We know from the definition of $\tiprop$ that
    ${\ttable_1}'(\pstack) = \pair{\tval_{i1}'}{\tval_{o1} \tvjoin[]
      \tval_{o1}'}$ and likewise ${\ttable_2}'(\pstack) = \pair{\tval_{i2}'}{\tval_{o2} \tvjoin[]
      \tval_{o2}'}$.

    From the definition of $\subtype$ we further know
    \[\tval_{i2} \subtype \tval_{i1} \quad \text{and} \quad \tval_{o1} \subtype \tval_{o2}\]
    From the induction hypothesis it then follows that
    \begin{align*}
      \pair{\tval_{i2}}{\tval_{i1}} = \tval_{i2} \tiprop \tval_{i1}
\quad \text{and} \quad \pair{\tval_{o1}}{\tval_{o2}} = \tval_{o1} \tiprop \tval_{o2}
    \end{align*}
    Thus, we can directly conclude $\tval_{i1}' = \tval_{i1}$ and
    $\tval_{i2}' = \tval_{i2}$. Moreover,
    $\tval_{i1},\tval_{i2},\tval_{o1}$, and $\tval_{o2}$ must all be
    safe. Further note that we have
    $\tupdate{\tval_{o1}}{x}{\tval_{i2}} \tvord[] \tval_{o1}$ and
    $\tupdate{\tval_{o2}}{x}{\tval_{i2}} \tvord[] \tval_{o2}$ because
    strengthening is reductive (Lemma~\ref{lem:strengthening}\ref{lem:strengthening-reductive}). By monotonicity of $\tiprop$ we can
    therefore conclude $\tval_{o1}' \tvord[] \tval_{o1}$ and
    $\tval_{o2}' \tvord[] \tval_{o2}$. This implies that
    $\tval_{o1} \tvjoin[] \tval_{o1}' = \tval_{o1}$ and
    $\tval_{o2} \tvjoin[] \tval_{o2}' = \tval_{o2}$ hold.

    For the other direction, assume that
    $\pair{\tval_1}{\tval_2} = \tval_1 \tiprop \tval_2$ and that
    $\tval_1$ and $\tval_2$ are safe. It follows from
    $\tval_1 \neq \tbot$ and $\tval_2 \notin \bdomain$ as well as the
    definition of $\tiprop$ that we must have $\tval_1 = x:\ttable_1$
    and $\tval_2 = x:\ttable_2$. Let $\pstack \in \pstacks$ and define
    \begin{align*}
    \pair{\tval_{i1}}{\tval_{o1}} = {} &  \ttable_1(\pstack) &
    \pair{\tval_{i2}'}{\tval_{i1}'} = {} & \tval_{i2} \tiprop
\tval_{i1} & (a) \\
    \pair{\tval_{i2}}{\tval_{o2}} = {} & \ttable_2(\pstack) &
    \pair{\tval_{o1}'}{\tval_{o2}'} = {} &
\tupdate{\tval_{o1}}{x}{\tval_{i2}}  \tiprop
\tupdate{\tval_{o2}}{x}{\tval_{i2}} & (b)
    \end{align*}
    From the assumption that $\tval_1$ and $\tval_2$ are safe we can
    further conclude that $\tval_{i1}, \tval_{i2}, \tval_{o1}$, and
    $\tval_{o2}$ are safe. Moreover, by the fact that strengthening is
    reductive, it follows that $\tupdate{\tval_{o1}}{x}{\tval_{i2}}$
    and $\tupdate{\tval_{o2}}{x}{\tval_{i2}}$ must be safe, too.  Also
    note that $\pair{\tval_1}{\tval_2} = \tval_1 \tiprop \tval_2$
    additionally implies the following equalities
    \begin{align*}
      \tval_{i1} = {} & \tval_{i1}'  \quad (c) & \tval_{o1} = {} &
\tval_{o1} \tvjoin[] \tval_{o1}' \quad (e) \\
      \tval_{21} = {} & \tval_{21}' \quad (d) & \tval_{o2} = {} &
\tval_{o2} \tvjoin[] \tval_{o2}' \quad (f)
    \end{align*}
    Thus we obtain from $(a),(c),(d)$ that
    $\pair{\tval_{i2}}{\tval_{i1}} = {} \tval_{i2} \tiprop \tval_{i1}$
    and, by induction hypothesis, can conclude
    $\tval_{i2} \subtype \tval_{i1}$.

    From $(e)$ we can conclude $\tval_{o1}' \tvord[] \tval_{o1}$ and
    by monotonicity of strengthening
    (Lemma~\ref{lem:strengthening}\ref{lem:strengthening-monotone})
    we obtain
    $\tupdate{\tval_{o1}'}{x}{\tval_{i2}} \tvord[]
    \tupdate{\tval_{o1}}{x}{\tval_{i2}}$. Conversely, we have:
    \begin{align*}
      && \tupdate{\tval_{o1}}{x}{\tval_{i2}} \tvord[] {} & \tval_{o1}' &
\text{($\tiprop$ increasing and $(b)$)}\\
      &\Rightarrow& \tupdate{\tupdate{\tval_{o1}}{x}{\tval_{i2}}}{x}{\tval_{i2}}
\tvord[] {} & \tupdate{\tval_{o1}'}{x}{\tval_{i2}} &
 \text{(strengthening monotone, Lemma~\ref{lem:strengthening}\ref{lem:strengthening-monotone})}\\
      &\Rightarrow& \tupdate{\tval_{o1}}{x}{\tval_{i2}}
\tvord[] {} & \tupdate{\tval_{o1}'}{x}{\tval_{i2}} &
\text{(strengthening idempotent, Lemma~\ref{lem:strengthening}\ref{lem:strengthening-idempotent})}
    \end{align*}
    Thus, we have $\tupdate{\tval_{o1}}{x}{\tval_{i2}} =
    \tupdate{\tval_{o1}'}{x}{\tval_{i2}}$. Using similar reasoning
    we can also conclude $\tupdate{\tval_{o2}}{x}{\tval_{i2}} =
    \tupdate{\tval_{o2}'}{x}{\tval_{i2}}$. From
    Lemma~\ref{lem:prop-strengthening} it also follows that
    $\tval_{o1}' = \tupdate{\tval_{o1}'}{x}{\tval_{i2}}$ and
    $\tval_{o2}' = \tupdate{\tval_{o2}'}{x}{\tval_{i2}}$. Hence, it
    follows from $(b)$ that
    \[\pair{\tupdate{\tval_{o1}}{x}{\tval_{i2}}}{\tupdate{\tval_{o2}}{x}{\tval_{i2}}}
      = {} \tupdate{\tval_{o1}}{x}{\tval_{i2}}  \tiprop
      \tupdate{\tval_{o2}}{x}{\tval_{i2}}\]
    Then, by induction hypothesis we obtain
    $\tupdate{\tval_{o1}}{x}{\tval_{i2}} \subtype
    \tupdate{\tval_{o2}}{x}{\tval_{i2}}$. It follows that $\tval_1
    \subtype \tval_2$.
\end{proof}

To prove the soundness of the typing rules, we first prove the
following theorem, which is slightly stronger than Theorem~\ref{thm:typing-soundness}.

\begin{theorem}
  \label{thm:typing-soundness-ind}
  Let $e$ be an expression, $\tenv$ a valid typing
  environment, $\pstack$ an abstract stack, and $\tval \in \tvalues[]$. If
  $\tenv,\pstack \typrel e: \tval$, then for all $\tmap,\penv$ such
  that $\tmap$ is safe and
  $\tenv = \tmap \circ \penv$ there exists $\tmap_1$ such that
  $\tenv = \tmap_1 \circ \penv$ and 
  $\pair{\tmap_1}{\tval} =
  \ttrans{e}(\penv,\pstack)(\tmap_1)$. Moreover, $\tmap_1$ is safe.
\end{theorem}

\begin{proof}
  Assume $\tenv \typrel e_\elabel: \tval$ and let $\tmap,\penv$ such
  that $\tmap$ is safe and $\tenv = \tmap \circ \penv$. Let further $\pnode = \cn{\elabel}{\penv}$.
  The proof goes by structural induction over $e_\elabel$.

  \paragraph{Case $e = c$.}
  Let $\tmap_1=\tmap[\pnode \mapsto \tval]$, then
  $\tenv = \tmap_1 \circ \penv$.  By definition of the typing relation
  we must have $\rstrengthen{\tconst}{\tenv} \subtype
  \tval$. Lemma~\ref{lem:subtyping-prop-fixpoint} thus implies
  $\pair{\rstrengthen{\tconst}{\tenv}}{\tval} =
  \rstrengthen{\tconst}{\tenv} \tiprop \tval$ and $\tval$ is
  safe. Moreover, the fact that $\tenv$ is valid
  and the definition of $\tconst$ and
  strengthening imply that
  $\rstrengthen{\tconst}{\tenv} \in \bdomain_{\pscope_{\pnode}}$. By
  definition of $\tiprop$ we then must have
  $\tval = \rstrengthen{\tconst}{\tenv} \tvjoin[] \tval$. Hence, by
  definition of $\ttransname$ we directly obtain
  $\ttrans{e_\elabel}(\penv,\pstack)(\tmap_1)=\pair{\tmap_1}{\tval}$. Since
  $\tmap$ and $\tval$ are safe, so is $\tmap_1$.

  \paragraph{Case $e = x$.} Let $\tmap_1=\tmap[\pnode \mapsto \tval]$,
  then $\tenv = \tmap_1 \circ \penv$. By definition of the typing
  relation we must have
  $\rstrengthen{\tenv(x)\tvareq{x}}{\tenv} \subtype
  \rstrengthen{\tval\tvareq{x}}{\tenv}$. Lemma~\ref{lem:subtyping-prop-fixpoint} thus implies
  \[\pair{\rstrengthen{\tenv(x)\tvareq{x}}{\tenv}}{\rstrengthen{\tval\tvareq{x}}{\tenv}} =
  \rstrengthen{\tenv(x)\tvareq{x}}{\tenv} \tiprop
  \rstrengthen{\tval\tvareq{x}}{\tenv}\] and that 
  $\tval$ is safe. Let
  $\ttrans{e}(\penv,\pstack)(\tmap_1) = \pair{\tmap_2}{\tval_2}$. We
  know from the definition of $\ttransname$ that we have
  $\tmap_2 = \tmap_1[\pnode \mapsto \tval', \pnode_x \mapsto
  \tval_x']$ where 
  $\tval' = \tmap_1(\pnode) \tvjoin[] \rstrengthen{\tval\tvareq{x}}{\tenv}$ and
$\tval_x' = \tmap_1(\pnode_x) \tvjoin[] \rstrengthen{\tenv(x)\tvareq{x}}{\tenv}]$ for
  $\pnode_x=\penv(x)$. By definition of $\tmap_1$ we have
  $\tmap_1(\pnode)=\tval$ and moreover $\tenv = \tmap_1 \circ \penv$ implies
  $\tmap_1(\pnode_x) = \tenv(x)$. Since
  strengthening is reductive, we further know
  $\rstrengthen{\tval\tvareq{x}}{\tenv} \tvord[] \tval$ and
  $\rstrengthen{\tenv(x)\tvareq{x}}{\tenv} \tvord[] \tenv(x)$. Thus,
  we have
  $\tval' = \tenv(x) \tvjoin[]
  \rstrengthen{\tval\tvareq{x}}{\tenv} = \tval'$ and
  $\tval_x' = \tenv(x) \tvjoin[]
  \rstrengthen{\tenv(x)\tvareq{x}}{\tenv} = \tenv(x)$. Hence, we can conclude
  $\tmap_2 = \tmap_1$.
  
  \paragraph{Case $e = e_\elabela \, e_\elabelb$.}
  By definition of the typing relation we must have $\tenv,\pstack
  \typrel e_\elabela : \tval_\elabela$ and $\tenv,\pstack
  \typrel e_\elabelb : \tval_\elabelb$ and $\tval_\elabela \subtype
  x:[\elabela \pconcat \pstack: \tval_\elabelb \to \tval]$ for some
  $\tval_\elabela,\tval_\elabelb$.
  
  By induction hypothesis we conclude that there exist
  $\tmap_\elabela$ and $\tmap_\elabelb$ such that
  $\pair{\tmap_\elabela}{\tval_\elabela} =
  \ttrans{e_\elabela}(\penv,\pstack)(\tmap_\elabela)$,
  $\pair{\tmap_\elabelb}{\tval_\elabelb} =
  \ttrans{e_\elabelb}(\penv,\pstack)(\tmap_\elabelb)$ and
  $\tmap_\elabela \circ \penv = \tmap_\elabelb \circ \penv =
  \tenv$. Moreover, $\tmap_\elabela$, $\tmap_\elabelb$ are
  safe. Define $\tmap_1$ as follows:
  \[\tmap_1 = \Lambda \pnode_1.\,
    \begin{cases}
      \tmap_\elabela(\pnode_1) & \text{if } \nloc(\pnode_1) \in e_\elabela\\
      \tmap_\elabelb(\pnode_1) & \text{if } \nloc(\pnode_1) \in
      e_\elabelb\\
      \tval & \text{if } \pnode_1 = \pnode\\
      \tmap(\pnode_1) & \text{otherwise}
    \end{cases}
  \]
  Note that we have
  $\pair{\tmap_1}{\tval_\elabela} =
  \ttrans{e_\elabela}(\penv,\pstack)(\tmap_1)$,
  $\pair{\tmap_1}{\tval_\elabelb} =
  \ttrans{e_\elabelb}(\penv,\pstack)(\tmap_1)$ and
  $\tmap_1 \circ \penv = \tenv$. Also, $\tmap_1$ is safe. By
  Lemma~\ref{lem:subtyping-prop-fixpoint} we further know
  $\pair{\tval_\elabela}{ x:[\elabela \pconcat \pstack: \tval_\elabelb
    \to \tval]} = \tval_\elabela \tiprop x:[\elabela \pconcat \pstack:
  \tval_\elabelb \to \tval]$. Thus, we have
  $\ttrans{e}(\penv,\pstack)(\tmap_1) = \pair{\tmap_1}{\tval}$.

  \paragraph{Case $e = \lambda x. e_\elabela$.}
  Let
  $\pstack' \in \tval$.  By definition of the typing relation we must
  have $\tenv_\elabela,\pstack' \typrel e_\elabela : \tval_\elabela$
  and
  $x\!:[\pstack':\tval_x \to \tval_\elabela] \subtype \tval(\pstack')$
  where $\tenv_\elabela = \tenv.x\!:\tval_x$ for some
  $\tval_x,\tval_\elabela$.
   
  First, by Lemma~\ref{lem:subtyping-prop-fixpoint} we know that
  $\pair{x\!:[\pstack':\tval_x \to \tval_\elabela]}{\tval(\pstack')} =
  x\!:[\pstack':\tval_x \to \tval_\elabela] \tiprop \tval(\pstack')$
  and that $\tval(\pstack),\tval_x,\tval_\elabela$ are all safe.

  Define
  $\tmap_{\pstack'} = \tmap[\pnode_{\pstack',x} \mapsto \tval_x]$ and
  $\penv_{\pstack',\elabela} = \penv[x \mapsto \pnode_{\pstack',x}]$
  where $\pnode_{\pstack',x} = \cns{x}{\penv}{\pstack'}$. Then
  $\tenv_\elabel = \tmap_{\pstack'} \circ
  \penv_{\pstack',\elabela}$. Because $\pstack' \in t$ and
  $x\!:[\pstack':\tval_x \to \tval_\elabela] \subtype \tval(\pstack')$
  we know that $\tval_x \neq \tbot$. If follows that $\tenv_\elabela$
  is valid. Hence, applying the induction hypothesis, it follows that
  there exists $\tmap_{\pstack',1}$ such that
  $\ttrans{e_\elabela}(\penv_{\pstack',\elabela},\pstack')(\tmap_{\pstack',1})
  = \pair{\tmap_{\pstack',1}}{\tval_\elabela}$. Altogether, we can
  therefore conclude
  \[\tstepbody(x, e_\elabela, \pnode, \penv,
    \tval)(\pstack')(\tmap_{\pstack',1}) = \pair{\tmap_{\pstack',1}}{\tval(\pstack')}\enspace.\]
  Now define
  \[\tmap_1 = \Lambda \pnode_1.\, \begin{cases}
      \tval & \text{if } \pnode = \pnode_1\\
      \tmap_{\pstack',1}(\pnode_1) & \text{if }
      \pnode_{\pstack',x} \in \pnode_1\\
      \tmap(\pnode_1) & \text{otherwise}
    \end{cases}\] Note that we now have for all
  $\pstack' \in \pstacks$,
  $\tstepbody(x, e_\elabela, \pnode, \penv, \tval)(\pstack')(\tmap_1)
  = \pair{\tmap_1}{\tval(\pstack')}$. Moreover, $\tmap_1$ is safe and
  satisfies $\tenv = \tmap_1 \circ \penv$. Finally, note that because
  $\tval \in \ttables[]$, we have
  $\tmap_1(\pnode) = \tval = \tval \tvjoin[] x:\ttable_\bot =
  \tvjoinb_{\pstack' \in \tval} \tval(\pstack')$. Hence, it follows
  that $\ttrans{e}(\penv,\pstack)(\tmap_1) = \pair{\tmap_1}{\tval}$
  holds.
\end{proof}

\begin{proof}[Proof of Theorem~\ref{thm:typing-soundness}]
The theorem follows directly from Theorem~\ref{thm:typing-soundness-ind} by
letting $\cenv(x) = \cvn{x}{\epsilon}{\epsilon}$ and
\[\tmap(n) = \begin{cases}
    \gamma(x) & \text{if } n = \cenv(x)\\
    \tbot & \text{otherwise}
  \end{cases}
\]
\end{proof}

\begin{proof}[Proof of Theorem~\ref{thm:typing-completeness}]
  Assume $\ttrans{e_\elabel}(\penv,\pstack)(\tmap)=\pair{\tmap}{\tval}$ and
  $\tmap$ is safe. The proof goes by structural induction over
  $e$. Let $\pnode = \cn{\elabel}{\penv}$ and $\Gamma = \tmap \circ \penv$.

  \paragraph{Case $e = c$.}
  By definition of $\ttransname$, we have $\tmap(\pnode) = \tval =
  \tval \tvjoin[] \rstrengthen{\tconst}{\tenv}$. Moreover, since
  $\rstrengthen{\tconst}{\tenv} \in \bdomain$ and since $\tval \neq \ttop$ we
  have by definition of $\tiprop$ that
  $\pair{\rstrengthen{\tconst}{\tenv}}{\tval} =
  \rstrengthen{\tconst}{\tenv} \tiprop \tval$ holds. It follows from
  Lemma~\ref{lem:subtyping-prop-fixpoint}, that
  $\rstrengthen{\tconst}{\tenv} \subtype \tval$ must also hold. By the
  definition of the typing relation we then obtain $\tenv,\pstack \typrel e : \tval$.

  \paragraph{Case $e = x$.}
  Let $\pnode_x = \penv(x)$.
  By definition of $\ttransname$, we have $\pair{\tval_x}{\tval'} =
  \rstrengthen{\tenv(x)\tvareq{x}}{\tenv} \tiprop \rstrengthen{\tenv(x)\tvareq{x}}{\tenv}$ where
  $\tval_x \tvord[] \tenv(x)$ and $\tval' \tvord[] \tval$. Hence by
  monotonicity of strengthening we obtain
  $\rstrengthen{\tval_x\tvareq{x}}{\tenv} \tvord[]
  \rstrengthen{\tenv(x)\tvareq{x}}{\tenv}$ and
  $\rstrengthen{\tval'\tvareq{x}}{\tenv} \tvord[]
  \rstrengthen{\tval\tvareq{x}}{\tenv}$. 

  Because $\tiprop$ is increasing, we further know
  $\rstrengthen{\tenv(x)\tvareq{x}}{\tenv} \tvord[] \tval_x$ and
  $\rstrengthen{\tval\tvareq{x}}{\tenv} \tvord[] \tval'$. Applying
  monotonicity and idempotence of strengthening again we can conclude
  $\rstrengthen{\tenv(x)\tvareq{x}}{\tenv} \tvord[]
  \rstrengthen{\tval_x\tvareq{x}}{\tenv}$ and
  $\rstrengthen{\tval\tvareq{x}}{\tenv} \tvord[]
  \rstrengthen{\tval'\tvareq{x}}{\tenv}$. Hence, we obtain in fact
  $\rstrengthen{\tenv(x)\tvareq{x}}{\tenv} =
  \rstrengthen{\tval_x\tvareq{x}}{\tenv}$ and
  $\rstrengthen{\tval\tvareq{x}}{\tenv} =
  \rstrengthen{\tval'\tvareq{x}}{\tenv}$. By
  Lemma~\ref{lem:prop-strengthening} it further follows that
  $\tval_x = \rstrengthen{\tval_x\tvareq{x}}{\tenv}$ and
  $\tval' = \rstrengthen{\tval'\tvareq{x}}{\tenv}$. Hence, we can
  apply Lemma~\ref{lem:subtyping-prop-fixpoint} to conclude
  $\rstrengthen{\tenv(x)\tvareq{x}}{\tenv} \subtype
  \rstrengthen{\tval\tvareq{x}}{\tenv}$. Finally, using the typing
  rule for variables, it follows that $\tenv,\pstack \typrel e :
  \tval$ holds.

  \paragraph{Case $e = e_\elabela \, e_\elabelb$.}
  By the definition of $\ttransname$ and the fact that $\ttransname$ is
  increasing, we must have
  $\ttrans{e_\elabela}(\penv,\pstack)(\tmap) =
  \pair{\tmap}{\tval_\elabela}$ and $\ttrans{e_\elabelb}(\penv,\pstack)(\tmap) =
  \pair{\tmap}{\tval_\elabelb}$. Hence, by induction hypothesis, we
  conclude $\tenv,\pstack \typrel e_\elabela : \tval_\elabela$ and
  $\tenv,\pstack \typrel e_\elabelb : \tval_\elabelb$. Furthermore, we
  know that
  \[\tval'_\elabela,\, \absn:[\tentry{\elabela \pconcat \pstack}{\tval_\elabelb'}{\tval'}] = \tval_\elabela \tiprop
\absn:[\tentry{\elabela \pconcat \pstack}{\tval_\elabelb}{\tval}]\]
  for some $\tval_\elabela',\tval_\elabelb',\tval'$ such that
  $\tval_\elabela' \tvord[] \tval_\elabela$, $\tval_\elabelb' \tvord[]
  \tval_\elabelb$ and $\tval' \tvord[] \tval$. Since $\tiprop$ is
  increasing, we conclude
  \[\tval_\elabela,\, \absn:[\tentry{\elabela \pconcat \pstack}{\tval_\elabelb}{\tval}] = \tval_\elabela \tiprop
    \absn:[\tentry{\elabela \pconcat \pstack}{\tval_\elabelb}{\tval}]\]
  Moreover, since $\tmap$ is safe, so must be $\tval_\elabela$,
  $\tval_\elabelb$, and $\tval$. Applying
  Lemma~\ref{lem:subtyping-prop-fixpoint} we therefore obtain $\tval_\elabela \subtype
    \absn:[\tentry{\elabela \pconcat \pstack}{\tval_\elabelb}{\tval}]$
    and, hence, $\tenv,\pstack \typrel e : \tval$ using the typing
    rule for function application.

    \paragraph{Case $e = \lambda x.\,e_\elabela$.}
    By the definition of $\ttransname$, we must have $\tmap(\pnode) =
    \tval = \tval \tvjoin[] \ttable_\bot$. Since $\tval \neq \ttop$ it
    follows that $\tval \in \ttables[]$.

    Now let $\pstack \in \tval$. Since $\tstepbody$ is increasing, it
    follows that we must have
    \[\tstepbody(x, e_\elabela, \pnode, \penv, \tval)(\pstack')(\tmap)
      = \pair{\tmap}{\tval(\pstack')}\]
    
    Using similar reasoning again for the definition of $\tstepbody$,
    we conclude that we must have
    $\ttrans{e_\elabela}(\penv_\elabela, \pstack')(\tmap) =
    \pair{\tmap}{\tval_\elabela}$ where
    $\tval_\elabela = \tmap(\pnode_\elabela)$,
    $\pnode_\elabela = \cn{\elabela}{\penv_\elabela'}$,
    $\penv_\elabela, = \penv.x\!:\!\pnode_x$, and
    $\pnode_x = \cns{x}{\penv}{\pstack'}$. Let further $\tval_x =
    \tmap(\pnode_x)$ and $\tenv_\elabela = \tenv.x\!:\tval_x$. It
    follows that $\tenv_\elabela = \tmap \circ \penv_\elabela$. Hence,
    the induction hypothesis entails that we must have
    $\tenv_\elabela,\pstacks' \typrel e_\elabela : \tval_\elabela$.

    Finally, we know
    $x:[\ttab[\pstack']{\tval_x'}{\tval_\elabela'}], \tval' =
    x:[\ttab[\pstack']{\tval_x}{\tval_\elabela}] \tiprop
    \tval(\pstack')$ for some $\tval_x',\tval_\elabela',\tval'$ such
    that $\tval_x' \tvord[] \tval_x$,
    $\tval_\elabela' \tvord[] \tval_\elabela$, and
    $\tval' \tvord \tval(\pstacks')$. From the fact that $\tiprop$ is
    increasing it further thus follows that $\tval_x' = \tval_x$,
    $\tval_\elabela' = \tval_\elabela$, and
    $\tval' = \tval(\pstacks')$. We can therefore conclude using
    Lemma~\ref{lem:subtyping-prop-fixpoint} that
    $x:[\ttab[\pstack']{\tval_x}{\tval_\elabela}] \subtype
    \tval(\pstack')$. Using the typing rule for lambda abstractions
    it follows that $\tenv,\pstack \typrel e : \tval$.    
\end{proof}


}

\end{document}